%% file: main.tex
\documentclass[11pt]{article}

\usepackage[T1]{fontenc}
\usepackage{braket}
\usepackage[sfdefault]{biolinum} 
\usepackage{amsmath}
\usepackage{graphicx}
\usepackage[colorlinks=true, allcolors=blue]{hyperref}
\input{commands}

\title{Commitments from Quantum One-Wayness}

\author{Dakshita Khurana\thanks{UIUC. \texttt{\{dakshita,ktomer2\}@illinois.edu}}
\and
Kabir Tomer\footnotemark[1]
}

\date{}
\begin{document}

\maketitle

\thispagestyle{empty}
\input{abstract}
\newpage

\thispagestyle{empty}
\tableofcontents
\newpage

\pagenumbering{arabic}
\input{intro-new}

\input{overview-new}
\input{prelim}
\input{owsg-imply-efi}

\input{combiners}

\bibliographystyle{alpha}
\bibliography{abbrev0, crypto,bib,custom}

\appendix
\input{owpuzzle-from-crypto}
\input{fullproof}
\end{document}

%% file: commands.tex
\newcommand{\pred}[1]{\ensuremath{\mathsf{#1}}}

\usepackage{fullpage,amsfonts,latexsym,amsthm,graphicx,amssymb,amsmath,enumitem,mathrsfs}
\usepackage{xcolor}
\definecolor{almostblack}{rgb}{0.6, 0.0, 0.0}

\definecolor{almostblackk}{rgb}{0.3, 0, 0.0}

\usepackage{hyperref,cleveref}
\hypersetup{
    colorlinks=true,
    citecolor=almostblack,
    linkcolor=almostblack,
    filecolor=almostblack,      
    urlcolor=almostblack,
}

\mathchardef\mhyphen="2D
\usepackage{mathtools}
\DeclarePairedDelimiter{\ceil}{\lceil}{\rceil}

\newtheorem{theorem}{Theorem}[section]
\newtheorem{claim}{Claim}[section]
\newtheorem{subclaim}{SubClaim}[section]

\newtheorem{definition}{Definition}[section]
\newtheorem{corollary}{Corollary}[section]
\newtheorem{lemma}{Lemma}[section]

\newcommand{\defref}[1]{Definition~\protect\ref{def:#1}}

\newlength{\protowidth}

\definecolor{blue(ncs)}{rgb}{0.0, 0.53, 0.74}

\newcommand{\Com}{\mathsf{Com}}

\usepackage{amssymb}
\usepackage{amsmath}
\usepackage{amsthm}
\usepackage{amsfonts}
\usepackage{xspace}
\usepackage{bm}
\usepackage{bbm}
\usepackage{boxedminipage}
\usepackage{xparse}
\usepackage{float}
\usepackage{multirow}
\usepackage{graphicx}
\usepackage{caption}
\usepackage{subcaption}

\usepackage{tikz}
\usetikzlibrary{positioning}
\usetikzlibrary{fit}
\usetikzlibrary{calc}
\usetikzlibrary{backgrounds}
\usetikzlibrary{shapes}
\usetikzlibrary{patterns}
\usetikzlibrary{matrix}

\usepackage{enumitem}

\usepackage{geometry}
\newcommand{\sdotfill}{\textcolor[rgb]{0.2,0.2,0.2}{\dotfill}} 

\newcommand{\namedref}[2]{\hyperref[#2]{#1~\ref*{#2}}\xspace}

\def\poly{\ensuremath{\mathsf{poly}}\xspace}
\def\negl{\ensuremath{\mathsf{negl}}\xspace}

\newcommand{\sR}{\ensuremath{\mathsf{R}}}

\newcommand{\xor}[0]{\ensuremath{\oplus}\xspace}

  \newcommand{\cA}{\ensuremath{{\mathcal A}}\xspace}
  \newcommand{\cB}{\ensuremath{{\mathcal B}}\xspace}
  
  \newcommand{\cD}{\ensuremath{{\mathcal D}}\xspace}
  \newcommand{\cE}{\ensuremath{{\mathcal E}}\xspace}
  
  \newcommand{\cG}{\ensuremath{{\mathcal G}}\xspace}

  \newcommand{\cL}{\ensuremath{{\mathcal L}}\xspace}

  \newcommand{\cR}{\ensuremath{{\mathcal R}}\xspace}
  
  \newcommand{\cT}{\ensuremath{{\mathcal T}}\xspace}
  \newcommand{\cU}{\ensuremath{{\mathcal U}}\xspace}

  \newcommand{\cX}{\ensuremath{{\mathcal X}}\xspace}

  \newcommand{\bbA}{\ensuremath{{\mathbb A}}\xspace}
  
  \newcommand{\bbC}{\ensuremath{{\mathbb C}}\xspace}
  
  \newcommand{\bbE}{\ensuremath{{\mathbb E}}\xspace}
  \newcommand{\bbF}{\ensuremath{{\mathbb F}}\xspace}
  \newcommand{\bbG}{\ensuremath{{\mathbb G}}\xspace}
  
  \newcommand{\bbI}{\ensuremath{{\mathbb I}}\xspace}
  
  \newcommand{\bbK}{\ensuremath{{\mathbb K}}\xspace}
  \newcommand{\bbM}{\ensuremath{{\mathbb M
  }}\xspace}
  \newcommand{\bbN}{\ensuremath{{\mathbb N}}\xspace}

  \newcommand{\bbS}{\ensuremath{{\mathbb S}}\xspace}

  \newcommand{\supp}[0]{\ensuremath{\mathsf{Supp}}\xspace}
  
  \renewcommand{\Pr}[0]{\mathrm{Pr}\xspace}

\newcommand{\kabir}[1]{{{\textcolor{red}{}}}}
\newcommand{\dakshita}[1]{{{\textcolor{blue}{}}}}

\newcommand{\bin}{\{0,1\}}

\newcommand{\sample}{\leftarrow}

\definecolor{lg}{gray}{0.89}

\DeclareMathOperator*{\Prr}{Pr}

\newcommand{\KeyGen}{\ensuremath{\pred{KeyGen}}\xspace}
\newcommand{\StateGen}{\ensuremath{\pred{StateGen}}\xspace}

\newcommand{\Ver}{\ensuremath{\pred{Ver}}\xspace}
\newcommand{\TD}{\ensuremath{\pred{TD}}\xspace}
\newcommand{\SD}{\ensuremath{\pred{SD}}\xspace}

\newcommand{\n}{\ensuremath{n}\xspace}

\newcommand{\x}{\ensuremath{k}\xspace}

\newcommand{\EFI}{\ensuremath{\pred{EFI}}\xspace}

\newcommand{\q}{q}

\newcommand{\sA}{\ensuremath{\mathsf{A}}\xspace}
\newcommand{\sB}{\ensuremath{\mathsf{B}}\xspace}
\newcommand{\sC}{\ensuremath{\mathsf{C}}\xspace}

\newcommand{\sS}{\ensuremath{\mathsf{S}}\xspace}
\newcommand{\sV}{\ensuremath{\mathsf{V}}\xspace}
\newcommand{\sX}{\ensuremath{\mathsf{X}}\xspace}
\newcommand{\sY}{\ensuremath{\mathsf{Y}}\xspace}

\newcommand{\Supp}{\ensuremath{\mathsf{Supp}}\xspace}

\newcommand{\tr}{\mathsf{Tr}}

\newcommand{\wgen}{\ensuremath{\mathsf{G_0}}\xspace}
\newcommand{\simwgen}{\ensuremath{\mathsf{G_1}}\xspace}
\newcommand{\sgen}{\overline{\mathsf{G}}_0}
\newcommand{\simsgen}{\overline{\mathsf{G}}_1}
\newcommand{\EFIGen}{\mathsf{EFIGen}}

\newcommand{\Hs}{\ensuremath{\mathsf{H}}}

\newcommand{\Imbalanced}{\text{Imbalanced}\xspace}
\newcommand{\imbalanced}{\text{imbalanced}\xspace}

\newcommand{\prm}{s}
\newcommand{\prmlimit}{s^*}

\newcommand{\Samp}{\mathsf{Samp}}
\newcommand{\shadow}{\mathsf{ShadowGen}}

\newcommand{\len}{L}
\newcommand{\Enc}{\mathsf{Enc}}
\newcommand{\Dec}{\mathsf{Dec}}

%% file: abstract.tex
\begin{abstract}
One-way functions are central to classical cryptography. They are necessary for the existence of non-trivial classical cryptosystems, and also sufficient to realize meaningful primitives including commitments, pseudorandom generators and digital signatures.
At the same time, a mounting body of evidence suggests that assumptions {\em even weaker} than one-way functions may suffice for many cryptographic tasks of interest in a quantum world, including bit commitments and secure multi-party computation. 

This work studies one-way state generators  [Morimae-Yamakawa, CRYPTO 2022], a natural quantum relaxation of one-way functions. Given a secret key, a one-way state generator outputs a hard to invert quantum state. 
A fundamental question is whether this type of {\em quantum} one-wayness suffices to realize quantum cryptography.
We obtain an affirmative answer to this question, by proving that one-way state generators with pure state outputs imply quantum bit commitments and secure multiparty computation.

Along the way, we use efficient shadow tomography [Huang et. al., Nature Physics 2020] to build an intermediate primitive with classical outputs, which we call a (quantum) one-way puzzle. 
Our main technical contribution is a proof that one-way puzzles imply quantum bit commitments. This proof develops new techniques for pseudoentropy generation [Hastad et. al., SICOMP 1999] from arbitrary distributions, which may be of independent interest.
\end{abstract}

%% file: intro-new.tex
\section{Introduction}

A one-way function is a classically efficiently computable function that is hard to invert.
This is a fundamental hardness assumption, necessary for the existence of much of modern classical cryptography~\cite{STOC:LubRac86,FOCS:ImpLub89,STOC:ImpLevLub89}. The classical crypto-complexity class ``minicrypt'' contains primitives like bit commitments, pseudorandom generators, pseudorandom functions and symmetric encryption, that are all equivalent to the existence of one-way functions. 
On the other hand, there are tasks like key exchange and secure multi-party computation that classically require stronger, more structured assumptions~\cite{C:ImpRud88}.

The relationship between computational hardness and cryptography appears to be drastically different in a quantum world.
Here, the seminal works of Wiesner~\cite{Wiesner83} and Bennett and Brassard~\cite{BenBra84} first demonstrated the possibility of unconditional quantum key distribution (QKD) by exploiting the properties of quantum information.
Unfortunately, it was also shown that other useful cryptographic primitives like bit commitments and secure computation {\em cannot} exist unconditionally~\cite{LoChau97,Mayers97}, and must necessarily rely on computational hardness, even in a quantum world.
However, our understanding of computational hardness in a quantum world is still in its infancy. 
For instance, it was only recently understood~\cite{C:BCKM21b,EC:GLSV21} that one-way functions suffice to enable secure multi-party computation in a quantum world, a task that is believed to be impossible classically.\\ 

\vspace{-0.9mm}
\noindent{\bf Sources of Hardness in a Quantum World.}
Despite being necessary for classical cryptography, one-way functions may not be necessary for computational quantum cryptography.  

Two recent concurrent works~\cite{C:AnaQiaYue22,C:MorYam22} demonstrated that 
many cryptographic primitives including quantum bit commitments, (one-time secure) digital signatures, and multi-party secure computation can also be based on the existence of {\em pseudorandom state generators (PRSGs)}, which were 
introduced in~\cite{TCC:JiLiuSon18}.

Given a secret key, a PRSG efficiently generates a quantum state, several copies of which are computationally indistinguishable from equally many copies of a Haar random state.
There is some evidence that points to PRSGs being a weaker assumption than one-way functions.
Specifically, PRSGs can exist even if BQP = QMA (relative to a quantum oracle)~\cite{Kretschmer21} or if P = NP (relative to a classical oracle)~\cite{KQST}.
This indicates that PRSGs, and all the cryptographic primitives that they imply, can exist even if all quantum-secure (classical) cryptographic primitives, including one-way functions, are broken.

Can we base quantum cryptography on assumptions that are potentially even weaker than the existence of PRSGs?
As pointed out in~\cite{morimaeOneWaynessQuantumCryptography2022}, PRSGs and bit commitments are ``decision-type'' primitives that rely on the hardness of distinguishing pseudorandom states from truly (Haar) random ones.
On the other hand, there is a natural, simpler ``search-type'' assumption that significantly relaxes the pseudorandomness guarantee of a PRSG to one-wayness. 
\\

\vspace{-0.9mm}
\noindent {\bf A one-way state generator (OWSG)}~\cite{C:MorYam22} is a quantum algorithm that given a secret key, generates a hard-to-invert quantum state. 
This is a natural quantum analogue of a one-way function, and appears to be weaker, as a definition, than most known quantum cryptographic primitives. It is in particular known to be implied by several quantum cryptographic primitives including quantum signature and encryption schemes, as well as quantum money\footnote{See, for example,~\cite{morimaeOneWaynessQuantumCryptography2022}, the figure at \url{https://sattath.github.io/qcrypto-graph/} and references therein.}.

Given 
that one-way functions enable a variety of classical cryptosystems,
it is natural to ask whether one-way state generators play a similar role in quantum cryptography. Namely,
\begin{center}
    {\em Can we obtain quantum cryptosystems including bit commitments and MPC\\ only assuming the existence of one-way state generators?}
\end{center}
Our main theorem answers this question in the affirmative in the setting where OWSG outputs are pure states. 

\vspace{-5.5mm}

\textcolor{almostblackk}{
\begin{theorem}(Informal)
\label{thm:owsg}
    One-way state generators with pure state outputs imply quantum bit commitments.
\end{theorem}
}

By combining with prior work that demonstrates conversions between various types of commitments~\cite{C:AnaQiaYue22,BCQ22}
and builds secure multi-party computation from commitments~\cite{EC:GLSV21,C:BCKM21b,C:AnaQiaYue22}, we also obtain the following corollary.

\vspace{-5.5mm}

\textcolor{almostblackk}{
\begin{corollary}(Informal)
\label{cor:owsgsc}
    One-way state generators with pure state outputs imply secure multi-party computation for all quantum functionalities.
\end{corollary}
}

We note that OWSGs were initially defined in~\cite{C:MorYam22} to only output pure states; but this definition was later generalized in~\cite{morimaeOneWaynessQuantumCryptography2022} to also allow mixed states. 
OWSG with pure state outputs were also studied in~\cite{CX22}, who also showed equivalences between variants (weak, distributional) of OWSGs.
Outputs of random quantum circuits yield natural candidates for pure OWSG that do not rely on classical hardness%
; in fact the output states can even be conjectured to be pseudorandom~\cite{C:AnaQiaYue22}. 
However, only relying on one-wayness introduces the possibility of building cryptography from other natural candidates: for instance, the (pre-measurement) states generated by BosonSampling experiments are not indistinguishable from Haar random~\cite{AA14}, but can plausibly be one-way.

Pure OWSGs are also implied by various cryptographic primitives such as digital signatures with pure verification keys and quantum money with pure banknotes~\cite{morimaeOneWaynessQuantumCryptography2022}. 
This, combined with our theorem, shows that these other primitives also imply quantum bit commitments. In some sense, this establishes commitments as the leading candidate for a minimal/necessary assumption in quantum cryptography.\\

\noindent{\bf One-Way Puzzles.}
Enroute to our main theorem, we 
use efficient shadow tomography~\cite{HKP20} to 
prove that OWSG imply an intermediate cryptographic primitive with entirely classical outputs, that we call a {\em one-way puzzle}.
We find this implication from a OWSG with arbitrary quantum outputs to a simple, cryptographically useful primitive with classical outputs, noteworthy.

\vspace{-5.5mm}

\textcolor{almostblackk}{
\begin{theorem}(Informal)
\label{thm:s}
    One-way state generators with pure state outputs imply one-way puzzles.
\end{theorem}
}

\noindent {\bf %
A (quantum) {one-way puzzle}} is a pair of algorithms $(\mathsf{Samp},\mathsf{Ver})$
where $\mathsf{Samp}$ is quantum polynomial time and outputs a pair of classical strings -- a key and puzzle $(k,s)$ -- satisfying $\Ver(k, s) = 1$. The security guarantee is that given a ``puzzle'' $s$, it is (quantum) computationally infeasible to find a key $k$ such that $\Ver(k,s) = 1$, except with negligible probability.

Unlike prior definitions of one-way puzzles in the literature, we do not require the verification ($\mathsf{Ver}$) algorithm to be efficient. As we will see later, only asking for inefficient verification turns out to be {\em necessary} for our implication from OWSG.
Indeed, if verification were efficient, then a $\mathsf{QMA}$ oracle would be capable of breaking one-way puzzles, but such an oracle is unlikely to break OWSG~\cite{Kretschmer21}.
Somewhat surprisingly, we show that inefficiently verifiable one-way puzzles are also {\em sufficient} to build quantum bit commitments.

The reader may have observed that one-way puzzles generalize one-way functions to allow joint, randomized sampling of keys and outputs.
In a classical world, this generalization is unnecessary: one-way puzzles are equivalent to one-way functions. 
One direction of the implication is straightforward, since one-way functions imply one-way puzzles (almost) immediately by definition.
In the other direction, a one-way function can be obtained from a classical one-way puzzle by ``pulling out'' the (uniform) randomness $r$ used by $\Samp$. The one-way function $f$ on input $r$ samples $(k,s) \leftarrow \mathsf{Samp}(1^n,r)$ and outputs $f(r) = s$. It is easy to see that one-wayness of the puzzle implies one-wayness of $f$.

However, the conversion above is no longer applicable when $\mathsf{Samp}$ is quantum, because there may be no equivalent deterministic, efficient function that on input uniform randomness, outputs $(k,s)$ distributed according to the output of $\mathsf{Samp}$.
Nevertheless, enroute to proving our main result, we show:

\vspace{-5.5mm}

\textcolor{almostblackk}{
\begin{theorem}(Informal)
\label{thm:owp}
    One-way puzzles imply quantum bit commitments.
\end{theorem}
}

Theorem \ref{thm:owp} is the most technically involved part of this work.
In a nutshell, existing techniques for building commitments from classical one-way primitives (e.g., ~\cite{HILL}) crucially only apply when the preimage distribution of every image of the function is flat (i.e., uniform over all preimages). This work develops a method to generate pseudorandomness from one-way puzzles with arbitrary preimage distributions, which we believe to be of independent interest. \\

\noindent{\bf Local/Hybrid Quantum Cryptography and One-Way Puzzles.}
As an aside, we observe that one-way puzzles are
also implied by quantum cryptography with classical communication.
In fact there is a large body of work that aims to understand the computational hardness yielding quantum cryptography with classical communication, including protocols for quantum advantage~\cite{Mahadev,MYadv,MYadv2}, quantum commitments with classical communication~\cite{TCCnew}, and even black-box separations for key exchange~\cite{C:ACCFLM22}. 
Classical communication protocols are desirable as they can be used over the current infrastructure (e.g., the Internet). In this model, sometimes called the ``local'' or ``hybrid'' or quantum-computation classical-communication (QCCC) model~\cite{C:ACCFLM22}, all the quantum computation is done locally by parties who exchange only classical messages.

We observe that natural cryptographic primitives such as public-key encryption and signatures in the QCCC 
model imply one-way puzzles.
For example, given a public-key encryption scheme, a one-way puzzle
 can be defined as follows.
The one-way puzzle sampler will output a puzzle consisting of a public key along with an encryption of a random message, and the corresponding solution will be the (plaintext) message.  
It is easy to see that an adversary that breaks one-wayness of the resulting puzzle can be used to break CPA security of the encryption scheme. 
In fact, one can obtain a one-way puzzle even given any public key encryption with classical public and secret keys, but quantum ciphertexts.
In Appendix \ref{app:implications}, we formalize these ideas and also show how similar ideas prove that one-way puzzles are implied by digital signatures, natural bit commitments and symmetric encryption schemes in the QCCC model\footnote{One may ask whether computational cryptographic primitives in the QCCC model also imply one-way functions. But it is unclear if this is true; and at the very least this is challenging to prove, for the same reason as above -- namely, we cannot explicitly pull out the sampling randomness from an arbitrary quantum algorithm.}.

Finally, we note that~\cite{TCC:KNY23} recently discussed a related but stronger primitive -- {\em hard quantum planted problems for NP languages} -- which is implied by cryptography with publicly verifiable deletion. A hard quantum planted problem for a language is specified by a QPT sampler that samples an instance-witness pair $(x, w)$ for the language in a way that no adversary can find a witness for $x$ with non-negligible probability. 
These are like one-way puzzles except that they admit efficient, deterministic verification.
By definition, hard quantum planted problems imply one-way puzzles (and therefore by our work, imply quantum bit commitments).\\

\noindent{\bf Conclusion and Future Directions.}
Prior to this work, bit commitments were known to be implied by pseudorandom state generators~\cite{C:AnaQiaYue22,C:MorYam22} via a construction that roughly parallels the classical setting~\cite{Naor89}. 
They were also known~\cite{morimaeOneWaynessQuantumCryptography2022} from a restricted type of OWSG; namely one with injective, orthogonal outputs. 
However, as we discuss in the next section, building commitments from general-purpose OWSG requires methods that are quite different from known classical techniques, and which may be broadly applicable beyond this work.

We also hope that the one-way puzzle abstraction will enable a better understanding of quantum bit commitments.
For example, some existing attempts to understand the complexity of quantum commitments~\cite{Kretschmer21} build oracles relative to which complexity classes collapse, but pseudorandom states exist (and thus, one-way puzzles exist). Directly establishing the existence of one-way puzzles relative to these oracles may be easier, and may enable even more general oracle separations. One-way puzzles may also help better understand the relationship between quantum cryptography and quantum notions of Kolmogorov complexity.

Finally, we discuss some open questions related to this work. An obvious one is whether our results extend to {\em mixed state} OWSG. One avenue towards proving this would be to build one-way puzzles from mixed-state OWSG, perhaps via better tomography. 
In addition, answering the following questions will shed some more light on the complexity of quantum cryptography.
\begin{enumerate}
 \item Can quantum bit commitments with {\em classical communication} be based on the existence of OWSG or one-way puzzles? 
 This is plausible because one-way puzzle outputs are classical after all. Moreover, many other intermediate primitives that we build in this work also have entirely classical outputs. 
    \item 
    En route to building commitments, this work constructs pseudo-{\em entropy} generators from one-way state generators. 
    Can other pseudorandom primitives, such as pseudorandom quantum states be obtained from OWSG or one-way puzzles? 
    Techniques in this work may serve as a useful starting point towards addressing this question.
    
    \item Do quantum bit commitments imply one-way puzzles? 
    If not, 
    is there a separation?
    It is easy to observe that one-way puzzles can be broken given (quantum) access to an oracle for a related boolean function $f$. Is this also true for every quantum bit commitment?
    This question appears to be connected with the unitary synthesis problem~\cite{aark,ab}, for which a recent work~\cite{LMW} gave a general one-query lower bound. 
    \item Is there a quantum analogue to the classical implication from one-way puzzles to one-way functions? In other words, does the existence of one-way puzzles with hardness over arbitrary distributions imply one-way primitives with hardness over uniform inputs?

\end{enumerate}

%% file: overview-new.tex
\section{Technical Overview}

We begin this overview by outlining a well-known construction of classical commitments from any injective one-way function. This construction relies on hardcore predicates: roughly, a hardcore predicate for a one-way function $f$ is a bit that is easy to compute given a preimage $\x$ but hard to compute given $f(\x)$. The Goldreich-Levin theorem~\cite{GoldreichLevin} shows that the bit $\langle \x, r \rangle$ is hard-core for the function $f(\x)||r$. When $f$ is injective, the hardcore bit is uniquely determined for every element in the image, and gives rise to a simple commitment scheme, as follows. 

A commitment to bit $b$ is $f(\x), r, \langle \x, r\rangle \oplus b$ for randomly sampled $\x$ and $r$. This commitment is binding because of the injectivity of $f$, and computationally hides the bit $b$ due to $\langle \x, r \rangle$ being hardcore. 
This construction does not work when $f$ is not injective. In this case, for an image $y$, there may exist two preimages $\x_1, \x_2 \in \{f^{-1}(y)\}$ such that $\langle \x_1, r \rangle \neq \langle \x_2, r \rangle$, which will allow the committer to break binding.

The celebrated work of Hastad et. al.~\cite{HILL} showed how to overcome the binding issue, and base {\em classical} commitments on {\em general}  (not necessarily injective) one-way functions.
We outline (some relevant parts of) their technique next.\\

\noindent{\bf Pairwise independent hashing reduces the number of preimages.}
The starting point of the technique in~\cite{HILL} is to append to the image $f(\x)$ a {\em pairwise-independent} hash $h(\x)$, thereby reducing the total number of preimages of $f(\x),h(\x)$. This makes $f(\x),h(\x)$ behave somewhat like an injective function for carefully chosen output sizes of $h(\x)$.

In more detail, let $N_{\x}$ denote the number of preimages of $f(\x)$.
When the output size $h(\x)$ is set to (slightly larger than) $\mathsf{log} N_{\x}$, then~\cite{HILL} (roughly) show that:
\begin{itemize}
\item $h({\x})$ is computationally indistinguishable from uniform given $f(\x)$, and 
\item $h({\x})$ is statistically (somewhat) distinguishable from uniform given $f(\x)$.
\end{itemize}
The fact that $f(\x), h({\x})$ {\em appears} to a computationally bounded adversary to have more entropy than it actually does is formalized 
by building an object called a {\em weak pseudoentropy generator (WPEG)}~\cite{HILL,hill-revisited}. We will now describe this object in some more detail.

\subsection{Weak Pseudoentropy Generators (WPEG)}
A distribution $\wgen$ 
is a weak pseudoentropy generator (WPEG) if there exists another, possibly inefficient {\em simulated} distribution $\simwgen$ whose output is computationally indistinguishable from, and yet 
has more Shannon entropy than $\wgen$.

For a one-way function $f$ and pairwise independent hash $h$, we can consider distributions
\begin{align*}
    \wgen(1^n) &:= ~~~~f(\x), h, i, h({\x})_i \text{~~~~~~~~~~~and~~~~}\\
    \simwgen(1^n) &:= 
    \left\{\begin{array}{ll}
     f(\x), h, i, h({\x})_{i-1}, u_1 & \text{ if }i = \ceil{\log N_{\x}}+1\\
     f(\x), h, i, h({\x})_i & \text{ otherwise}
    \end{array}
    \right.
\end{align*}
where $\x, h$ are sampled uniformly in $\{0,1\}^n$, $i \leftarrow [n]$, $N_{\x}$ denotes the number of preimages of $f(\x)$, $h(\x)_i$ denotes $h(\x)$ truncated to the first $i$ bits, and $u_1$ denotes a uniformly random bit.

Prior works~\cite{HILL,hill-revisited} show that the distributions $\wgen$ and $\simwgen$ are computationally indistinguishable, but $\simwgen$ has more entropy than $\wgen$. 
This can be understood as follows.
\begin{itemize}
\item {\bf Entropy Gap.} 
Roughly, the pairwise independence of $h$ implies that 
with probability at least $\frac{1}{2}$, $f(\x),h({\x})_{\ceil{\log N_{\x}}}$ has a single preimage, i.e., $x$. Thus, with probability at least $\frac{1}{2}$ the last bit in $\wgen$ is a deterministic function of the remaining bits, and has less entropy than the corresponding (uniform) bit in $\simwgen$. 

\item {\bf Computational Indistinguishability.} 
By the Leftover Hash Lemma, for $x$ sampled from any distribution $\cX$ with min-entropy $\ell$, the first $\ell - 2c \log n$ bits of $h({x})$ are $\frac{1}{n^c}$ {\em statistically} close to uniform, even given $h$. 
By setting $\cX$ to be the (uniform) distribution over preimages of $f(\x)$, this implies that the first $\ell - O(\log N_{\x})$ bits of $h(\x)$ are statistically close to uniform given $h$,  for $\ell = {\log N_\x}$. Then applying the Goldreich-Levin theorem while guessing the last $O(\log n)$ bits of $h(\x)$ converts a distinguisher between $\wgen$ and $\simwgen$ to an inverter for $f$\footnote{This step requires the hash to be a specific inner-product based function which is compatible with the Goldreich-Levin technique.}.

\end{itemize}
Next, we discuss barriers in extending these ideas to quantum one-way state generators.\\

\noindent{\bf \underline{A Preliminary Approach that Does Not Work.\\}}
\vspace{0.2mm}

A natural first approach to building commitments from OWSG could be to replace the classical string $f(\x)$ in the distributions above, with the quantum state $\ket{\psi_\x}$ output by the OWSG. 

Then the two WPEG distributions $\wgen$ and $\simwgen$ are replaced by the following mixed states. 
$$\rho_0(1^n):= \sum_{\x, h, i} \ket{\psi_\x}\bra{\psi_\x}, h, i, h({\x})_i$$ 
and
\begin{align*}
    \rho_1(1^n):= 
    \left\{\begin{array}{ll}
     \sum_{\x,h,i,u} \ket{\psi_\x}\bra{\psi_\x}, h, i, h({\x})_{i-1}, u_1 & \text{ if }i = \ceil{\log N_{\x}}+1\\
     \sum_{\x,h,i,u} \ket{\psi_\x}\bra{\psi_\x}, h, i, h({\x})_i & \text{ otherwise}
    \end{array}
    \right.
\end{align*}
Unfortunately, classical arguments demonstrating statistical entropy gap and computational indistinguishability do not extend to the mixed states above, because of barriers that we describe next.\\

\noindent{\bf \underline{Barrier 1: Non-orthogonality of Outputs, or, What is a Preimage Anyway?\\}}
\vspace{0.2mm}

The mixed state $\rho_1$ above is well-defined only when $N_{\x}$ is.
In the classical setting, $N_{\x}$ denotes the number of pre-images of $f(\x)$.
But it is unclear how to define ``preimages'' of a quantum state under a OWSG. For two keys $x$ and $x'$, the corresponding OWSG output states $\ket{\psi_{x}}$ and $\ket{\psi_{x'}}$ could have arbitrary overlap. What overlaps qualify $x'$ to be a pre-image of $\ket{\psi_{x}}$?
One could consider fixing some inverse polynomial function (say $\frac{1}{n}$) and say that $x'$ is a pre-image of $\ket{\psi_{x}}$ whenever $\langle{\psi_x}|{\psi_{x'}}\rangle \geq \frac{1}{n}$.
Unfortunately, setting an arbitrary threshold does not accurately capture the adversary's uncertainty about $k$, given $\ket{\psi_k}$.
In fact, such an approach is fundamentally doomed for the following reason.

It is possible to build one-way state generators that are {\em unconditionally} statistically uninvertible given only a single copy of the output state $\ket{\psi_\x}$. 
A simple example is the following construction based on Weisner encodings/BB84 states. 
On input classical key $\x = (\theta,x)$ where $\theta, x \leftarrow \{0,1\}^n$, the OWSG outputs pure state $\ket{x}_\theta$. This OWSG is {\em statistically} single-copy secure, because $\ket{x}_\theta$ hides the string $\theta$ (over the randomness of $x$). 

Since quantum bit commitments cannot be secure against unbounded adversaries, this would rule out any possible constructions of commitments (including the one above) that rely only on the existence of single-copy (pure) OWSG. 
Instead, we will crucially
rely on multi-copy security of the OWSG to obtain an intermediate primitive where for every pair of keys $(k_1, k_2)$, their images are either orthogonal or parallel.\\

\noindent{\bf \underline%
{Resolving Barrier 1: From Quantum to Classical Outputs via Shadow Tomography.}}
\vspace{0.2mm}

Shadow tomography, introduced in~\cite{Aar?}, allows one to estimate a large number of observables by obtaining classical information from relatively few copies of an unknown quantum state. 
In more detail, shadow tomography is a procedure that applied to $t = \poly(n, \frac{1}{\epsilon})$ copies of an unknown state $\ket{\psi}$ yields a {\em classical} string, the shadow $S$. 
Given $S$, it is possible to simultaneously estimate $\langle \psi|O_j|\psi \rangle$ upto $\epsilon$ error for an exponentially large number of observables $\{O_j\}_{j \in [2^n]}$.

Applying shadow tomography to a OWSG with pure outputs yields (at least) a statistical inverter for the OWSG. Given $t = \poly(n)$ copies of some state $\ket{\psi_\x}$, an inverter can use shadow tomography on $\ket{\psi_\x}^{\otimes t}$ to (inefficiently) find a $\x'$ such that  $\langle{\psi_\x}|{\psi_{\x'}}\rangle > 1 - \frac{1}{n}$. 

Given a OWSG output state $\ket{\psi_\x}$, it may even be tempting to define its ``preimages'' as 
the set of possible keys $\x'$ returned by this statistical inverter, and try to apply arguments similar to the classical argument above.
Unfortunately this approach breaks down too. 
The statistical inverter given $\ket{\psi_\x}^{\otimes t}$ only finds a key $\x'$ where $\ket{\psi_{\x'}}$ has nontrivial {\em single-copy} overlap with $\ket{\psi_\x}$.
It is possible that for such $\x'$, $\langle{\psi_{\x}}|{\psi_{\x'}}\rangle ^{\otimes t}$ is close to $0$. Thus $\x'$ is {\em not even close to being} a preimage of $\ket{\psi_\x}^{\otimes t}$, at least for the purposes of arguing computational indistinguishability.

Thus instead of trying to define preimages of quantum states, we will crucially use the fact that certain shadow tomography methods~\cite{HKP20} have efficiently computable classical shadows. We now outline how this fact turns out to be useful.\\

\noindent{\bf Our Main Insight} is the following:
On input key $\x$, instead of having $\rho_0$ (and $\rho_1$) contain one or more copies of the OWSG state $\ket{\psi_{\x}}$, they will only contain a classical shadow $S_\x$ of $\ket{\psi_{\x}}$. 
Whenever these classical shadows can be efficiently computed, the WPEG distribution $\rho_0$ remains efficiently sampleable, and even becomes entirely classical!

While OWSG are defined to be secure given an arbitrary (unbounded) polynomial copies of $\ket{\psi_k}$, computing the shadow $S_k$ will require only a fixed linear number of copies. Indeed, our proof shows that commitments are implied by a weaker variant of OWSG, where security only holds given a fixed linear number of copies of $\ket{\psi_k}$.

We point out that the shadow $S_\x$ is a {\em randomized} (i.e., not deterministic) function of the key $\x$.
Moreover, given a shadow $S_\x$ obtained from $\ket{\psi_\x}$, it is computationally infeasible to find {\em any} key $\x'$ such that $\ket{\psi_\x}$ and $\ket{\psi_{\x'}}$ have non-negligible overlap, as otherwise this would break the OWSG.
Indeed, this means that 
the (randomized) {\em classical} map $\x \rightarrow S_\x$ is efficiently computable but computationally uninvertible, assuming OWSG security. 
However, given a shadow $S$ and a candidate key $\x$, it is not possible to efficiently verify whether $S$ was generated as a shadow of $\ket{\psi_\x}$. Indeed, as discussed before, the resulting primitive necessitates inefficient verification. 
In Section \ref{sec:owsgowp} we 
formalize this approach to build a {\bf one-way puzzle from any pure-state OWSG}. 

This allows us to reduce our problem to building commitments from one-way puzzles. The latter may at first appear to be easy, given the HILL technique.
But the quantum nature of one-way puzzles leads to a major technical barrier, that we describe next.\\

\noindent{\underline{\bf Barrier 2: No Flatness in a Quantum World.}}
\vspace{0.2mm}

Recall that a one-way puzzle sampler outputs classical $(k,s)$ pairs which satisfy the following: (1) $\mathsf{Ver}(k,s) = 1$ and (2) given $s$, it is computationally hard to find a preimage $k$ such that $\mathsf{Ver}(k,s) = 1$.
Here, observe that the distribution on preimage keys $k$ induced by fixing a puzzle output $s$ is not a ``flat'' distribution, i.e., it {\em does not necessarily assign equal probability mass to each preimage key}. 
Why does this matter?

For the following discussion, given any puzzle output string $s$, we let $K_s$ denote the distribution on keys induced by $s$, $\ell_s$ denote the min-entropy of $K_s$ and $N_s = |\mathsf{Supp}(K_s)|$.
Since $K_s$ is an arbitrary distribution, it can always be the case that $\ell_s \ll \ceil{\log N_s}$.

The construction of weak PEGs from one-way functions, discussed at the beginning of the overview, may seem to extend naturally to one-way puzzles as follows.
Consider distributions
\begin{align*}
\wgen(1^n) &:= ~~~~s, h, i, h({\x})_i ~~~~~~~~~~~\text{and} \\
    \simwgen(1^n) &:= 
    \left\{\begin{array}{ll}
     s, h, i, h({\x})_{i-1}, u_1 & \text{ if }i = \ceil{\log N_{s}}+1\\
     s, h, i, h({\x})_i & \text{ otherwise}
    \end{array}
    \right.
\end{align*}
where $(\x, s) \leftarrow \mathsf{OWPuzzle}.\mathsf{Samp}(1^n)$, $h \leftarrow \{0,1\}^n$, $i \leftarrow [n]$.

These distributions do differ in entropy, but they may not be computationally indistinguishable.
The leftover hash lemma (LHL) would imply that for any $s$, the first $\ell_s - 2c \log n$ bits of $h(k)$ are $\frac{1}{n^c}$-statistically close to uniform given $s$, where $\ell_s$ is the min-entropy of $K_s$. 
Any subsequent bits may leak information about the preimage $\x$. 
However, we note that $\wgen$ and $\simwgen$ differ on the $i^{th}$ bit of $h(k)$ for $i = \ceil{\log N_{s}}+1$ and $\ceil{\log N_s} \gg \ell_s$. But it is possible that all remaining bits of $h(k)$, i.e., $h(k)_{\ceil{\log N_{s}}}$, computationally leak the entire key $k$. This would make the distributions $\wgen$ and $\simwgen$ easily computationally distinguishable.

The argument above describes why modifying $\mathsf{G}_1$ on $i =  \ceil{\log N_{s}}+1$ doesn't work. What if we instead modified $\mathsf{G}_1$ on $i = \ell_s + 1$ instead, where $\ell_s$ is the min-entropy in $K_s$. That is, consider changing $\simwgen$ to the following distribution 
\begin{align*}
    \simwgen'(1^n):= 
    \left\{\begin{array}{ll}
     s, h, i, h({\x})_{i-1}, u_1 & \text{ if }\textcolor{black}{i = \ell_s+1}\\
     s, h, i, h({\x})_i & \text{ otherwise}
    \end{array}
    \right.
\end{align*}
In this case, the distributions $\wgen$ and $\simwgen$ become computationally indistinguishable, but the 
last bit in $\wgen$ corresponding to $i = \ell_s + 1$ could also be statistically close to uniform. As a result, 
$\wgen$ and $\simwgen'$ could end up being (almost) identical, with no entropy gap at all!

This problem does not arise in the classical setting, because flat preimages can be assumed without loss of generality 
by ``pulling out'' the (uniform) randomness from any classical algorithm. Letting $r$ denote the randomness used to sample $\x$, one can always define a (one-way) function that uses its uniform input $r$ to sample $\x$ and finally outputs $y = f(\x)$. This ensures that $\ell_s = N_s$ above, enabling simultaneous arguments for both computational indistinguishability and statistical entropy gap.
Unfortunately, this type of flattening is no longer possible when the sampler is a {\em quantum circuit}, because the randomness comes from a quantum process and we do not know how to explicitly pull it out.

At this point, it is natural to wonder whether there is some index $i$ for every key $k$ such that changing the $i^{th}$ bit of $h(k)$ in $\mathsf{G}_1$ yields a distribution that is computationally indistinguishable from $\mathsf{G}_0$ but has a statistical entropy gap. For example, perhaps one could consider modifying $\mathsf{G}_1$ at the first $i$ for which the statistical distance between $\mathsf{G}_0$ and $\mathsf{G}_1$ {\em jumps} from a value that is negligible at $i-1$ to a value that is not negligible at $i$. But there exist distributions for which there never is a clear cut ``jump''; for example, if the statistical distance between $\mathsf{G}_0$ and $\mathsf{G}_1$ increases proportionally to $2^{-{(n-i)}}$.


To overcome this issue, we will further modify $\mathsf{G}_1$. Our starting idea will be to fix for every puzzle $s$, a ``good set'' $\mathbb{G}_s$ of preimage keys which is almost flat. We set the distribution $\simwgen$ to differ from $\wgen$ {\em only when the key $k$ that is output by $\mathsf{Samp}$ belongs to the set $\mathbb{G}_s$}. Making this approach work requires several additional ideas, and we provide a detailed overview of these below. \\

\noindent{\bf \underline{Resolving Barrier 2: Pseudoentropy or ``HILL'' for Quantumly Sampled Distributions.}}
\vspace{0.2mm}

Our goal is to prove that the distribution
$$\wgen(1^n):= s, h, i, h({\x})_i$$ 
is a weak pseudoentropy generator, where $(\x, s) \leftarrow \mathsf{Samp}(1^n)$, $h \leftarrow \{0,1\}^n$ and $i \leftarrow [n]$.
Recall that this means we must demonstrate the existence of a different distribution which is computationally close to but has more entropy than $\wgen$.

We already outlined why setting $\simwgen$ identically to the case of one-way functions creates issues with proving either computational indistinguishability or statistical entropy gap.
Instead, we observe that 
for every puzzle $s$ and corresponding preimage distribution $K_s$, 
there is a ``good'' subset $\mathbb{G}_s$ of preimages $\x$ such that\footnote{We use numbers like $\frac{1}{n}$ below for simplicity. In our main section, we use slightly different fractions than the ones depicted here, for various technical reasons.}:
\begin{enumerate}
    \item The set $\mathbb{G}_s$ is dense enough in $K_s$, that is, for every $s$,
    \[\Pr_{x \leftarrow K_s}[x \in \mathbb{G}_s] \geq \frac{1}{n} \text{, and}\]
    \item For every pair of preimages $(x_1, x_2) \in \mathbb{G}_s$,  
\[0.5{\Pr_{K_s}[x_2]} \leq \Pr_{K_s}[x_1] \leq 2\Pr_{K_s}[x_2]\]
\end{enumerate}
The observation above essentially follows from a pigeonhole argument over the preimages of $s$.

We can now consider a {\em different} simulated distribution 
$\simwgen$ as follows.
\begin{align*}
    \simwgen(1^n):= 
    \left\{\begin{array}{ll}
     s, h, i, h({\x})_{i-1}, u_1 & \textcolor{red}{\text{ if }{i = \log|\mathbb{G}_s|+1}\text{ and }k \in \mathbb{G}_s}\\
     s, h, i, h({\x})_i & \text{ otherwise}
    \end{array}
    \right.
\end{align*}
where $(\x, s) \leftarrow \mathsf{Samp}(1^n)$, $h \leftarrow \{0,1\}^n$, $i \leftarrow [n]$, and $\mathbb{G}_s$ is the good slice of preimages defined above.
Computational indistinguishability between $\wgen$ and $\simwgen$ follows by noting that any distinguishing advantage can only exist when $\x \in \mathbb{G}_s$. Because we are reducing to a search problem, it is still possible to apply the Leftover Hash Lemma and the Goldreich-Levin theorem to convert any distinguisher into an inverter for the one-way puzzle.

Moreover, conditioned on $\x \in \mathbb{G}_s$, $\simwgen$ obviously has more entropy than $\wgen$ (and when $\x \not \in \mathbb{G}_s$, the two distributions are identical)\footnote{It may appear that we are close: we seem to have a pair of distributions that are statistically far but computationally close. Unfortunately, to obtain a commitment, we also need these distributions to be efficiently sampleable, which is something we will address in a later subsection. At this point, sampling from $\simwgen$ requires knowing $i_s$ and $\mathbb{G}_s$ which are not necessarily efficiently computable functions of $(k,s)$. But for now, we only aim to prove that $\wgen$ is a weak PEG, for which we only need $\wgen$ to be efficiently sampleable, and to prove that $\simwgen$ has {\em more entropy} than $\wgen$.}.
This unfortunately {\em does  not} imply that $\simwgen$ has more entropy than $\wgen$ overall.
The reason can best be explained with the following toy examples.\\ 

\noindent{\bf Example 1:}
There is a (hidden) event $B$  that occurs with probability $\frac{2}{3}$, and distributions $(A_0, A_1)$ such that
\begin{itemize}
\item Distribution $A_0$ outputs $0$ when $B$ occurs, and $1$ when $B$ doesn't occur.
\item Distribution $A_1$ outputs a uniform bit when $B$ occurs, and $1$ when $B$ doesn't.
\end{itemize}
$A_1$ has more entropy than $A_0$ conditioned on $B$, and the distributions are identical when $B$ doesn't occur -- but the overall entropy in $A_0$ is equal to that of $A_1$! 
Similarly, while $\simwgen$ has more entropy than $\wgen$ when $k \in \mathbb{G}_s$, and the distributions are identical otherwise, the overall entropy in $\wgen$ could end up being equal to that in $\simwgen$. 

Consider, however, the following example where both distributions $A_0$ and $A_1$ are uniform when $B$ doesn't occur. That is,\\

\noindent{\bf Example 2:}
There is a (hidden) event $B$ that occurs with probability $\frac{2}{3}$, and distributions $(A_0, A_1)$ such that
\begin{itemize}
\item Distribution $A_0$ outputs $0$ when $B$ occurs, and a uniform bit when $B$ doesn't occur.
\item Distribution $A_1$ outputs a uniform bit when $B$ occurs, and a uniform bit when $B$ doesn't.
\end{itemize}
$A_1$ has more entropy than $A_0$ conditioned on $B$. Moreover, since $A_0$ and $A_1$ are {\em uniform} when $B$ doesn't occur, then $A_1$ having higher entropy than $A_0$ conditioned on $B$ does imply that $A_1$ has higher entropy overall.

We could hope to apply a similar argument to $\wgen$ and $\simwgen$ if somehow it were the case that for $k \not \in \mathbb{G}_s$, the last bit of $\wgen$ (and also $\simwgen$) is close to uniform given the remaining bits $s, h, i, h(k)_{i-1}$. 
But why would this even be the case?\\

\noindent{\bf Establishing an Entropy Gap.}
Inspired by the insight above, we will modify $\simwgen$ to provably obtain an entropy gap.
For every $s$, the two distributions are identical when $i \neq \log |\mathbb{G}_s|+1$, therefore, we only focus on the case where $i = \log |\mathbb{G}_s|+1$.

In this case, ideally we want the bias in the $i^{th}$ bit of $h(k)$ when $k \not \in \mathbb{G}_s$ to not cancel out the bias that arises when $k \in \mathbb{G}_s$, in the distribution $\wgen$.
This would hold if the $i^{th}$ bit of $h(k)$ when $k \not \in \mathbb{G}_s$ were uniform (even given the remaining bits output by $\wgen$).
But we do not know if this is the case, or even what the distribution of keys when $k \not \in \mathbb{G}_s$ looks like.

To resolve this, let us first try to ensure that for {\em most} $k \in \mathbb{G}_s$, all preimages (besides $k$) of $(s, h(k)_{i-1})$ have {\em extremely} low sampling probability in $K_s$. (Recall that $K_s$ is the distribution induced on preimages of $s$.) 
We show that this is achieved by changing $\simwgen$ as follows:
\begin{align*}
    \simwgen(1^n):= 
    \left\{\begin{array}{ll}
     s, h, i, h({\x})_{i-1}, u_1 & \textcolor{red}{\text{ if }{i = (\log|\mathbb{G}_s|+600 \log n)}\text{ and }k \in \mathbb{G}_s}\\
     s, h, i, h({\x})_i & \text{ otherwise}
    \end{array}
    \right.
\end{align*}

In the modified distribution $\simwgen$, pairwise independence of $h$ ensures that for most $k \in \mathbb{G}_s$, all preimages (besides $k$) of $(s, h(k)_{i-1})$ are sampled with probability less than $\frac{1}{n^{600}}$ in $K_s$, for $i = (\log|\mathbb{G}_s|+600 \log n)$. 
For this overview, we assume that this holds for {\em all} $k \in \mathbb{G}_s$\footnote{In the technical sections, we further modify $\simwgen$ to account for the fact that for a few choices of $h$ and a few $k \in \mathbb{G}_s$, there are multiple preimages of $(h, h(k))$ that are each sampled with probability much higher than $\frac{1}{n^{600}}$.}.

We will now consider the following two cases.
\begin{enumerate}
    \item 
{\bf Preimages (besides $k$) of $(s, h(k)_{i-1})$ are sampled with total probability $\leq \frac{1}{n}$ in $K_s$.}

In this case, since the unique $k \in \mathbb{G}_s$ has probability mass at least $1 - \frac{1}{n}$, 
the bias in the $i^{th}$ bit of $h(k)$ from keys outside $\mathbb{G}_s$ barely stacks up against the bias that arises from $k \in \mathbb{G}_s$.
\item {\bf Preimages (besides $k$) of $s, h(k)_{i-1}$ are sampled with total probability $> \frac{1}{n}$ in $K_s$.} 

In this case, the individual probability mass of every preimage (besides $k$) is very low, but their total probability mass is high.
This means that the overall distribution of pre-images (besides $k$) of $(s, h(k)_{i-1})$ necessarily has high entropy. Then by the Leftover Hash Lemma, the $i^{th}$ bit of $h(k)$ when $k \not \in \mathbb{G}_s$ will be close to uniform, which is what we desired.
\end{enumerate} 

In both cases, we conclude that the {\em overall} Shannon entropy in $\simwgen$ is larger than that in $\wgen$ by a (fixed) inverse polynomial value: which means that $\wgen$ is indeed a WPEG. 
We prove this formally in Section \ref{sec:slicing}.

Furthermore, the proof of computational indistinguishability between $\wgen$ and $\simwgen$ goes through as before, with the leftover hash lemma guaranteeing that all but the last $O(\log n)$ bits of $h(k)$ are statistically indistinguishable from uniform. With this guarantee, the Goldreich-Levin inverter simply needs to guess the last $O(\log n)$ bits of $h(k)$, which it can do with inverse polynomial probability.

At this point, we have a WPEG. However, because $\simwgen$ cannot be efficiently sampled, we cannot use it directly to build commitments. Indeed, obtaining a full-fledged commitment requires additional steps, which we outline next.

\subsection{Pseudoentropy Generators (PEG)}
Our next step follows a similar technique as~\cite{HILL} to
(1) amplify the entropy gap between real and simulated distributions and (2) bring the min-entropy of the real distribution close to its Shannon entropy, all while maintaining computational indistinguishability.

    This is done by taking a product distribution of the outputs of the weak PEG.
    In more detail, we sample $q(n) = \poly(n)$ (for a large enough polynomial $\poly(\cdot))$ random keys $\x_1, \ldots, \x_q$ along with $q(n)$ independent $(h,i)$ values. We use these to generate $q(n)$ samples from distribution $\wgen$, and we append these samples together as our PEG output. 
    This also has the effect of ``concentrating'' the entropy to an expected value independent of the choice of $\x$ (whereas in weak PEG this entropy would necessarily depend on $\x$ via $i_s$).
    In the PEG, we have that for every choice of security parameter $n$, there is a {\em single} value $\widehat{h}_n$ that corresponds to the Shannon-entropy in the output of the PEG, and this value is smaller than the min-entropy in the corresponding simulated product distribution by $n^c$, for some $c>1$. 
\subsection{Imbalanced EFI Pairs}
    An EFI  is a pair of {\em efficiently sampleable} distributions that are statistically far but computationally indistinguishable. Such distributions have been shown~\cite{BCQ22,AC:Yan22} to be equivalent to quantum bit commitments.
     
    Let us consider hashing the output of our PEG to approximately $\widehat{h}_n + n$ bits. That is, the size of hash outcome is larger 
        than the actual Shannon entropy in the PEG output, making the resulting distribution statistically  distinguishable from uniform. At the same, since the PEG outputs {\em are computationally indistinguishable} from a distribution with (much) more than $\widehat{h}_n + n$ bits of min-entropy, the resulting hash output is still {\em be computationally indistinguishable} from uniform. 
        It may now seem like we have an EFI pair: consider distributions 
        \begin{itemize}
        \item $h(\mathsf{PEG}(n))$ where
        the output of $h$ is truncated to $\widehat{h}(n) + n$ bits
        \item $U_{\widehat{h}(n) + n}$ which
 is the uniform distribution over $\widehat{h}(n) + n$ bits        \end{itemize}
        While these distributions are computationally close but statistically far, they cannot be sampled efficiently without non-uniform advice, i.e., the value $\widehat{h}(n)$ for every $n$.

        In fact, observe that truncating the hash output to any less than $\widehat{h}_n + n$ bits would still preserve computational indistinguishability, and truncating to any more would still ensure statistical distance, but the two can {\em simultaneously} be guaranteed only when truncating to exactly $\widehat{h}(n) + n$ bits. This is why we call the resulting object an {\em imbalanced} EFI.  

        Due to the equivalence between EFI and commitment, we can equivalently claim to have statistically binding, computationally hiding quantum bit commitments~\cite{AC:Yan22,BCQ22} -- albeit dependent on non-uniform advice $z(n)$. When $z(n) \leq \widehat{h}(n) + n$, the commitments are hiding, and when $z(n) \geq \widehat{h}(n) + n$, these commitments are binding. We call this an {\em imbalanced commitment scheme}.
        The next few steps discuss how to remove this imbalanced drawback by appropriately combining variants of these commitments\footnote{This upcoming part diverges from techniques in~\cite{HILL} which build uniform pseudorandom generators by appropriately stretching the output of a nonuniform PRG. These techniques break down in our setting because there is no clear way to run a puzzle on its own output, and thus to achieve significant ``stretch'' in a puzzle-based PRG-type object.}.

\subsection{Always binding, Non-uniform hiding Commitments}
        In the next step, we rely on prior work in {\em flavor conversion} of quantum commitments~\cite{HMY22} to convert our statistically hiding, computationally binding EFI pairs/commitments to commitments with the reverse property: namely, where for
        $z(n) \geq \widehat{h}(n) + n$, the commitments are hiding, and when $z(n) \leq \widehat{h}(n) + n$, these commitments are binding.
        Next, given these two types of complementary commitments, we {\em combine them} by using both to commit to the {\em same} bit $b$: note that for every choice of advice $z(n)$ (i.e. length to which we truncate the hash outcome), at least one of the two commitments is necessarily binding. This allows us to show that the resulting combined commitment is always binding (for every choice of $z(n)$) and hiding whenever $z(n) = \widehat{h}(n) + n$. We call this a non-uniform hiding commitment.
    
\subsection{Standard (Uniform) Commitments}
    Finally, we observe that for each $n$, the number of possible $\widehat{h}(n)$ values is bounded by a fixed polynomial $t(n)$. Thus, we can repeat the above construction for every possible value of $\widehat{h}(n)$, obtaining a sequence of commitments where for every $n$, at least one commitment in the sequence is hiding (and all are binding). By secret sharing the committed bit between various commitments, we can show that the overall commitment scheme satisfies both hiding and binding. Thus, we have removed dependence on the advice string $z(n)$, yielding a uniform construction of commitments.\\

\noindent{\bf Quick Detour: An Alternative Template.} We briefly note an alternative technique~\cite{C:MorYam22} using quantum information to sidestep the use of the hardcore bit. 
While this was developed to build commitments from a strong ``injective'' variant of OWSG, for simplicity, we describe it as applied to injective one-way functions. Very roughly (and ignoring some garbage registers), a commitment to $0$ is $\sum_\x \ket{\x}_\mathsf{C} \ket{\x, f(\x)}_\mathsf{D}$ and a commitment to $1$ is $\sum_\x \ket{\x}_\mathsf{C} \ket{0, f(\x)}_\mathsf{D}$; where $\mathsf{C}$ is the commit register and $\mathsf{D}$ the decommit register. Statistical hiding follows because tracing out the $\mathsf{D}$ register leaves identical mixtures on $\mathsf{C}$ in both cases. Computational binding follows by the hardness of finding $\x$ given $f(\x)$ for a random $\x$. Our methods, including hashing preimages to appropriate lengths and slicing, will also similarly apply to this template. We do not find any one of these templates to be simpler than the other, but we focus on the Goldreich-Levin template because it yields interesting intermediate primitives with entirely classical outputs.\\

\noindent{\bf \underline{Roadmap.}}
In the following section, we recall useful definitions and facts about quantum cryptography and various notions of entropy. In Section \ref{sec:owsgowp}, we show that OWSG with pure states imply one-way puzzles.
Our construction of commitments from one-way puzzles is detailed in Sections \ref{sec:slicing}-\ref{sec:com}.
In Section \ref{sec:slicing}, we show how one-way puzzles imply a weak pseudoentropy generator (WPEG). Then, in Section \ref{sec:product} we describe how a product distribution of quantum weak PEGs gives rise to a pseudoentropy generator (PEG).
Next in Section \ref{sec:imbefi} we obtain imbalanced EFI by appropriately hashing the output of the PEG. Finally, in Section \ref{sec:com} we apply flavor swap and other combiners to the imbalanced EFI to obtain a uniform construction of quantum bit commitments.

In Appendix \ref{app:implications}, we provide evidence that one-way puzzles are a necessary assumption for quantum cryptography with classical communication.

%% file: prelim.tex
\section{Preliminaries}
In this section, we discuss some notation and preliminary information, including definitions, that will be useful in the rest of the exposition.

\subsection{Notation and Conventions}

We write $\negl(\cdot)$ to denote any \emph{negligible} function, which is a function $f$ such that for every constant $c \in \mathbb{N}$ there exists $N \in \mathbb{N}$ such that for all $n > N$, $f(n) < n^{-c}$.
For any $k$, we will denote by $U_k$ the uniform distribution supported on $k$ bits.
We will use $\mathsf{SD}(A,B)$ to denote the statistical distance between (classical) distributions $A$ and $B$.\\

\noindent{\bf Quantum conventions.} A register $\sX$ is a named Hilbert space $\bbC^{2^n}$. A pure state on register $\sX$ is a unit vector $\ket{\psi} \in \bbC^{2^n}$, and we say that $\ket{\psi}$ consists of $n$ qubits. A mixed state on register $\sX$ is described by a density matrix $\rho \in \bbC^{2^n \times 2^n}$, which is a positive semi-definite Hermitian operator with trace 1. 

A \emph{quantum operation} $F$ is a completely-positive trace-preserving (CPTP) map from a register $\sX$ to a register $\sY$, which in general may have different dimensions. That is, on input a density matrix $\rho$, the operation $F$ produces $F(\rho) = \tau$ a mixed state on register $\sY$.
A \emph{unitary} $U: \sX \to \sX$ is a special case of a quantum operation that satisfies $U^\dagger U = U U^\dagger = \bbI^{\sX}$, where $\bbI^{\sX}$ is the identity matrix on register $\sX$. A \emph{projector} $\Pi$ is a Hermitian operator such that $\Pi^2 = \Pi$, and a \emph{projective measurement} is a collection of projectors $\{\Pi_i\}_i$ such that $\sum_i \Pi_i = \bbI$.


The security of all our constructions holds against adversaries that receive non-uniform quantum advice. More specifically, we refer to security against all quantum polynomial-sized adversaries in our definitions. By this we mean every family of (polynomial-sized) quantum circuits $\{C_n\}_{n\in \bbN}$ along with a family of states $\{\rho_n\}_{n\in\bbN}$ such that for every $n\in\bbN$, $C_n$ obtains ancilla registers set to $\rho_n$.
We rely on a non-uniform reduction in Section \ref{sec:product}, which obtains (classical) non-uniform advice for every value of the security parameter $\n$. 

Finally, we often define security of primitives by bounding an adversary's advantage in a search/distinguishing game by $\negl(\n)$. Depending on the order of quantifiers, this is typically considered to mean one of two things -- (1) there exists a negligible function $\mu(\cdot)$ such that the probability that any quantum polynomial-sized adversary $\cA$ wins a search/distinguishing game is at most $\mu(n)$ (for large enough $n$), and (2) for every quantum polynomial-sized adversary $\cA$, there is a negligible function $\mu_\cA(\cdot)$ such that the probability that $\cA$ wins a search/distinguishing game is at most $\mu_\cA(n)$. The two definitions/interpretations are equivalent (proved in~\cite{Bellarenote} for classical circuits but the proof also applies to quantum circuits), and we use both interchangeably.

\subsection{Quantum Cryptographic Primitives}
\begin{definition}[$t(n)$-Copy Secure One-Way State Generators]~\cite{C:MorYam22}
\label{def:owsg}
A one-way state generator (OWSG) is a set of QPT algorithms $(\KeyGen, \StateGen, \Ver)$ where:
\begin{itemize}
    \item $\KeyGen(1^\n)$: On input the security parameter $\n$, output a classical key string $\x \in \bin^n$.
    \item $\StateGen(\x)$: On input key $\x \in \{0,1\}^n$, output an $m$-qubit quantum state $\ket{\psi_\x}$.
    \item $\Ver(\x, \ket{\psi})$ : On input key $\x \in \{0,1\}^n$ and $m$-qubit quantum state $\ket{\psi}$, output $\top$ or $\bot$.
\end{itemize}
These algorithms satisfy the following properties.
\begin{itemize}
\item \textbf{Correctness.}
For every $n \in \mathbb{N}$,
\begin{align*}
    \Prr_{\substack{\x \leftarrow \KeyGen(1^\n)\\
    \ket{\psi_\x} \leftarrow \StateGen(\x)}}\left [ 
    \top\leftarrow\Ver(\x, \ket{\psi_\x})
    \right ] 
    \geq 1-\negl(\n)
\end{align*}
    
\item \textbf{$t(\n)$-Copy Security.} 
For every quantum polynomial-sized adversary $\cA = \{\cA_\n\}_{\n \in \mathbb{N}}$
and $\n \in \mathbb{N}$,
\begin{align*}
    \Prr_{\substack{\x \leftarrow \KeyGen(1^\n)\\
    \ket{\psi_\x} \leftarrow \StateGen(\x)}}
    \left [ 
    \top\leftarrow\Ver\left(\cA_n((\ket{\psi_\x})^{\otimes t(\n)}), \ket{\psi_\x}\right)
    \right ] 
    \leq\negl(\n)
    \end{align*}

\end{itemize}
\end{definition}
This definition was later generalized in~\cite{morimaeOneWaynessQuantumCryptography2022} to allow $\mathsf{StateGen}$ outputs to be mixed states.
Furthermore, existing definitions of OWSG ~\cite{C:MorYam22} require $t(n)$-copy security for every (a-priori unbounded) polynomial $t(\cdot)$. In this work, we only need to consider $cn$-copy security for a large enough, a-priori fixed, constant $c$. We will show that even this weaker variant implies commitments, thus obtaining a stronger result.

\begin{definition}[One-way Puzzles]
\label{def:owp}
A one-way puzzle is a pair of sampling and verification algorithms $(\mathsf{Samp},\mathsf{Ver})$
with the following syntax. 
\begin{itemize}
\item $\mathsf{Samp}(1^n) \rightarrow (k,s)$, is a QPT algorithm that outputs a pair of classical strings $(k,s)$.
We refer to $s$ as the puzzle and $k$ as its key. Without loss of generality we may assume that $k\in\bin^n$.
\item $\mathsf{Ver}(k, s) \rightarrow \top$ or $\bot$,
is an unbounded algorithm that on input any pair of classical strings $(k,s)$ halts and outputs either $\top$ or $\bot$.
\end{itemize}
These satisfy the following properties.
\begin{itemize}
\item {\bf Correctness.} Outputs of the sampler pass verification with overwhelming probability, i.e., 
$$\Prr_{(k, s) \leftarrow \mathsf{Samp}(1^n)} [\Ver(k,s) = \top] = 1 - \negl(n)$$ 
\item {\bf Security.}
Given $s$, it is (quantum) computationally infeasible to find $k$ satisfying $\Ver(k,s) = \top$, i.e., for every quantum polynomial-sized adversary $\cA$,
 $$\Prr_{(k,s) \leftarrow \Samp(1^n)}[\Ver(\mathcal{A}(s), s) = \top] = \negl(n)$$
\end{itemize}
Note that since puzzles are efficiently sampleable, there exists a polynomial $p(\cdot)$ such that $|s| \leq p(n)$.
\end{definition}

\begin{definition}[EFI pairs]~\cite{BCQ22} 
\label{def:efi}
An EFI pair is a QPT algorithm $\EFI(1^\n,b)\rightarrow\rho_b$ that on input $b \in \bin$ and the security parameter $\n$, outputs a (potentially mixed) quantum state $\rho_b$ such that the following hold:

\begin{enumerate}
    \item {\bf Computational Indistinguishability.} 
    There exists a negligible function $\mu(\cdot)$ such that for every quantum polynomial-sized adversary $\cA$, for large enough $\n \in \bbN$, 
    \[
    \left| \Pr[1\leftarrow \cA(\EFI(1^\n,0))] - \Pr[1\leftarrow \cA(\EFIGen(1^\n,1))]\right| \leq \mu(\n)
    \]
    \item {\bf Statistical Binding.}
    There exists a negligible function $\delta(\cdot)$ such that for large enough $\n \in \bbN$, 
    \[\TD(\EFIGen(1^\n, 0), \EFIGen(1^\n, 1)) \geq 1 - \delta(\n)\]
\end{enumerate}
\end{definition}

\subsection{Entropy and Randomness Extractors}
We will use $\Hs(X)$ to denote the Shannon entropy, $\Hs_{\min}(X)$ to denote the min-entropy and $\Hs_{\max}(X)$ to denote the max-entropy of distribution $X$.
For an arbitrary classical random variable $X$, the sample entropy~\cite{hill-revisited} of any $x \in \mathsf{Supp}(X)$ is defined as
\[\Hs_X(x) := \log(1/\Pr[X=x]).\] 
The min-entropy of distribution $X$ is then
$$\Hs_{\min}(X) := \underset{x \in \mathsf{Supp}(X)}{\min}\Hs_X(x)$$
and the max-entropy is
$$\Hs_{\max}(X) := \underset{x \in \mathsf{Supp}(X)}{\max}\Hs_X(x)$$
For discrete random variables $X$ and $Y$, we define conditional entropy of $X$ given $Y$ as
\[
	\Hs(X|Y) := \sum_{y\in\supp(Y)} \Pr[Y=y]\cdot H(X|Y=y)
\]
We also note the chain rule for conditional entropy
\[
	\Hs(X,Y) = \Hs(X|Y) + \Hs(Y)
\]

\begin{definition}(Smooth Min-Entropy)
 Let $X$ be a discrete random variables. The $\epsilon$-Smooth Min Entropy of $X$ is defined as:
 \[
   \Hs^\epsilon_{\min}(X) := \mathsf{sup}_{X'} \Hs_{\min}(X')
\]
where $X'$ is an arbitrary random variable at most $\epsilon$ statistically far from $X$.
\end{definition}

We also have the following well-known lemmas and theorems about entropy and randomness extraction.

\begin{lemma}[Leftover Hash Lemma] 
\label{thm:LHLforSmoothEntropy}Fix $\epsilon, \epsilon' > 0$. Let $X$ be a discrete random variable distributed over $\bin^n$, and $\Hs^\epsilon_{\min}(X) \geq k$. Let $H$ be a universal hash family with output length $k - 2\log(1/\epsilon')$. Then
    \[
    \SD \Big( (H(X), H), (U, H) \Big) \leq \epsilon+\epsilon'
    \] where 
    $U$ is uniformly distributed in  $\bin^{k - 2\log(1/\epsilon')}$. 
    \end{lemma}
\begin{theorem}[Entropy Concentration in Product Distributions]~\cite{Renner08}
\label{thm:conc}
Let $X$ be a discrete random variable taking values in a universe $\cU$, let $t \in \bbN$ and let $\epsilon > 0$.  Then,
    \begin{gather*}
   \Hs^\epsilon_{\min}(X^t) \geq  t\cdot \Hs(X) - \sqrt{2t\cdot \log(1/\epsilon)}\cdot\log(3 + |\cU|) 
    \end{gather*}
\end{theorem}

Finally, we will use the following version of the Goldreich-Levin theorem (in the presence of quantum advice).
\begin{theorem}[Goldreich-Levin with Quantum Advice]~\cite{AC}
\label{thm:gl}
    There exists a QPT algorithm $\cE$ such that if
    for any QPT algorithm $\cD$, any $\epsilon > 0$, any mixed state $\rho$, 
    and any $a \in \{0,1\}^n$,
     $$\Prr_{x \leftarrow \{0,1\}^n}[\cD(\rho, x) = \langle a, x \rangle] \geq \frac{1}{2} + \epsilon$$
    then,
    $$\Prr[\cE^\cD(\rho) = a] \geq 4\epsilon^2$$
\end{theorem}

%% file: owsg-imply-efi.tex
\section{Pure OWSG imply One-Way Puzzles}
\label{sec:owsgowp}

In this section, we show how to use shadow tomography to build one-way puzzles (with inefficient verification) from any OWSG with pure state outputs.

We will rely on the following theorem on shadow tomography from $\cite{HKP20}$.
\begin{theorem}~\cite{HKP20} (Rephrased, following~\cite{Yuennotes})
\label{thm:hkp}
Fix any $\epsilon, \delta > 0$. There exists a polynomial $p(\cdot)$ and QPT algorithm $\shadow$ that, given $T = O(\log(1/\delta)/\epsilon^2)$ copies of an unknown state $\ket{\psi}$ generates a classical string (called the ``shadow'') $S$ of size $p(n)$ with the following property: 

For some $t\in\bbN$, let $\{M_i\}_{i\in[t]}$ be a set of observables such that $\tr(M_i^2) \leq 1$. Then there exists an ``estimator'' function $E$ such that:
\[
    \Prr_{S\leftarrow\shadow(\ket{\psi^{\otimes T}}}
    \Big[\forall i \in [t],  \big|E(S, M_i)- \bra{\psi}M_i\ket{\psi}\big| \leq \epsilon \Big] \geq 1 - t\delta 
\]
\end{theorem}

We now proceed to state and prove our main theorem for this section.

\begin{theorem}
\label{thm:five-two}
    There exists a constant $c>0$ such that any $cn$-copy secure one-way state generator with pure state outputs (Definition \ref{def:owsg}) implies a one-way puzzle (Definition \ref{def:owp}).
\end{theorem}

\begin{proof} (of Theorem \ref{thm:five-two})
Let $(\KeyGen, \StateGen)$ be a one-way state generator (OWSG) with pure state outputs and let $\ket{\psi_k}$ represent the output of $\StateGen(k)$.

To build a puzzle from this OWSG, we will apply shadow tomography to the output states of the OWSG. 
In fact, the one-way puzzle will simply sample a OWSG key $k$, compute $\ket{\psi_k} \leftarrow \StateGen(k)$, and finally compute $s_k$ as a classical shadow of $\ket{\psi_k}$. It will output $s_k$ as the puzzle, with solution $k$. In what follows, we formalize this construction and define an (inefficient) verification algorithm for the one-way puzzle.\\

\noindent{\bf Defining preimage keys of a classical shadow.}
First, it will be useful to define an (inefficient) algorithm $\cL$ that obtains a classical shadow and outputs a list of keys, roughly corresponding to possible preimages of the shadow.

Set $\epsilon = 1/10$ and for $n \in \bbN$, set $\delta=\delta(n) = 2^{-2n}$. 
For all $k \in \Supp(\KeyGen(1^n))$, define $M_k := \ket{\psi_k}\bra{\psi_k}$. 
Note that these satisfy $\tr(M_k^2) = 1$.
Let $\shadow$ and $E$ be algorithms as defined by Theorem \ref{thm:hkp} applied to $\delta, \epsilon$ and $\{M_k\}_{k \in \Supp(\KeyGen(1^n))}$. Let $T = T(n) =O(n)$ be the required number of copies, and let $t = t(n) := \left|\{M_k\}_{k \in \Supp(\KeyGen(1^n))}\right| \leq \left| \Supp(\KeyGen(1^n)) \right| \leq 2^n$.

Define the (inefficient) deterministic algorithm $\cL$ that takes a shadow $s$ as input and outputs a list of keys such that the estimated overlap of the shadow with each key in the list is at least $1-\epsilon$. That is, 
\[
    \cL(s) = \Bigg\{ k: \Big( k \in \Supp(\KeyGen(1^n)) \Big) \bigwedge \Big( E(s, M_k) \geq 1-\epsilon \Big) \Bigg\}
\]

The following claim about the algorithm $\cL$ 
states that for any key $k$, with high probability over sampling a corresponding shadow $s_k$ of $\ket{\psi_k}$,
(1) the key $k$ appears in $\cL(s_k)$  and 
(2) for all $j \in \cL(s_k)$, the (pure) states $\ket{\psi_k}$ and $\ket{\psi_j}$ have high overlap.
The proof of this claim follows from the correctness of shadow tomography (Theorem \ref{thm:hkp}).

\begin{claim}
\label{clm:shadowProperties}
For large enough $n \in \bbN$, for all $k \in \Supp(\KeyGen(1^n))$: 
 \begin{enumerate}
     \item $\Prr_{s_k \leftarrow\shadow(\ket{\psi_k}^{\otimes T})}[k \in \cL(s_k)] \geq 1 - 2^{-n}$
    \item $\Prr_{s_k \leftarrow\shadow(\ket{\psi_k}^{\otimes T})}[\forall j \in \cL(s_k), |\langle \psi_k | \psi_j \rangle|^2  \geq 1-2\epsilon] \geq 1 - 2^{-n}$
 \end{enumerate}
\end{claim}
\begin{proof}
For any large enough $n \in \bbN$ and any $k \in \Supp(\mathsf{KeyGen}(1^n))$, applying Theorem \ref{thm:hkp} on $\delta, t, \epsilon$ set as above, we have:
    \begin{equation}
    \label{eq:abc}
        \Prr_{s_k \leftarrow\shadow(\ket{\psi_k}^{\otimes T})}[\forall j \in \Supp(\KeyGen(1^n)), \left|E(s_k, M_j)- |\langle \psi_k| \psi_j \rangle|^2 \right| \leq \epsilon] \geq 1 - t\delta\geq 1 - 2^{-n}
    \end{equation}
    Setting $j=k$, we have:
    \[
        \Prr_{s_k \leftarrow\shadow(\ket{\psi_k}^{\otimes T})}[E(s_k, M_k) \geq 1 - \epsilon] \geq 1 - 2^{-n}
    \]
    By definition of $\cL$, this implies 
    \[
    \Prr_{s_k \leftarrow\shadow(\ket{\psi_k}^{\otimes T})}[k \in \cL(s_k)] \geq 1 - 2^{-n}
    \]
    which is the first part of the claim. 
    
    Again, fix any $k \in \Supp(\mathsf{KeyGen}(1^n))$. If we restrict $j$ to $\cL(s_k)$, then by equation (\ref{eq:abc}), we have
    \[
        \Prr_{s_k \leftarrow\shadow(\ket{\psi_k}^{\otimes T})}[\forall j \in \cL(s_k), \left|E(s_k, M_j)- |\langle \psi_k| \psi_j \rangle|^2 \right| \leq \epsilon] \geq 1 - 2^{-n}
    \]
    But $j \in \cL(s_k) \iff E(s_k,M_j) \geq 1-\epsilon$. Substituting in the above equation gives:
    \[
        \Prr_{s_k \leftarrow\shadow(\ket{\psi_k}^{\otimes T})}[\forall j \in \cL(s_k), |\langle \psi_k| \psi_j \rangle|^2 \geq 1 - 2\epsilon] \geq 1 - 2^{-n}
    \]
    which is the second part of the claim.
\end{proof}
\noindent
Before describing our puzzle, 
we will define the set $\bbC$ of keys that have low correctness error, as follows:
\[
\bbC := \left\{k \in \bin^n \text{ such that } \Prr[\Ver(k, \ket{\phi_k}) = \top] \geq 1- 1/100\right\}
\]
Observe that with overwhelming probability,  the OWSG $\mathsf{KeyGen}$ algorithm outputs keys in the set $\bbC$ (otherwise, this would contradict correctness of the one-way puzzle).
Looking ahead, our puzzle verification algorithm will reject keys that are not in $\bbC$. We can now formally describe the puzzle.\\

\noindent {\bf Constructing the One-Way Puzzle.}
Define a one-way puzzle as follows.
\begin{itemize}
    \item $\mathsf{Puzz}.\Samp(1^n):$
    \begin{itemize}
        \item Sample $k \leftarrow \KeyGen(1^n)$.
        \item Compute $s \leftarrow \shadow(\ket{\psi_k}^{\otimes T})$
        \item Return $(k,s)$
    \end{itemize}
    \item $\mathsf{Puzz}.\Ver(k,s):$
    \begin{itemize}
        \item If $k\in \cL(s)$ and $k \in \bbC$, return $\top$
        \item Else return $\bot$
    \end{itemize}
\end{itemize} 
\begin{claim}
    \label{OWSGtoOWP}
    $(\mathsf{Puzz}.\Samp, \mathsf{Puzz}.\Ver)$ satisfies Definition \ref{def:owp}.
\end{claim}
{\bf Correctness.} By Claim \ref{clm:shadowProperties} part 1, for large enough $n \in \bbN$ and all $k \in \Supp(\KeyGen(1^n))$ 
\[
\Prr_{s\leftarrow\shadow(\ket{\psi_k}^{\otimes T})}[k \in \cL(s)] \geq 1 - 2^{-n}
\]
Since the OWSG must have negligible correctness error, a Markov argument applied to Definition \ref{def:owsg} shows that 
\[
\Prr_{k \leftarrow \KeyGen(1^n)}\left[k \notin \bbC\right] \leq \negl(n)
\]
Putting these together,
\[
\Prr_{\substack{k \leftarrow \KeyGen(1^n)\\s\leftarrow\shadow(\ket{\psi_k}^{\otimes T})}}\left[k \in \cL(s) \text{ and }k \in \bbC\right] \geq 1-2^{-n} - \negl(n)
\]
which by the definition of $\mathsf{Puzz}.\Ver$ implies
\[
    \Prr_{(k,s)\leftarrow\mathsf{Puzz}.\Samp(1^n)}[\mathsf{\top \leftarrow Puzz}.\Ver(k,s)] \geq 1 - 2^{-n} - \negl(n) = 1 - \negl(n)
\]
{\bf Security.} We prove one-wayness by contradiction. Suppose there exists a quantum polynomial-sized adversary $\cA$ that breaks the one-wayness of the puzzle, i.e. there exists a polynomial $q(\cdot)$ 
such that for infinitely many $n\in\bbN$,
\[
    \Prr_{(k,s)\leftarrow\mathsf{Puzz}.\Samp(1^n)}[\top \leftarrow \mathsf{Puzz}.\Ver(\cA(s),s)] \geq 1/q(n)
\]
We build a reduction that breaks the one-wayness of the OWSG. 
First, by the definition of $\mathsf{Puzz}.\Ver$ and $\mathsf{Puzz}.\Samp$, for infinitely many $n\in\bbN$,
\begin{equation}
\label{eq:t}
    \Prr_{\substack{k \leftarrow \KeyGen(1^n)\\s\leftarrow\shadow(\ket{\psi_k}^{\otimes T})}}[\cA(s) \in \left(\cL(s)\cap \bbC\right)] \geq 1/q(n)
\end{equation}
By Claim \ref{clm:shadowProperties}, for all $n \in \bbN$ and all $k \in \Supp(\KeyGen(1^n))$, 
\begin{equation}
\label{eq:s}
\Prr_{s \leftarrow \shadow(\ket{\psi_k}^{\otimes T})}[\forall k' \in \cL(s), |\langle \psi_k | \psi_{k'} \rangle|^2  \geq 4/5] \geq 1 - 2^{-n}
\end{equation}
For any events $A$ and $B$, $\Pr[A\wedge B]\geq \Pr[A] - \Pr[\neg B]$. Therefore, from equations (\ref{eq:t}) and (\ref{eq:s}), for infinitely many $n\in\bbN$,
\begin{equation}
    \Prr_{\substack{k \leftarrow \KeyGen(1^n)\\s\leftarrow\shadow(\ket{\psi_k}^{\otimes T})}}\Big[ \big(\cA(s) \in (\cL(s) \cap \bbC) \big) \wedge \big(\forall k' \in \cL(s), |\langle \psi_k | \psi_{k'} \rangle|^2  \geq 4/5 \big)\Big] \geq 1/q(n) - 2^{-n}
\end{equation}
which can be simplified to say that for infinitely many $n\in\bbN$,
\[
    \Prr_{\substack{k \leftarrow \KeyGen(1^n)\\s\leftarrow\shadow(\ket{\psi_k}^{\otimes T})\\ k' \leftarrow \cA(s)}}\Big[\big(k'\in \bbC\big) \wedge \big(|\langle \psi_k | \psi_{k'} \rangle|^2  \geq 4/5\big)\Big] \geq 1/q(n) - 2^{-n} 
\]
If $|\langle \psi_k | \psi_{k'} \rangle|^2  \geq 4/5$ then the success probabilities of $\Ver(k', \ket{\psi_{k'}})$ and $\Ver(k', \ket{\psi_{k}})$ can differ by at most $\frac{1}{\sqrt{5}}$. Since for all $k'$ in $\bbC$,  $\Ver(k', \ket{\psi_{k'}})$ succeeds with probability atleast $1-\frac{1}{100}$, for infinitely many $n\in\bbN$,
\[
    \Prr_{\substack{k \leftarrow \KeyGen(1^n)\\s\leftarrow\shadow(\ket{\psi_k}^{\otimes T})\\ k' \leftarrow \cA(s)}}[ \top \leftarrow\Ver(k', \ket{\psi_k})] \geq (1-1/100-1/\sqrt{5})\cdot(1/q(n) - 2^{-n}) >\frac{1}{2q(n)}
\]
Then, letting $\cB$ be the algorithm that on input $\ket{\psi_k}^{\otimes T}$ outputs $\cA(\shadow(\ket{\psi_k}^{\otimes T}))$, we have that for infinitely many $n\in\bbN$,
\[
    \Prr_{\substack{k \leftarrow \KeyGen(1^n),\\k' \leftarrow \cB(\ket{\psi_k}^{\otimes T})}}[\top \leftarrow\Ver(k', \ket{\psi_k})] > \frac{3}{5q(n)}
\]
Since $\shadow$ and $\cA$ are quantum polynomial-sized circuits, this contradicts $O(n)$-copy security of the OWSG.
\end{proof}

\section{One-way Puzzles imply Quantum Weak PEGs}
\label{sec:slicing}

Here, we show that (inefficiently verifiable) one-way puzzles imply quantum weak pseudoentropy generators, defined below.

\begin{definition}[Quantum Weak Pseudoentropy Generator]
\label{def:qwpeg}
    A Quantum Weak Pseudoentropy Generator consists of an ensemble of distributions $\{\wgen(n), \simwgen(n)\}_{n\in\bbN}$ over classical strings :
    \begin{itemize}
        \item \textbf{Efficiency.} There exists a QPT algorithm $\cG$ where for all $n \in \bbN$, $\cG(1^n)$ returns a sample from $\wgen(n)$.
      	\item \textbf{Bounded Length}. There exists a polynomial $p(\cdot)$ such that for all $n \in \bbN$, for all $z_0 \in \Supp(\wgen(n))$, for all $z_1 \in \Supp(\simwgen(n))$, $|z_0|=| z_1| \leq p(n)$.
        \item \textbf{(Shannon) Entropy Gap}. There exists an explicit constant $c>0$ such that for all sufficiently large $n \in \bbN$,
        \[
            \Hs(\simwgen(n)) - \Hs(\wgen(n)) \geq \frac{1}{n^c} 
        \]
        \item \textbf{Indistinguishability.} There exists a negligible function $\mu$ such that for all quantum polynomial-sized adversaries $\cA$, for all large enough $n \in \bbN$:
        \[
            \left|\Prr_{z\leftarrow \wgen(n)}[\cA(z) = 1] - \Prr_{z\leftarrow \simwgen(n)}[\cA(z) = 1]\right| \leq \mu(n)
        \]
    \end{itemize}
\end{definition}

\begin{theorem}
\label{thm:OWPtoQWPEG}
    One-way puzzles (Definition \ref{def:owp}) imply quantum weak pseudoentropy generators (Definition \ref{def:qwpeg}).
\end{theorem}
The rest of this section is devoted to the proof of this theorem.
For the following discussion, fix some sufficiently large $n \in \bbN$. In what follows, we will sometimes drop explicit parameterization on $n$ when it is clear from the context.

\subsection{Defining \texorpdfstring{$\wgen$}{G0} and \texorpdfstring{$\simwgen$}{G1}}
In order to build QWPEG, we first need to define the two distributions, $\wgen$ and $\simwgen$. \\

\noindent{\bf The Distribution $\wgen$.}
The distribution $\wgen$ is defined as:
    \[
        \wgen(n):= s, h, i, r, h(k)_i, \langle k, r \rangle
    \]
where 
$(k, s) \leftarrow \mathsf{Samp}(1^n)$, 
$h \in \bin^\ell$ is a (uniform) seed for a pairwise independent hash function mapping $\bin^{n}$ bits to $\bin^{3n}$ bits,
$i \leftarrow [3n]$, and $r \leftarrow \bin^n$.
Moreover, $h(k)_i$ denotes the output of $h(k)$ truncated to the first $i$ bits.\\

\noindent{\bf The Distribution $\simwgen$.} The distribution $\simwgen$ is defined as:
\begin{align*}
    \simwgen(n):= 
    \left\{\begin{array}{ll}
     s, h, i, r, h(k)_i, u & \text{ if $i = i_{s}, k \in \mathbb{G}_{s}, \left(h, h(k)_i\right) \in \mathbb{F}_{s}$}\\ 
     s, h, i, r, h(k)_i, \langle k, r \rangle & \text{ otherwise}
    \end{array}
    \right.
\end{align*}
where
$(k, s) \leftarrow \mathsf{Samp}(1^n)$, 
$h \in \bin^\ell$ is a (uniform) seed for a pairwise independent hash function mapping $n$ bits to $3n$ bits,
$i \leftarrow [3n]$, $u$ is a uniform bit, and $r \leftarrow \bin^n$. 

It remains to define $i_s$ and the sets $\mathbb{G}_s, \mathbb{F}_s$, which we do next. 
First, we define some notation.
\begin{itemize}
\item Let $\mathbb{K} = \Supp(K)$, where $K$ is the marginal distribution on keys induced by $\Samp(1^n)$. 
\item Let $\mathbb{S} = \Supp(S)$, where $S$ is the marginal distribution on puzzles induced by $\Samp(1^n)$. 
\item For any $s \in \mathbb{S}$, denote by $K_s$ the distribution on keys output by $\mathsf{Samp}$ conditioned on the puzzle being $s$, and let $\bbK_s=\Supp(K_s)$.
\item For any $k \in \mathbb{K}$, denote by $S_k$ the distribution on puzzles output by $\mathsf{Samp}$ conditioned on the key being $k$, and let $\bbS_k=\Supp(S_k)$.
\end{itemize}

\noindent{\bf Defining $\mathbb{G}_s$ and $i_s$.}
We claim the existence of a dense enough set $\mathbb{G}_s$ of preimage keys of every puzzle $s$, where all keys in $\mathbb{G}_s$ have (roughly) the same sample entropy in distribution $K_s$. We call this sample entropy $j_s$, and we will later use it to define $i_s$.

\begin{claim}
    \label{clm:flatSet}
    For all $n \in \mathbb{N}$ and all $s \in \bbS$, there exists a set $\mathbb{G}_s \subseteq \bbK_s$ 
    and $j_s \in [0,2n-1]$
    such that
    \begin{enumerate}
        \item $\Prr_{k\leftarrow K_s}[k \in \mathbb{G}_s] \geq 1/3n$
        \item 
        For all $k \in \mathbb{G}_s$, we have $j_s \leq \Hs_{K_s}(k) \leq j_s + 1$.
    \end{enumerate}
\end{claim}
\begin{proof}
    We first show that most of the probability mass in $K_s$ is on keys with sample entropies less than $2n$. We fix some arbitrary $s \in \bbS$ for the following discussion.
\begin{align*}
    \Prr_{k \leftarrow K_s}[\Hs_{K_s}(k) \geq 2n] &= \sum_{k: \Hs_{K_s}(k) \geq 2n}\Pr_{K_s}[k]\\
    &= \sum_{k: \Pr_{K_s}(k) \leq 2^{-2n}}\Pr_{K_s}[k]\\
    &\leq \sum_{k: \Pr_{K_s}(k) \leq 2^{-2n}}2^{-2n}\\
    &\leq 2^n \cdot 2^{-2n} \leq 2^{-n}
\end{align*}
Now, we can divide these keys into sets with similiar sample entropies. For $j \in [0,2n-1]$, let $\mathbb{C}_j:= \{k \in K_s: j \leq \Hs_{K_s}(k) \leq j+1\}$. We now show that with overwhelming probability, $k$ sampled from $K_s$ will belong to some $\mathbb{C}_j$. 
\begin{align*}
    \Prr_{k\leftarrow K_s}[\exists j \in [0, 2n-1] \text{ s.t } k\in \mathbb{C}_j] &= \Prr_{k\leftarrow K_s}[0 \leq \Hs_{K_s}(k) \leq 2n]\\
    &\geq 1 - 2^{-n}
\end{align*}
Then the Pigeonhole Principle says there must exist some $j_s \in [0,2n-1]$ such that
\[
    \Prr_{k\leftarrow K_s}[k \in \mathbb{C}_{j_s}] \geq \frac{1-2^{-n}}{2n} \geq 1/3n
\]
Set $\mathbb{G}_s := \mathbb{C}_{j_s}$.  This completes the proof.
\end{proof}

For every $s \in \bbS$, fix $\mathbb{G}_s, j_s$ to be an arbitrary set and index satisfying the above claim. Set $i_s := j_s + 600 \log (n)$. Note that since $j_s \leq 2n$, $i_s \in [3n]$ for sufficiently large $n$.
Furthermore, let $$\mathbb{A}_s := \{k \in \bbK_s: \Hs_{K_s}(k) \leq j_s + 500\log(n)\}$$
We will now proceed to defining $\mathbb{F}_s$.\\

\noindent{\bf Defining $\mathbb{F}_s$.}
We now claim the existence of a dense enough set $\mathbb{F}_s$ of $h, y$ values such that (1) the preimage distribution of $h, y$ in $K_s$ (i.e. $K_s$ conditioned on $h(K_s)_{i_s} = y$) has reasonably high probability mass in $\mathbb{G}_s$ and (2) there is at most one key $k \in \mathbb{A}_s$ that hashes (under $h$) to $y$. Formally, we have the following claim.

\begin{claim}
\label{clm:hashSlice}
    For all $n\in\bbN$ and all $s\in \bbS$, there exists a set $\mathbb{F}_s$ such that
    \begin{enumerate}
        \item $\Prr_{\substack{h\sample\bin^\ell\\ k \leftarrow K_s}}[h, h(k)_{i_s} \in \mathbb{F}_s] \geq 1/6n - 1/n^{100}$
        \item For all $h,y \in \mathbb{F}_s$, the following holds:\\
        $\Pr_{k \leftarrow K_s}[k \in \mathbb{G}_s | h(k)_{i_s} = y] \geq 1/6n$ and there exists atmost one $k\in \mathbb{A}_s$ such that $h(k)_{i_s} = y$
    \end{enumerate}
\end{claim}
\begin{proof}
    We prove the claim in two parts. The following claim shows that with high probability over the randomness of sampling $(h, k)$, there exists atmost one $k'\in \mathbb{A}_s$ such that $h(k')_{i_s} = h(k)_{i_s}$.
    \begin{subclaim}
        For all $s \in \bbS$,
        \[
        \Prr_{\substack{h\sample\bin^\ell\\ k \leftarrow K_s}}[\exists k' \in \mathbb{A}_s \text{ s.t. } k' \neq k \text{ and } h(k')_{i_s} = h(k)_{i_s}]\leq 1/n^{100}
        \]
    \end{subclaim}
    \begin{proof}
        We first show a bound on the size of $\mathbb{A}_s$. By definition of $\mathbb{A}_s$, for all $k \in \mathbb{A}_s$, we have that
        $\Hs_{K_s}(k) \leq j_s + 500\log n$.
        This implies
        \begin{align*}
        &\Prr_{K_s}[k] \geq \frac{1}{2^{j_s + 500\log n}}\\
            \implies &\Prr_{k' \leftarrow K_s}[k' \in \mathbb{A}_s] \geq \frac{|\mathbb{A}_s|}{2^{j_s + 500\log n}}\\
            \implies &|\mathbb{A}_s| \leq 2^{j_s + 500\log n}
        \end{align*}
    By the properties of universal hash functions, for all $s \in \bbS$, for all $k, k'$ such that $k \neq k'$
    \[
        \Prr_{h\sample\bin^\ell}[h(k')_{i_s} = h(k)_{i_s}] = \frac{1}{2^{i_s}}
    \]
    Applying a union bound to the set $\mathbb{A}_s \setminus \{k\}$
    \[
        \Prr_{h\sample\bin^\ell}[\exists k' \in \mathbb{A}_s \setminus \{k\} \text{ s.t. } h(k')_{i_s} = h(k)_{i_s}] \leq \frac{|\mathbb{A}_s|}{2^{i_s}}
    \]
    Recall that $i_s$ is defined as $j_s + 600\log n$ and $|\mathbb{A}_s| \leq 2^{j_s + 500\log n}$. Substituting these above gives
    \[
        \Prr_{h\sample\bin^\ell}[\exists k' \in \mathbb{A}_s \setminus \{k\} \text{ s.t. } h(k')_{i_s} = h(k)_{i_s}] \leq \frac{1}{n^{100}}
    \]
    which proves this subclaim.
    \end{proof}
    In the next subclaim we show that for some noticeable fraction of $(h, h(k)_{i_s})$, the induced preimage distribution on keys has noticeable probability mass in $\mathbb{G}_s$.
    \begin{subclaim}
        For all $s \in \bbS$,
        \[
            \Prr_{\substack{k \leftarrow K_s\\ h \sample \bin^\ell}}\left[\Prr_{k' \leftarrow K_s}[k' \in \mathbb{G}_s| h(k')_{i_s} = h(k)_{i_s}] \geq 1/6n\right] \geq 1/6n
        \]
    \end{subclaim}
    \begin{proof}
        Consider the following inequality: For all $s \in \bbS$,
        \[
            \Prr_{\substack{k \leftarrow K_s\\ h \sample \bin^\ell \\k' \leftarrow K_s }}[k' \in \mathbb{G}_s| h(k')_{i_s} = h(k)_{i_s}] \geq 1/3n
        \]
        It suffices to prove the above inequality, since the statement of our subclaim follows from it by a Markov argument. Note that
        \[
            \Prr
            [
                k' \in \mathbb{G}_s| h(k')_{i_s} = h(k)_{i_s}
            ] 
            =
            \frac{
                \Prr
                [
                    \left(k' \in \mathbb{G}_s\right) \wedge \left(h(k')_{i_s} = h(k)_{i_s}\right)
                ]
            }{
                \Prr
                [
                    h(k')_{i_s} = h(k)_{i_s}
                ] 
            }
        \]
        where the probabilities are over $k \leftarrow K_s, h \sample \bin^\ell, k' \leftarrow K_s$.
        
         By applying the properties of universal hashing to the RHS,
        \begin{align*}
            \Prr
            [
                k' \in \mathbb{G}_s| h(k')_{i_s} = h(k)_{i_s}
            ] 
            &=
            \Prr_{
                k' \leftarrow K_s 
            }
            [
                k' \in \mathbb{G}_s
            ]
            \cdot
            \frac{
                1/2^{i_s}
            }{
                1/2^{i_s}
            }\\
            &=
            \Prr_{
                k' \leftarrow K_s 
            }
            [
                k' \in \mathbb{G}_s
            ]
        \end{align*}
        By Claim \ref{clm:flatSet} for all $s \in \bbS$,
        \[
        \Prr_{k \leftarrow K_s}[k \in \mathbb{G}_s]\geq 1/3n
        \]
        which gives
        \[
            \Prr[k' \in \mathbb{G}_s| h(k')_{i_s} = h(k)_{i_s}] \geq 1/3n
        \]
    \end{proof}
We now combine these subclaims to define $\mathbb{F}_s$. For any events $A$ and $B$, $\Pr[A \wedge B] \geq \Pr[A] - \Pr[\neg B]$. Applying this to the above two subclaims, for all $s \in \bbS$,
\begin{align*} 
    \Prr_{\substack{k \leftarrow K_s\\ h \sample \bin^\ell}}\left[\left(\Prr_{k' \leftarrow K_s}[k' \in \mathbb{G}_s| h(k')_{i_s} = h(k)_{i_s}] \geq 1/6n\right) \wedge\left(\forall k' \in \mathbb{A}_s \setminus \{k\}, h(k')_{i_s} \neq h(k)_{i_s}\right)  \right] \\ \geq 1/6n -  1/n^{100}
\end{align*}
which implies
\[   
    \Prr_{\substack{k \leftarrow K_s\\ h \sample \bin^\ell}}\left[
        \begin{array}{c}
             \Prr_{k' \leftarrow K_s}[k' \in \mathbb{G}_s| h(k')_{i_s} = h(k)_{i_s}] \geq 1/6n\\
             \wedge \left(\exists \text{ at most one } k' \in \mathbb{A}_s \text{ s.t. } h(k')_{i_s} = h(k)_{i_s}\right)
        \end{array}
    \right] \geq 1/6n -  1/n^{100}
\]
Define $\mathbb{F}_s$ as follows:
\[
    \mathbb{F}_s := \left\{h,y \left|
        \begin{array}{c}
             h\in\bin^\ell \wedge y \in \bin^{i_s} \wedge\\
             \Prr_{k \leftarrow K_s}[k \in \mathbb{G}_s| h(k)_{i_s} = y] \geq 1/6n\\
             \wedge \left(\exists \text{ at most one } k \in \mathbb{A}_s \text{ s.t. } h(k)_{i_s} = y\right)
        \end{array}\right.
    \right\}
\]
which implies 
\[
    \Prr_{\substack{k \leftarrow K_s\\ h \sample \bin^\ell}}[h,h(k)_{i_s}\in \mathbb{F}_s]\geq 1/6n -  1/n^{100}
\]
which concludes the proof of our claim.
\end{proof}
This completes a description of $\wgen$ and $\simwgen$.
It is easy to see that $\wgen(n)$ is efficiently sampleable and that $\wgen$ and $\simwgen$ are polynomially bounded in length, i.e. there exists some polynomial $p(\cdot) $ such that for all $z_0 \in \Supp(\wgen(n))$ and for all $z_1 \in \Supp(\simwgen(n))$, $|z_0|=| z_1| \leq p(n)$.
 Next, we prove that $\simwgen(n)$ has inverse-polynomially higher Shannon entropy  than $\wgen(n)$.

\subsection{Establishing an Entropy Gap}

\begin{lemma}
\label{lem:egmain}
For all sufficiently large $n \in \bbN$,
    \[
        \Hs(\simwgen(n)) - \Hs(\wgen(n)) \geq 1/n^5
    \]
\end{lemma}
The rest of this subsection is dedicated to a proof of the above lemma.  We will make use of some intermediate lemmas and claims, that we state below.

First, we state a 
standalone lemma described below, the proof of which is in Appendix \ref{app:fullproof}.
        Intuitively, the lemma defines a special distribution $X$ over $\{0,1\}^n$ that samples a single element $x^*$ with much higher probability than all other elements $x \neq x^*$. It then finds a lower bound on the difference in Shannon entropy between distributions $A_0$ and $A_1$, where $A_0 = R, \langle X, R \rangle$ and $A_1$ is identical to $A_0$ except replacing $\langle X, R \rangle$ with a uniform random bit when $X = x^*$ (and where $R$ is the uniform distribution on $\{0,1\}^n$). The proof makes use of the leftover hash lemma to show that the two distributions are close to uniform (and hence have high entropy) conditioned on $X \neq x^*$. Thus, a difference in entropies arises from the case of $X = x^*$.

                \begin{lemma}
        \label{lem:entropyGapCore}
            For sufficiently large $n \in \bbN$, let $X$ be a distribution over $\bin^n$ such that there exists $x^*\in \bin^n$ s.t.
            \begin{enumerate}
                \item $\Pr_{x\leftarrow X}[x=x^*] \geq 1/6n$
                \item $\forall x' \neq x^*$, $\Pr_{x \leftarrow X}[x=x']\leq 2/n^{600}$ 
            \end{enumerate}
            Let $R$ be uniformly distributed on $\bin^n$, and $U_1$ denote the uniform distribution over $\bin$. Define $\alpha_0$ and $\alpha_1$ as follows.
            \begin{gather*}
                \alpha_0(x,r) := \langle x, r \rangle\\
                \alpha_1(x,r) := \left\{
                \begin{array}{cl}
                    U_1 &\text{ if $x = x^*$}\\
                    \langle x, r \rangle &\text{ otherwise}
                \end{array}
                \right.
            \end{gather*}
            For $b\in \bin$, define distribution $A_b := R, \alpha_b(X,R)$. Then
            \[
            \Hs(A_1) - \Hs(A_0) \geq 1/100n^2
            \]
        \end{lemma}

To prove lemma \ref{lem:egmain}, we will also do the following. For each fixing of $s \in \bbS$, we will establish an entropy gap between the distributions $\wgen$ and $\simwgen$ when conditioned on sampling $(k, h, i)$ such that $i = i_s$ and $h, h(k)_{i_s} \in \mathbb{F}_s$.
    To do this, we will use the following claim, which helps establish a difference in entropy in the last bit of $\wgen$ and $\simwgen$ for every fixing of $s \in \bbS$ and $(h, y) \in \mathbb{F}_s$. 
    \begin{claim}
    \label{clm:constrainedEntropyGap}
        Fix any $s \in \mathbb{S}$ and any $(h,y) \in \mathbb{F}_{s}$. Let $R$ be uniformly distributed on $\bin^n$.

        For any $r\in\bin^n$, define distributions $P_r,Q_r$ as follows:
        First sample $k \leftarrow K_s$ conditioned on $h(k)_{i_s} = y$. Set
        \begin{equation*}
        P_r = 
    \left\{\begin{array}{cl}
        U_1 & \text{ if $k \in \mathbb{G}_{s}$}\\
        \langle k, r \rangle &\text{ otherwise}
    \end{array}
    \right.
\end{equation*}
        \begin{center}
            and
        \end{center}
        \begin{equation*}
        Q_r = 
        \begin{array}{cl}
        \langle k, r \rangle
    \end{array}
    \end{equation*}
    Then,
        \[
            \Hs(R, P_R) - \Hs(R, Q_R) \geq 1/100n^2
        \]
    \end{claim}

Before proceeding, we provide a complete proof of Claim \ref{clm:constrainedEntropyGap}.
\begin{proof} (of Claim \ref{clm:constrainedEntropyGap})
        We first require some claims about the preimage distribution on keys in $\mathbb{K}_s$ such that $h(k)_{i_s}=y$. 
        These will set us up to use Lemma \ref{lem:entropyGapCore}.
        
        First, we show that there is a unique key $k^* \in \mathbb{G}_s$ that hashes to $y$.
        \begin{subclaim}
            There exists a unique $k^* \in \mathbb{G}_s$ such that $h(k^*)_{i_s} = y$.
        \end{subclaim}
        \begin{proof}
            By definition of $\mathbb{F}_s$,
            $$\Prr_{k\leftarrow K_s}[k\in \mathbb{G}_s | h(k)_{i_s} = y] \geq 1/6n$$
            Thus, there must exist at least one $k^* \in \mathbb{G}_s$ such that $h(k^*)_{i_s} = y$.  
            
            By definition, $\mathbb{G}_s \subseteq \mathbb{A}_s$. By Claim \ref{clm:hashSlice}, there exists atmost one $k \in \mathbb{A}_s$ such that $h(k)_{i_s} = y$. Therefore, $k^*$ is unique.
        \end{proof}
        Next we show that since all other preimage keys are outside of $\mathbb{A}_s$, they have significantly lower sampling probabilities than $k^*$. 
        \begin{subclaim}
            For all $k'\in \bin^n$ such that $k' \neq k^*$
            \[
                \Prr_{k\leftarrow K_s}[k=k'|h(k)_{i_s}=y] \leq 2/n^{600}
            \]
        \end{subclaim}
        \begin{proof}
            By Claim \ref{clm:hashSlice}, there is atmost one $k\in \mathbb{A}_s$ such that $h(k)_{i_s} = y$. Note that $k^*\in \mathbb{G}_s \subseteq \mathbb{A}_s$. 
            
            Fix any $k' \neq k^*$ such that $h(k')_{i_s} = y$ (the subclaim follows trivially if $h(k')_{i_s} \neq y)$. It must therefore be the case that $k' \notin \mathbb{A}_s$. By definition of $\mathbb{A}_s$,
            $$\Hs_{K_s}(k') \geq j_s + 600\log n$$
            or equivalently,
 $$\Prr_{K_s}[k'] \leq \frac{1}{n^{600}\cdot2^{j_s}}$$

            We aim to bound $\Prr_{k\leftarrow K_s}[k=k'|h(k)_{i_s}=y]$. We do so by noting that the ratio of probabilities of $k^*$ and $k'$ is unaffected by conditioning on $h(k)_{i_s} = y$.
            \begin{gather*}
                \Prr_{k\leftarrow K_s}[k = k'|h(k)_{i_s}=y] = \frac{\Prr_{k\leftarrow K_s}[k = k']}{\Prr_{k\leftarrow K_s}[h(k)_{i_s}=y]}\\
                \Prr_{k\leftarrow K_s}[k = k^*|h(k)_{i_s}=y] = \frac{\Prr_{k\leftarrow K_s}[k = k^*]}{\Prr_{k\leftarrow K_s}[h(k)_{i_s}=y]}
            \end{gather*}
            which implies $$\frac{\Prr_{k\leftarrow K_s}[k = k'|h(k)_{i_s}=y]}{\Prr_{k\leftarrow K_s}[k = k^*|h(k)_{i_s}=y]} = \frac{\Prr_{k\leftarrow K_s}[k = k']}{\Prr_{k\leftarrow K_s}[k = k^*]}
            $$
            Since $k^*\in\mathbb{G}_s$, by Claim \ref{clm:flatSet}
            \begin{gather*}
                \Hs_{K_s}(k^*) \leq j_s + 1
                \implies \Prr_{K_s}[k^*] \geq 1/{2^{j_s+1}}
            \end{gather*}
        Together with the previous bound on the probability of $k'$ this implies
        \[
        \frac{\Prr_{k\leftarrow K_s}[k = k'|h(k)_{i_s}=y]}{\Prr_{k\leftarrow K_s}[k = k^*|h(k)_{i_s}=y]} = \frac{\Prr_{k\leftarrow K_s}[k = k']}{\Prr_{k\leftarrow K_s}[k = k^*]} \leq \frac{2^{j_s+1}}{n^{600}\cdot2^{j_s}} \leq \frac{2}{n^{600}}
        \]
        Since probabilities are upper bounded by one, this proves the subclaim.
        \end{proof}

        We may now apply Lemma \ref{lem:entropyGapCore} with $X := \left(K_s|h(K_s)_{i_s}=y\right)$ and $x^*:= k^*$. By Claim \ref{clm:hashSlice}
        \[
        \Prr_{k\leftarrow K_s}[k\in \mathbb{G}_s | h(k)_{i_s} = y] \geq 1/6n 
        \]
         Since $k^*$ is the only inverse in $\mathbb{G}_s$, this implies
        \[
         \Prr_{k\leftarrow K_s}[k=k^* | h(k)_{i_s} = y] \geq 1/6n
        \]
        Moreover, for all $k' \in \bin^n$ such that $k' \neq k^*$, $\Prr_{k\leftarrow K_s}[k=k'|h(k)_{i_s}=y] \leq 2/n^{600}$. Thus, Lemma \ref{lem:entropyGapCore} implies that 
        \[
            \Hs(R, P_R) - \Hs(R, Q_R) \geq 1/100n^2
        \]
        which proves Claim \ref{clm:constrainedEntropyGap}.
        \end{proof}
The next claim relates the entropy gap between arbitrary joint distributions $(A, B_0)$ and $(A, B_1)$ to their entropy gap conditioned on an event.
\begin{claim}
\label{clm:publicSlicing}
    Let $A$, $B_0$, and $B_1$ be random  variables and let $\bbA := \Supp(A)$. Additionally, let $\bbA^*\subseteq \bbA$ be a set with the following properties:
    \begin{enumerate}
        \item There exists some $d\geq 0$ such that for all $a \in \bbA^*, \Hs(B_1|A=a) -  \Hs(B_0|A=a) \geq d$.
        \item For all $a\in \bbA \setminus \bbA^*, \Hs(B_1|A=a) = \Hs(B_0|A=a)$
    \end{enumerate}
    Then the following holds
    \[
    \Hs(A,B_1) - \Hs(A, B_0) \geq d\cdot \Prr[A \in \bbA^*]
    \]
\end{claim}
\begin{proof}
    By the chain rule, for $b\in \bin$
    \[
    \Hs(A,B_b) = \Hs(B_b|A) + \Hs(A)
    \]
    Applying this to the difference in entropies
    \begin{align*}
    \Hs(A,B_1) - \Hs(A, B_0) &= \Hs(B_1|A) + \Hs(A) - \Hs(B_0|A) -\Hs(A)\\
    &= \Hs(B_1|A)- \Hs(B_0|A)\\
    &= \sum_{a\in \bbA}\Prr_A[a]\Big(\Hs(B_1|A=a)- \Hs(B_0|A=a)\Big)
    \end{align*}
We may now split the sum into terms where $a\in \bbA^*$ and terms where $a\in \bbA\setminus\bbA^*$, and apply the properties of $\bbA^*$. 
\begin{align*}
    \Hs(A,B_1) - \Hs(A, B_0)
    =& \sum_{a\in \bbA^*}\Prr_A[a]\Big(\Hs(B_1|A=a)- \Hs(B_0|A=a)\Big)\\
    &+ \sum_{a\in \bbA \setminus \bbA^*}\Prr_A[a]\Big(\Hs(B_1|A=a)- \Hs(B_0|A=a)\Big)\\
    =& \sum_{a\in \bbA^*}\Prr_A[a]\Big(\Hs(B_1|A=a)- \Hs(B_0|A=a)\Big)\\
    \geq& \sum_{a\in \bbA^*}\Prr_A[a] \cdot d\\
    =&d\cdot\Prr[A\in \bbA^*]
    \end{align*}
    which concludes the proof.
\end{proof}
We are now ready to prove the main lemma of this subsection.
\begin{proof} (of Lemma \ref{lem:egmain})
Let $H$ be uniformly distributed over $\bin^\ell$, $I$ be uniformly distributed over $[3n]$, and $R$ be uniformly distributed over $\bin^n$. 
Recall that 
$$\wgen = (S_K, H, I, R, H(K)_I, Q_{K, R})$$
and
$$\simwgen = (S_K, H, I, R, H(K)_I, P_{K, R})$$
where $$Q_{K,R} := \langle K, R \rangle$$
and \begin{equation*}
        P_{K, R} = 
    \left\{\begin{array}{cl}
        U_1 & \text{ if $K \in \mathbb{G}_{S_K}$
        and $I = i_{S_K}$ and $h,h(K)_I \in \mathbb{F}_{S_K}$}\\
        \langle K, R \rangle &\text{ otherwise}
    \end{array}
    \right.
\end{equation*}
First, we note that for any $s\in \bbS$ and any $(h,y)\in \bbF_s$, the distribution $(R, P_{K,R}|S_K=s, H=h, H(K)_I = y, I=i_s)$ is identical to the distribution $(R,P_R)$ as defined in Claim \ref{clm:constrainedEntropyGap}. Similarly, $(R, Q_{K,R}|S_K=s, H=h, H(K)_I = y, I=i_s)$ is identical to the distribution $(R,Q_R)$ as defined in Claim \ref{clm:constrainedEntropyGap}. Applying the claim therefore gives
\begin{multline*}
    \Hs(R, P_{K,R}|S_K=s, H=h, H(K)_I = y, I=i_s) \\- \Hs(R, Q_{K,R}|S_K=s, H=h, H(K)_I = y, I=i_s) \geq 1/100n^2
\end{multline*}
Additionally, note that if $H,H(K)_I \notin \bbF_S$ or $I\neq i_S$, the distributions are identical. We now apply Claim $\ref{clm:publicSlicing}$. Setting $A$ to be $(S_K, H, H(K)_I, I)$, $B_1$ to $(R, P_{K,R})$, $B_0$ to $(R, Q_{K,R})$, and  $\bbA^*$ to the set $\{(s, h, y, i_s):s\in \bbS \wedge (h,y) \in \bbF_s\}$, we get
\begin{multline*}
    \Hs(S_K, H, I,R, H(K)_I, P_{K,R}) - \Hs(S_K, H, I,R, H(K)_I, Q_{K,R}) \geq \frac{\Prr[I = i_S \wedge (H, H(K)_I) \in \bbF_S]}{100n^2}
\end{multline*}
which may be written as
\begin{align*}
    \Hs(\wgen)- \Hs(\simwgen) &\geq \Prr[I = i_S \wedge (H, H(K)_I) \in \bbF_S]/100n^2
\end{align*}
Since $I$ is uniform over $[3n]$ and $i_S\in [3n]$,
\begin{align*}
    \Hs(\wgen)- \Hs(\simwgen) &\geq \Prr[(H, H(K)_{i_S}) \in \bbF_S]/300n^3
\end{align*}
By Claim \ref{clm:hashSlice}, for every $s\in \bbS$, $\Prr[H, H(K)_{i_s} \in \mathbb{F}_s] \geq 1/6n - 1/n^{100}$. Therefore for large enough $n \in \bbN$
\begin{align*}
    \Hs(\wgen)- \Hs(\simwgen) &\geq (1/6n - 1/n^{100})\cdot(1/300n^3)\\
    &\geq (1/7n)\cdot (1/300n^3)\\
        &\geq \frac{1}{2100n^4} > \frac{1}{n^5}
\end{align*}
which concludes the proof.
\end{proof}

\subsection{Establishing Computational Indistinguishability}

\begin{lemma}
There exists a negligible function $\mu$ such that for every quantum polynomial sized adversary $\cA$, for all $n \in \bbN$:
        \[
            \left|\Prr_{z\leftarrow \wgen(n)}[\cA(z) = 1] - \Prr_{z\leftarrow \simwgen(n)}[\cA(z) = 1]\right| \leq \mu(n)
        \]
\end{lemma}
\begin{proof}
    We prove this lemma by contradiction. 
    Suppose there exists some quantum polynomial-sized $\cA$ and an inverse polynomial function $\varepsilon$ such that for infinitely many $n \in \mathbb{N}$,
    \[
        \left|\Prr_{z\leftarrow \wgen(n)}[\cA(z) = 1] - \Prr_{z\leftarrow \simwgen(n)}[\cA(z) = 1]\right| \geq \varepsilon(n)
    \]
    For every $s$, $\wgen$ and $\simwgen$ are identical whenever $i \neq i_s$ or $k \not\in \mathbb{G}_{s}$. As a result, the case where this constraint is met must give rise to all the distinguishing advantage. By definition of $\wgen$ and $\simwgen$, this means that
    \begin{align}
        &\Big| \Prr%
         \left[\cA(s,h,i_s,r, h(k)_{i_s},\langle k, r\rangle) = 1\ \middle| \begin{array}{r}
         k \in \mathbb{G}_s 
        \end{array}\right] \Big. \nonumber \\
        & -
        \Prr%
        \Big. \left[\cA(s,h,i_s,r,h(k)_{i_s}, b) = 1\ \middle| \begin{array}{r}
         k \in \mathbb{G}_s 
        \end{array}\right] \Big| \nonumber \\
        & \geq \varepsilon(n)
    \end{align}
    where the probability is over $k,s \leftarrow (K,S_K), b\sample\bin, h\sample\bin^\ell, r\sample\bin^n$.

    Let $\cA'$ be a {\em predictor} that on input $(s,h,i,r,y)$ samples a uniform bit $b$ and outputs $b \oplus \cA(s,h,i,r,y,b)$. Then the previous equation implies:
    \begin{equation}
        \label{eq:aprime}
        \Prr \left[\cA'(s,h, i_s, r, h(k)_{i_s}) = \langle k,r \rangle\ \middle| \begin{array}{r}
         k \in \mathbb{G}_s 
        \end{array}\right] \geq 1/2 + \varepsilon(n)
    \end{equation}

Note that $\cA'$ is a quantum polynomial-sized circuit family that guesses $\langle k, r \rangle$ with noticeable probability, given $(s,h, i_s, h(k)_{i_s})$ as an auxiliary function of the key $k$, along with $r$, and when conditioned on $k \in \mathbb{G}_s$.
Therefore, by the Goldreich-Levin theorem (Theorem \ref{thm:gl}) applied to $\cA'$,
    there exists a quantum polynomial-sized circuit family $\cB$ and an inverse polynomial function $\varepsilon'$ such that:
    \begin{equation}
    \label{eq:lkl}
    \Prr \left[\cB(s,h, i_s,h(k)_{i_s}) = k\ \middle| \begin{array}{r}
         k \in \mathbb{G}_s 
        \end{array}\right] \geq \varepsilon'(n)
    \end{equation}

    Our goal is to prove the existence of a quantum polynomial-sized algorithm that outputs inverses $k$ given a puzzle $s$.
    In the above equation, however, $\cB$ requires $(h, i_s h(k)_{i_s})$ as input in addition to the puzzle $s$. The next claim proves that $\cB$ will output an inverse of the puzzle with non-negligible probability even when $i_s$ and $h(k)_{i_s}$ are replaced with uniform strings. 

\begin{claim}
\label{clm:ff}
There exists an inverse polynomial function $\varepsilon''$ such that for infinitely many $n \in \mathbb{N}$:
\[
\Prr\left[\top\leftarrow\Ver(\cB(s,h,i,u),s)\right] \geq \varepsilon''(n)
    \]
    where the probability is over 
    $s \leftarrow S,i \sample[3n], h\sample\bin^\ell,u\sample \bin^{i}$.
\end{claim}

\begin{proof}
To prove this claim, we will first define a statistical test using the input-output behavior of $\cB$.

\begin{subclaim}
\label{clm:abv}
Define test $\cT$ as follows.
    \[
        \cT(s,h,y) := \left\{\begin{array}{cl}
             1 & \text{if }\Prr[\top\leftarrow\Ver(\cB(s,h, i_s,y),s)] \geq \varepsilon'(n)/4 \\
             0 & \text{otherwise}
        \end{array}\right.
    \]
where the probability is over the randomness of $\cB$ and $\mathsf{Ver}$.
Then,
\[
    \Prr_{(k,s) \leftarrow (K,S), h \leftarrow \bin^\ell} \left[\cT(s,h,h(k)_{i_s})=1 \middle| 
         k \in \mathbb{G}_s\right] \geq \varepsilon'(n)/4
    \]
\end{subclaim}
\begin{proof} (of SubClaim \ref{clm:abv})
     For any events $A$ and $B$, $\Pr[A \wedge B] \geq \Pr[A] - \Pr[\neg B]$. Applying this to the correctness of the puzzle
    \begin{align*}
    &\Prr_{\substack{k,s \leftarrow (K,S_K)\\ h\sample\bin^\ell}}\left[\cB(s,h, i_s,h(k)_{i_s}) = k\ \wedge\ \top\leftarrow\Ver(k,s)\middle| 
         k \in \mathbb{G}_s\right]\\ &\geq \varepsilon'(n) - \Prr_{k,s \leftarrow (K,S_K)}[\bot\leftarrow\Ver(k,s)| k\in\mathbb{G}_s]\\
         &\geq \varepsilon'(n) - \frac{\Prr_{k,s \leftarrow (K,S_K)}[\bot\leftarrow\Ver(k,s) ]}{\Prr_{k,s \leftarrow (K,S_K)}[k\in\mathbb{G}_s]}\\
         &= \varepsilon'(n) - 3n \cdot \negl(n) > \varepsilon' (n)/2
    \end{align*}
    where the last step follows due to Claim \ref{clm:flatSet}, and due to the correctness of the one-way puzzle. 
    Thus,
    \[
    \Prr \left[\top\leftarrow\Ver(\cB(s,h, i_s,h(k)_{i_s}),s)\middle| 
         k \in \mathbb{G}_s\right] \geq \varepsilon'(n)/2
    \]
    and by a Markov argument
    \[
    \Prr \left[\Prr[\top\leftarrow\Ver(\cB(s,h, i_s,h(k)_{i_s}),s)] \geq \varepsilon'(n)/4 \middle| 
         k \in \mathbb{G}_s\right] \geq \varepsilon'(n)/4
    \]
    which proves this subclaim.
    \end{proof}
    
    Let $\ell^*:= 600\log n + \log 3n + 2\log(8/\varepsilon'(n))$.
    Fix any $s \in \bbS$, and define $\ell^*_s = \min(i_s, \ell^*)$.
    Let $K'_s$ denote the distribution of the random variable $(K_s|K_s \in \bbG_s)$.
    In what follows, we will bound the distance between $h(k)_{{i_s - \ell^*_s}}$ and uniform.
    \begin{subclaim}
    \label{clm:bv}
     For every $s \in \bbS$,
        \[
            \SD\Big(
            (H,H(K'_s)_{i_s - \ell^*_s}), (H,U_{i_s - \ell^*_s}) \Big) \leq \varepsilon'(n)/8
        \]
    \end{subclaim}
    \begin{proof}
        For all $s \in \bbS$, for all $k^* \in \mathbb{G}_s$,
        \begin{align*}
            \Prr_{k\leftarrow K'_s}[k = k^*] &= \Prr_{k\leftarrow K_s}[k= k^* | k\in\mathbb{G}_s]\\
            &= \Prr_{k\leftarrow K_s}[k=k^*]/\Prr_{k\leftarrow K_s}[k\in\mathbb{G}_s]
        \end{align*}
        By Claim \ref{clm:flatSet}, for all $s \in \bbS$, $\Prr_{k\leftarrow K_s}[k \in \mathbb{G}_s] \geq 1/3n$. 
        Additionally, for all $k^* \in \mathbb{G}_s$,
       $\Hs_{K_s}(k^*) \geq j_s
            \implies \Prr_{K_s}[k^*]\leq 2^{-j_s}
        $.

        Therefore for all $s \in \bbS$, for all $k^* \in \mathbb{G}_s$,
        \begin{align*}
            \Prr_{k\leftarrow K'_s}[k=k^*] &= \Prr_{k\leftarrow K_s}[k=k^*]/\Prr_{k\leftarrow K_s}[k\in\mathbb{G}_s]\\
            &\leq 2^{-j_s}\cdot 3n
        \end{align*}
         Thus for $s \in \bbS$, for $k^* \in \mathbb{G}_s$, $\Hs_{K'_s}(k^*) \geq j_s - \log3n$. In other words, $\Hs_{\min}(K'_s) \geq j_s - \log3n$.

        By the Leftover Hash Lemma (Theorem \ref{thm:LHLforSmoothEntropy}),
        \begin{align*}
        \SD\Big(
            (H,H(K'_s)_{i_s - \ell^*_s}), (H,U_{i_s - \ell^*_s}) \Big)  
            &\leq 2^{- 0.5( \Hs_{\min}(K'_s) - i_s + \ell^*_s)}\\
        &= {2^{-0.5(j_s - \log3n - i_s + 600\log n + \log 3n + 2\log(8/\varepsilon'(n)))}}\\
        &= {2^{-\log(8/\varepsilon'(n))}} = {\varepsilon'(n)}/8
        \end{align*}
        which proves the claim.
    \end{proof}
\noindent Next, relaxing the probability bound from SubClaim \ref{clm:abv} gives us
    \[
    \Prr \left[\exists \gamma \in \bin^{\ell^*_s} \text{ s.t }\cT(s,h,(h(k)_{i_s - \ell^*_s}\Vert\gamma))=1 \middle| 
         k \in \mathbb{G}_s\right] \geq \varepsilon'(n)/4
    \]
    Defining $\cT'$ as:
    \[
        \cT'(s,h,y) := \left\{\begin{array}{cl}
             1 & \text{if }\exists \gamma \in \bin^{\ell^*_s} \text{ s.t }\cT(s,h,(y\Vert\gamma))=1 \\
             0 & \text{otherwise}
        \end{array}\right.
    \]
    we have
    \[
    \Prr_{(k,s) \leftarrow (K,S_K), h \leftarrow \bin^\ell} \left[\cT'_{\ell^*}(s,h,h(k)_{i_s - \ell^*_s})=1 \middle| 
         k \in \mathbb{G}_s\right] \geq \varepsilon'(n)/4
    \]
    Since no statistical test can distinguish between two distributions with advantage better than their statistical distance, replacing the hash with a uniform string in the equation above, and invoking SubClaim \ref{clm:bv} gives us
    \[
        \Prr_{\substack{k,s \leftarrow (K,S_K)\\ h\sample\bin^\ell,u\sample \bin^{i_s - \ell^*_s}}}\left[\cT'(s,h,u)=1 \middle| k \in \mathbb{G}_s\right] \geq \frac{\varepsilon'(n)}{4} - \frac{\varepsilon'(n)}{8} = \frac{\varepsilon'(n)}{8}
    \]
    which by definition of $\cT'$ implies
    \[
    \Prr_{\substack{k,s \leftarrow (K,S_K)\\ h\sample\bin^\ell, u\sample \bin^{i_s - \ell^*_s}}}\left[\exists \gamma \in \bin^{\ell^*_s} \text{ s.t }\cT(s,h,u\Vert\gamma))=1 \middle| 
         k \in \mathbb{G}_s\right] \geq \frac{\varepsilon'(n)}{8}
    \]
    If we sample $\gamma$ randomly as well, we obtain
    \[
    \Prr_{\substack{k,s \leftarrow (K,S_K)\\ h\sample\bin^\ell, u\sample \bin^{i_s}}}\left[\cT(s,h,u))=1 \middle| 
         k \in \mathbb{G}_s\right] \geq \frac{\varepsilon'(n)}{8} \cdot 2^{-\ell^*}
    \]
    We may now relax the constraint that $k\in \mathbb{G}_s$ as follows:
    \[
    \Prr_{\substack{k,s \leftarrow (K,S_K)\\ h\sample\bin^\ell,u\sample \bin^{i_s}}}\left[\cT(s,h,u))=1 \right] \geq \frac{\varepsilon'(n)}{8} \cdot 2^{-\ell^*} \cdot \Prr_{k,s\leftarrow K, S_K}[k\in \mathbb{G}_s] \geq \frac{\varepsilon'(n)}{24n} \cdot 2^{-\ell^*}
    \]
    where the last inequality follows from Claim \ref{clm:flatSet}.

    Substituting $\ell^*:= 600\log n + \log 3n + 2\log(8/\varepsilon'(n))$, we obtain
    \[
        \Prr_{\substack{k,s \leftarrow (K,S_K)\\ h\sample\bin^\ell, u\sample \bin^{i_s}}}\left[\cT(s,h,u))=1 \right] \geq \frac{(\varepsilon'(n))^3}{24n \cdot n^{600} \cdot 3n \cdot 64} \geq \frac{(\varepsilon'(n))^3}{4608n^{602}} := \epsilon(n)
    \]
    By the definition of $\cT$, this means
    \[
        \Prr_{\substack{k,s \leftarrow (K,S_K)\\ h\sample\bin^\ell,u\sample \bin^{i_s}}}\left[\Prr[\top\leftarrow\Ver(\cB(s,h,i_s,u),s)] \geq \varepsilon'(n)/4 \right] \geq \epsilon(n)
    \]
    which after undoing the Markov argument implies
    \[
        \Prr_{\substack{k,s \leftarrow (K,S_K)\\ h\sample\bin^\ell,u\sample \bin^{i_s}}}\left[\top\leftarrow\Ver(\cB(s,h,i_s,u),s)\right] \geq \frac{\varepsilon'(n) \epsilon(n)}{4}
    \]
    When $i$ is sampled uniformly from $[3n]$, $i=i_s$ occurs with probability $1/3n$. Therefore
    \[
        \Prr_{\substack{k,s \leftarrow (K,S_K)\\ h\sample\bin^\ell\\i\leftarrow[3n], u\sample \bin^{i}}}\left[\top\leftarrow\Ver(\cB(s,h,i,u),s)\right] \geq \frac{\varepsilon'(n) \epsilon(n)}{12n} := \varepsilon''(n)
    \]
    which completes the proof of the claim.
    \end{proof}
\noindent    To complete the proof of the lemma, let $\cB'$ denote an algorithm that on input $s$ does the following.
    \begin{itemize}
        \item Sample $h \sample \bin^\ell$
        \item Sample $i \sample [3n]$
        \item Sample $u \sample \bin^{i}$
        \item Output $\cB(s,h,i,u)$
    \end{itemize}
    Then Claim \ref{clm:ff} implies
    \[
        \Prr_{(k,s)\leftarrow\Samp(1^n)}\left[ \top \leftarrow\Ver(\cB'(s),s)\right] \geq \varepsilon''(n)
    \]
    which contradicts security of the puzzle, as desired.
\end{proof}
\section{Quantum Weak PEGs imply Quantum PEGs}
\label{sec:product}
In this section, we show that a parallel repetition of quantum weak PEGs yields a strong pseudoentropy property, which we formalize into a quantum PEG, defined below.

\begin{definition}[Quantum Pseudoentropy Generator]
\label{def:qpeg}
    A Quantum Pseudoentropy Generator consists of an ensemble of distributions $\{\sgen(n), \simsgen(n)\}_{n\in\bbN}$ over classical strings such that:
    \begin{itemize}
     	\item \textbf{Bounded Length}. There exists a polynomial $p(\cdot)$ such that for all $n \in \bbN$, for all $z_0 \in \Supp(\sgen(n))$, for all $z_1 \in \Supp(\simsgen(n))$, $|z_0|=| z_1| \leq p(n)$.
        \item \textbf{Efficiency.} There exists some QPT algorithm that for all $n \in \bbN$, on input $1^n$, returns a sample from $\sgen(n)$.
        \item \textbf{Indistinguishability.} There exists a negligible function $\mu$ such that for all quantum polynomial-sized adversaries $\cA$, for all large enough $n \in \bbN$,
        \[
            \left|\Prr_{z\leftarrow \sgen(n)}[\cA(z) = 1] - \Prr_{z\leftarrow \simsgen(n)}[\cA(z) = 1]\right| \leq \mu(n)
        \]
        \item \textbf{Entropy Gap}. Here, we work with min and max entropies, as opposed to Shannon entropy. We require the min-entropy of $\simsgen$ to be higher than the max-entropy of $\sgen$. Formally, there is some explicit constant $c>0$ and some negligible function $\epsilon$ such that for all sufficiently large $n \in \bbN$,
        \[
            \Hs^{\epsilon(n)}_{\min}(\simsgen(n)) - \Hs^{\epsilon(n)}_{\max}(\sgen(n)) \geq n^c
        \]
    \end{itemize}
\end{definition}

\begin{theorem}
\label{thm:QWPEGtoQPEG}
    Quantum weak pseudoentropy generators  (Definition \ref{def:qwpeg}) imply quantum pseudoentropy generators (Definition \ref{def:qpeg}).
\end{theorem}
Let $\{\wgen(n), \simwgen(n)\}_{n\in\bbN}$ be a Quantum Weak Pseudoentropy Generator with entropy gap greater than $1/n^c$ for some constant $c>0$. By Definition \ref{def:qwpeg}, there exists a polynomial $\len(\cdot)$ such that for all $z_0 \in \Supp(\wgen(n))$ and $ z_1 \in \Supp(\simwgen(n))$, $|z_0|=| z_1| \leq \len(n)$. 
Let $q(n) := n^{c+3}\len(n)^2$.

We construct a pseudoentropy generator $\{\sgen(n), \simsgen(n)\}_{n\in\bbN}$ by simply generating $q(n)$ independent samples of $\wgen(n)$ and $\simwgen(n)$ respectively.
That is,
\begin{itemize}
    \item $\sgen(n):= \wgen(n)^{q(n)}$
        \item $\simsgen(n):= \simwgen(n)^{q(n)}$
\end{itemize}
\begin{claim}
    The distribution ensembles $\{\sgen(n)\}_{n\in\bbN}$ and $\{\simsgen(n)\}_{n\in\bbN}$ satisfy Definition \ref{def:qpeg}.
\end{claim}
Note that for all $z_0 \in \Supp(\sgen(n))$ and $z_1 \in \Supp(\simsgen(n))$, $|z_0|=| z_1| \leq q(n)\cdot \len(n)$ and $\sgen(n)$ is efficiently sampleable since $\wgen(n)$ is efficiently sampleable. Only indistinguishability and the entropy gap remain to be shown.\\

\noindent\textbf{Computational Indistinguishability.} There exists a negligible function $\mu_2$ such that for all quantum polynomial-sized adversaries $\cA$, $n \in \mathbb{N}$:
$$\left|\Prr_{z \leftarrow \sgen(n)}\left[1\leftarrow\cA(z)\right] - \Prr_{z \leftarrow \simsgen(n)}\left[1\leftarrow\cA(z)\right]\right| \leq \mu_2(n)$$   
\begin{proof}
    We prove this by contradiction. Suppose there exists an adversary $\cA$ that distinguishes the two distributions with non-negligible advantage $\epsilon$. We use this adversary to build a (non-uniform) reduction to computational indistinguishability of $\wgen$ and $\simwgen$, contradicting the security of the QWPEG. We proceed through a series of hybrids $H_0, H_1, \ldots, H_{q(n)}$. In hybrid $H_j$, the adversary is run on inputs sampled from $D_j := \wgen(n)^{j}\times \simwgen(n)^{q(n)-j}$. By definition, $D_0 = \sgen$ and $D_{q(n)} - \simsgen$ Suppose the adversary distinguishes adjacent hybrids $H_{j-1}$ and $H_{j}$ with advantage $\epsilon_j$. Note that by assumption, the adversary distinguishes $H_0$ and $H_q$ with advantage $\epsilon$ which implies that $\epsilon \leq \sum_j \epsilon_j$. Therefore, for each sufficiently large $n$, there exists a $j$ such that $\epsilon_j \geq \epsilon/q(n)$, i.e:
    \begin{align*}
    	\left|\Prr_{z \leftarrow D_{j-1}}\left[1\leftarrow\cA( z)\right] - \Prr_{z \leftarrow D_j}\left[1\leftarrow\cA( z)\right]\right| \geq \epsilon/q
    \end{align*}
    
    This adversary therefore distinguishes between samples from $\wgen$ and $\simwgen$ at the $j$th position, given appropriate samples at the other positions. While samples can be computed from $\wgen$ efficiently, computing $\simwgen$ is not necessarily efficient. To resolve this we show that for each sufficiently large $n$ there must exist fixed samples for which the adversary distinguishes in the $j$th position, and which the reduction may receive as non-uniform advice. If we parse $z$ as $z_1,\ldots,z_{q(n)}$, by a Markov argument:
    \begin{align*}
    	\Prr_{\substack{z^*_1, \ldots, z^*_{q(n)}} \leftarrow D_j}\left[
    	 \left|\begin{aligned}
    		&\Prr_{z_j \leftarrow \wgen(n)}\left[1\leftarrow\cA(z^*_1, \ldots,  z^*_{j-1}, z_j, z^*_{j+1}, \ldots, z^*_{q(n)})\right]\\
    		&-\Prr_{z_j \leftarrow \simwgen(n)}\left[1\leftarrow\cA(z^*_1, \ldots,  z^*_{j-1}, z_j, z^*_{j+1}, \ldots, z^*_{q(n)})\right]
    	\end{aligned}\right| \geq \epsilon/{2q}\right] \geq \epsilon/{2q}
        \end{align*}
Therefore, for each sufficiently large $n$ there exist $z^*_1, \ldots, z^*_{j-1}, z^*_{j+1}, \ldots z^*_q \in \Supp(\wgen(n)) \cup \Supp(\simwgen(n))$ such that:
 \begin{align*}
	\left|\begin{aligned}
    		&\Prr_{z_j \leftarrow \wgen(n)}\left[1\leftarrow\cA(z^*_1, \ldots,  z^*_{j-1}, z_j, z^*_{j+1}, \ldots, z^*_{q(n)})\right]\\
    		&-\Prr_{z_j \leftarrow \simwgen(n)}\left[1\leftarrow\cA(z^*_1, \ldots,  z^*_{j-1}, z_j, z^*_{j+1}, \ldots, z^*_{q(n)})\right]
    	\end{aligned}\right| \geq \epsilon/{2q}
\end{align*}
 We may now build a reduction $R$ that receives as input $z$ and takes as non-uniform advice $\tau$ defined as:
\[
\tau := j, z^*_1, \ldots, z^*_{j-1}, z^*_{j+1}, \ldots z^*_q
\]
The reduction $R$ on input $z, \tau$ outputs $\cA(z^*_1, \ldots,  z^*_{j-1}, z, z^*_{j+1}, \ldots, z^*_{q(n)})$
and achieves a distinguishing advantage of atleast $\epsilon/2q$, contradicting the security of QWPEG.
\end{proof}
\noindent\textbf{Entropy Gap.} There exists a constant $c'$ such that for $\epsilon(n) = 1/2^{n}$ and sufficiently large $n$:
     $$ \Hs_{\min}^{\epsilon} \left(\simsgen(n)\right) - \Hs^\epsilon_{\max}\left(\sgen(n)\right) \geq n^{c'}$$
\begin{proof}
    Setting $t = q(n), \epsilon = 2^{-n}$ and $\cU = 2^{|\simwgen|} \leq 2^{\len(n)}$, we may apply Theorem \ref{thm:conc} to obtain:
    \[
    \Hs^{\epsilon}_{\min}(\simsgen(n)) \geq  q(n)\cdot \Hs\left(\simwgen(n)\right) - O(\sqrt{qn\len^2})
    \]
Since $\Hs(\simwgen(n)) - \Hs(\wgen(n)) \geq 1/n^c$
\[
    \Hs^{\epsilon}_{\min}(\simsgen(n)) \geq  q(n)\cdot \Hs\left(\wgen(n)\right) + q(n)/n^c - O(\sqrt{qn\len^2})
    \]
    Similarly setting $t = q(n), \epsilon = 2^{-n}$ and $\cU = 2^{|\wgen|} \leq 2^{\len(n)}$, we may apply Theorem \ref{thm:conc} to obtain:
    \[
    \Hs^\epsilon_{\max}(\sgen(n)) \leq  q(n)\cdot \Hs\left(\wgen(n)\right) + O(\sqrt{qn\len^2})
    \]
Taking the difference gives us
\begin{gather*}
    \Hs^{\epsilon}_{\min}(\simsgen(n)) -  \Hs^{\epsilon}_{\max}\left(\sgen(n)\right) \geq q(n)/n^c - O(\sqrt{qn\len^2})
\end{gather*}
Substituting $q(n) = n^{2c+3}\len(n)^2$
\begin{align*}
    \Hs^{\epsilon}_{\min}(\simsgen(n)) -  \Hs^{\epsilon}_{\max}\left(\sgen(n)\right) &\geq n^{c+3}\len(n)^2 - O(n^{c+2}\len(n)^2)\\
    &\geq n^{c+3}\len(n)^2/2
\end{align*}
which implies there exists some constant $c'>0$ such that for sufficiently large $n$
\[
    \Hs^{\epsilon}_{\min}(\simsgen(n)) -  \Hs^{\epsilon}_{\max}\left(\sgen(n)\right)\geq n^{c'}
\]
which concludes the proof.
\end{proof}

\section{Quantum PEGs Imply \Imbalanced \EFI}
\label{sec:imbefi}
In this section, we prove that quantum pseudoentropy generators imply a (non-uniform) variant of EFI, that we define below. 
This definition modifies the standard definition (\defref{efi}) of EFI to allow the algorithm to depend on a nonuniform parameter $s$, and requires the existence of a function $s^*(n)$ such that
computational indistinguishability (resp. statistical distance) holds when $s \leq s^*(n)$ (resp. $s \geq s^*(n)$).

\begin{definition}[$s^*$-\Imbalanced EFI]
\label{def:nuefi}
Let $s^*(\cdot)$ denote a function.
An $s^*$-non-uniform $\EFI$ is a QPT algorithm $\EFI_{\prm}(1^\n,b)\rightarrow \rho_b$ that obtains classical parameter-dependent advice string $\prm$, and on input $b \in \bin$ and security parameter $\n$, outputs a (potentially mixed) quantum state such that:
\begin{enumerate}
\item {\bf Computational Indistinguishability.} There exists a negligible function $\mu(\cdot)$ such that for all quantum polynomial-sized circuits $\cA$,  for
large enough $n \in \mathbb{N}$ and every $\prm \leq \prmlimit(n)$,
\[
    \left| \Pr[1\leftarrow \cA(\EFI_s(1^\n, 0))] - \Pr[1\leftarrow \cA(\EFI_s(1^\n, 1))]\right| \leq \mu(\n)
\]
\item  {\bf Statistical Distance.} There exists a negligible function $\delta(\cdot)$ such 
for large enough $n \in \mathbb{N}$ and every $\prm \geq \prmlimit(n)$, 
    $$\TD(\EFI_{\prm}(1^n, 0), \EFI_{\prm}(1^n, 1)) \geq 1 - \delta(n)$$
\end{enumerate}
\end{definition}

\begin{theorem}
\label{thm:QPEGtoEFI}
There exists a function $s^*(\cdot)$ such that
\begin{itemize}
\item There exists a polynomial $p$ such that for all $n \in \bbN$, $|s^*(n)| \leq \log p(n)$ and
\item Quantum pseudoentropy generators (Definition \ref{def:qpeg}) imply $s^*$-\imbalanced \EFI (Definition \ref{def:nuefi}). 
\end{itemize}
\end{theorem} 
Let $\{\sgen(n), \simsgen(n)\}_{n\in\bbN}$ be a QPEG (\defref{qpeg}). As a result, there exists a constant $c$ and some negligible smoothing parameter $\epsilon$ such that the entropy gap is greater than $n^c$. Additionally, there exists a polynomial $\len(\cdot)$ such that the length of the generator output is upper bounded by $\len(n)$. 
Finally, let $\ell$ be the key length of a universal hash function from $\bin^{\len(n)}$ to $\bin^{\len(n)}$. Let $h$ be sampled uniformly from $\bin^\ell$. We define the EFI distributions as follows.
\begin{itemize}
    \item $\EFI_{\prm}(1^n, 0) := h, h(\sgen)_{\prm}$
    \item $\EFI_{\prm}(1^n, 1) := h, U_{\prm}$
\end{itemize}
We define $\prmlimit(n) :=  \Hs_{\max}(\sgen) + \frac{n^c}{2}$. Note that $\prmlimit(n) \leq p(n):=L(n) + \frac{n^c}{2}$.

We will now show that $\EFI_\prm$ is an $s^*$-\imbalanced \EFI (Definition \ref{def:nuefi}).
It is easy to see that $\EFI_{\prm}(1^n, 0)$ and $\EFI_{\prm}(1^n, 1)$ are efficiently sampleable for all $\prm, n \in \bbN$.
The proof of the theorem follows from the two lemmas below.
\begin{lemma} (Computational Indistinguishability.) There exists a negligible function $\mu(\cdot)$ such that for all quantum polynomial-sized adversaries $\cA$,  for
every $n \in \mathbb{N}$ and every $\prm \leq \prmlimit(n)$,
\[
    \left|\Prr_{z \leftarrow \EFI_{\prm}(1^n, 0)}[\cA(z) = 1] - \Prr_{z \leftarrow \EFI_{\prm}(1^n, 1)}[\cA(z) = 1]\right|\leq \mu(n)
\]
\end{lemma}
\begin{proof}
Fix some quantum polynomial-sized adversary $\cA$. Consider the distribution $\{h, h(\simsgen(n))_{\prm}\}$. Since $\sgen$ and $\simsgen$ are computationally indistinguishable, there exists some negligible function $\mu'(\cdot)$ such that 
\[
    \left|\Prr_{\substack{h \leftarrow \bin^\ell\\x\leftarrow \sgen(n)}}[\cA(h, h(x)_{\prm}) = 1] - \Prr_{\substack{h \leftarrow \bin^\ell\\x\leftarrow \simsgen(n)}}[\cA(h, h(x)_{\prm}) = 1]\right|\leq \mu(n)
\]
Additionally, since for sufficiently large $n$ $$\Hs^\epsilon_{\min}(\simsgen(n)) \geq \prmlimit(n) + n^c/2 \geq \prm + n^c/2$$ by the Leftover Hash Lemma (Theorem \ref{thm:LHLforSmoothEntropy}) $$\SD\left(\{h, h(\simsgen(n))_{\prm}\}, \{h, U_{\prm}\}\right) \leq \frac{1}{2^{n^c/4}} + \epsilon(n)$$
Putting both together, 
\[
    \left|\Prr_{\substack{h \leftarrow \bin^\ell\\x\leftarrow \sgen(n)}}[\cA(h, h(x)_{\prm}) = 1] - \Prr_{\substack{h \leftarrow \bin^\ell\\x\leftarrow \bin^\prm}}[\cA(h, U_\prm) = 1]\right|\leq \mu(n) + \frac{1}{2^{n^c/4}} + \epsilon(n)
\]
which is negligible.
\end{proof}

\begin{lemma} (Statistical Distance.) There exists a negligible function $\mu(\cdot)$ such 
for every $n \in \mathbb{N}$ and every $\prm \geq \prmlimit(n)$, 
    $$\TD(\EFI_{\prm}(1^n, 0), \EFI_{\prm}(1^n, 1)) \geq 1 - \mu(n)$$
\end{lemma}
\begin{proof} 
Since the output of $\EFI$ is classical, trace distance is equivalent to statistical distance. We use the following sub-claim.
\begin{subclaim}
    Let $X$ be a random variable such that $\Hs^\epsilon_{\max}(X) \leq v$, and let $d$ be the keylength of a universal hash function from $\bin^{|X|}$ to $\bin^{v+t}$ for some $t\geq 0$. Let $h$ be sampled uniformly from $\bin^d$. Then statistical distance between the distributions 
    \[
        \{h, h(X)\} \text{ and } \{h, U_{v+t}\}
    \]
    is atleast $1 - 2^{-t} - \epsilon$.
\end{subclaim}
\begin{proof}
    Since $\Hs^\epsilon_{\max}(X) \leq v$, there must exist a random variable $X'$ such that $
    \SD(X,X') \leq \epsilon$ and $\Hs_{\max}(X') \leq v$. This means that for all $x \in X'$
    \begin{gather*}
        \Hs_{X'}(x) \leq v\\
        \implies \Prr_{X'}(x) \geq 1/2^{v}\\
        \implies |\Supp(X')| \leq 2^{v}
    \end{gather*}
    Now, 
    \[
    SD(\{h,h(X')\}, \{h, U_{v+t}\}) = \bbE_h[\SD(h(X'), U_{v+t})]
    \]
    For any $h$, $\SD(h(X'), U_{v+t})$ may be written as
    \begin{align*}
        \SD(h(X'), U_{v+t}) &= \sum_{z\in\bin^{v+t}} \left|\Prr_{x\leftarrow X'}[h(x) = z] - \frac{1}{2^{v+t}}\right|\\
        & \begin{multlined}
        =\sum_{z\in\Supp(h(X'))} \left|\Prr_{x\leftarrow X'}[h(x) = z] - \frac{1}{2^{s+n}}\right| \\+ \sum_{z\in \bin^{v+t}\setminus\Supp(h(X'))} \left|\Prr_{x\leftarrow X'}[h(x) = z] - \frac{1}{2^{v+t}}\right|
        \end{multlined}\\
        &\geq \sum_{z\in \bin^{v+t}\setminus\Supp(h(X'))} \left|\Prr_{x\leftarrow X'}[h(x) = z] - \frac{1}{2^{v+t}}\right|\\
        &\geq \sum_{z\in \bin^{s+n}\setminus\Supp(h(X'))} \frac{1}{2^{v+t}}\\
         &\geq \frac{2^{v+t} - \Supp(h(X'))}{2^{v+t}}\\
        &\geq \frac{2^{v+t} - \Supp(X')}{2^{v+t}}\\
        &\geq 1 - 2^{-t}
    \end{align*}
    Therefore 
    \[
        \SD(\{h,h(X')\}, \{h, U_{v+t}\}) \geq 1 - 2^{-t}
    \]
    By the triangle inequality
    \[
        \SD(\{h,h(X)\}, \{h, U_{v+t}\}) \geq 1 - 2^{-t} - \epsilon
    \]
    which concludes the proof of the subclaim.
\end{proof}
Note $\prm-n^c/2 \geq \prmlimit(n)-n^c/2 \geq \Hs^\epsilon_{\max}(\sgen(n))$. Therefore, setting $\sgen(n)$ as $X$, $\prm(n) - n^{c}/2$ as $v$, and $n^c/2$ as $t$ in the above claim, the statistical distance between $\{h,h(\sgen(n))_{\prm}\}$ and $\{h, U_{\prm}\}$ is atleast $1 - \epsilon(n) - 2^{-n^c/2}$ which concludes the proof of the claim.
    \end{proof}
\begin{corollary}
    \label{cor:nuefi}
    Assume the existence of pure one-way state generators according \defref{owsg}. Then there exists a function $\prmlimit(n)$ such that 
    \begin{itemize}
        \item There exists a polynomial $p$ such that for all $n \in \bbN$, $|s^*(n)| \leq \log p(n)$ and
        \item There exists a $s^*$-\imbalanced \EFI (Definition \ref{def:nuefi}). 
\end{itemize}
\end{corollary}
\begin{proof}
    This follows immediately from Theorems \ref{thm:five-two}, \ref{thm:OWPtoQWPEG}, \ref{thm:QWPEGtoQPEG} and \ref{thm:QPEGtoEFI}.
\end{proof}

\section{Imbalanced EFI Imply Commitments}
\label{sec:com}
In this section, we prove that imbalanced EFI imply non-uniform commitments. First, we state the definitions of quantum commitments that we work with.

\input{extra-prelim}

\subsection{Imbalanced EFI imply Imbalanced Commitments}

We import the following lemmas from prior work: while these were proven in the uniform setting, we observe that they also carry over to the imbalanced setting with similar proofs.
\begin{lemma}(Imported, rephrased) ~\cite{AC:Yan21,AC:Yan22}
\label{lem:yan}%
    For every function $s^*(\cdot)$, an $s^*$-\imbalanced EFI satisfying \defref{nuefi} implies an \imbalanced quantum bit commitment satisfying right-$s^*$-statistical collapse binding according to \defref{imb-binding} and left-$s^*$-computational hiding according to \defref{imb-hiding}.
\end{lemma}
\begin{proof}(Sketch.)
Consider the purification $\EFI'_\prm$ of the procedure from \defref{nuefi}.
When $\prm \leq \prmlimit(n)$, the unitary $\Com_\prm(1^n, b):=\EFI'_\prm(1^n,b)$ satisfies computational hiding according to \defref{imb-binding} by inheriting the hiding from \defref{nuefi}. 

When $\prm \geq \prmlimit(n)$, then by Uhlmann's theorem there are no purification registers that can result in $\EFI'_\prm(1^n,b)$ and $\EFI'_\prm(1^n,b)$ having noticeable overlap. Thus the corresponding $\Com_\prm$ is statistical honest binding~\cite{AC:Yan21} for $\prm \geq \prmlimit(n)$, which is known \cite{AC:Yan22} to imply sum-binding, which in turn implies collapse-binding~\cite{DS23,STOC:GJMZ23} for the same parameters.
\end{proof}

The lemma below follows from the generic transformation in~\cite{HMY22,STOC:GJMZ23} (referred to as a ``flavour swap"), which can be applied to any commitment $\mathsf{Com}_s$ to obtain a commitment $\widetilde{\mathsf{Com}}_s$ that is hiding whenever $\mathsf{Com}_s$ is binding and binding whenever $\mathsf{Com}_s$ is hiding.

\begin{lemma}(Imported, rephrased) ~\cite{HMY22,STOC:GJMZ23}
\label{lem:hmy}%
    For every function $s^*(\cdot)$, 
    an \imbalanced quantum bit commitment satisfying right-$s^*$-statistical collapse binding according to \defref{imb-binding} and left-$s^*$-computational hiding according to \defref{imb-hiding} implies
    an \imbalanced quantum bit commitment satisfying right-$s^*$-statistical hiding according to \defref{imb-hiding} and left-$s^*$-computational collapse binding according to \defref{imb-binding}.
\end{lemma}

We obtain the following corollary by combining the results in previous sections with the two lemmas above.

\begin{corollary}
Assume the existence of pure one-way state generators according to Definition \ref{def:owsg}.
Then there exists a function $\prmlimit(\cdot)$ for which there exists
\begin{enumerate}
\item An \imbalanced quantum bit commitment satisfying right-$s^*$-statistical collapse binding according to \defref{imb-binding} and left-$s^*$-computational hiding according to \defref{imb-hiding}.
\item An \imbalanced quantum bit commitment satisfying left-$s^*$-computational collapse binding according to \defref{imb-binding} and right-$s^*$-statistical hiding according to \defref{imb-hiding}.
\end{enumerate}
Furthermore, there exists an explicit constant $c$ such $\prmlimit(\cdot)$ satisfies that for large enough $n \in \mathbb{N}$, $|\prmlimit(n)| \leq n^c$.
\end{corollary}
\begin{proof}
    This follows immediately from Corollary \ref{cor:nuefi}, Lemma \ref{lem:yan} and Lemma \ref{lem:hmy}.
\end{proof}

%% file: extra-prelim.tex
\subsection{Defining Commitments}
\label{sec:com-def}

\subsubsection{Quantum Bit Commitments}
A non-interactive quantum bit commitment is defined by a commitment unitary $\mathsf{Com}$, and has the following syntax. 
\begin{itemize}
\item {\bf Commit Phase.} On input bit $b$, the committer initializes a single bit register $\mathsf{M}$ to $b$ and $\ell(n)$-qubit auxiliary register $\mathsf{W}$ to $\ket{0}^{\ell(n)}$. The committer then applies a unitary $\Com$ to registers $\mathsf{M}, \mathsf{W}$ and writes the resulting state on registers $(\mathsf{C}, \mathsf{D})$, where $\mathsf{C}$ denotes the ``commit'' register, that is sent to the receiver.
\item {\bf Decommit Phase.}
The committer decommits by sending the $\mathsf{D}$ register. The receiver applies $\Com^\dagger$ to
the pair $(\mathsf{C}, \mathsf{D})$, obtaining $(\mathsf{M}, \mathsf{W})$. 
It verifies the decommitment by checking that $\mathsf{W}$ is in the state $\ket{0}^{\ell(n)}$ via a projective measurement. If the measurement does not succeed, it outputs $\bot$. Otherwise, it outputs the classical bit obtained by measuring $\mathsf{M}$.
\end{itemize}
We will require these bit commitments to satisfy the following properties. \newline

\noindent {\bf Completeness.}
The (honest) receiver's output at the end of the decommit phase equals the (honest) committer's input.\\

\noindent {\bf Collapse Binding for Bit Commitments.}
For the case of commitments to classical bits, it will sometimes be convenient to work with collapse binding~\cite{C:Unruh14,STOC:GJMZ23}, which intuitively says that an adversary cannot detect whether or not their opened message was measured.
Formally, we consider the following experiment parameterized by a (malicious) committer $\cA$ and security parameter $n$. \newline

\noindent \underline{$\mathsf{Expmt}\text{-}\mathsf{Binding}_{\cA,n}^{\Com}$}:
\begin{itemize}
\item $\mathsf{\cA}$ (arbitrarily) prepares and outputs registers $(\mathsf{C}, \mathsf{D})$ to a challenger.
\item {\bf Validity Check.} The challenger applies $\mathsf{Com}^{\dagger}$ to $(\mathsf{C}, \mathsf{D})$ to obtain $(\mathsf{M}, \mathsf{W})$, then projects $\mathsf{W}$ onto $\ket{0}^{\ell(n)}$, i.e. measures $\{\ket{0^{\ell(n)}}\bra{0^{\ell(n)}}, \mathbb{I} - \ket{0^{\ell(n)}}\bra{0^{\ell(n)}}\}$ on $\mathsf{W}$. 
If the measurement rejects, the challenger sets $b' \leftarrow \{0,1\}$ and ends the experiment. If not, the experiment continues.
\item The challenger samples $b \leftarrow \{0,1\}$ and:
\begin{itemize}
    \item If $b = 0$, the challenger does nothing.
    \item If $b = 1$, the challenger measures $\mathsf{M}$ in the standard basis. 
\end{itemize}
\item 
Finally, the challenger applies $\mathsf{Com}$ to $(\mathsf{M}, \mathsf{W})$
    to obtain $(\mathsf{C}, \mathsf{D})$, and returns $\mathsf{D}$ to $\cA$. 
\item Set $b'$ to equal the bit output by $\cA$.
\end{itemize}
We define $\mathsf{Adv}_{\Com,\cA,n}^{\mathsf{Expmt}\text{-}\mathsf{Binding}}$ as the probability that $b' = b$ in $\mathsf{Expmt}\text{-}\mathsf{Binding}_{\cA,n}^{\Com}$.

\begin{definition}[Binding for Quantum Bit Commitments~\cite{C:Unruh14,STOC:GJMZ23}]
\label{def:binding-bitcom}
A non-interactive quantum bit commitment satisfies statistical (resp., computational) binding if for every unbounded (resp., quantum polynomial-sized) adversary $\cA = \{\cA_{n}\}_{n \in \mathbb{N}}$, there exists a negligible function $\nu(\cdot)$ such that 
$\mathsf{Adv}_{\Com,\cA,n}^{\mathsf{Expmt}\text{-}\mathsf{Binding}} \leq \frac{1}{2} + \nu(n)$.
\end{definition}

\noindent {\bf Hiding for Quantum Bit Commitments.}
We consider the following experiment parameterized by a (malicious) receiver $\cA$ and security parameter $n$. Intuitively, this says that an adversary cannot tell whether a commitment is to zero or one. \newline

\noindent \underline{$\mathsf{Expmt}\text{-}\mathsf{Hiding}_{\cA,n}^{\Com}$}:
\begin{itemize}

    \item The challenger samples $b \leftarrow \{0,1\}$, then sets register $\mathsf{M}$ to the bit $b$.
    It initializes $\mathsf{W}$ to $\ket{0^{\ell(n)}}$. It then applies $\mathsf{Com}$ to registers $(\mathsf{M}, \mathsf{W})$ writing the result on registers $(\mathsf{C}, \mathsf{D})$. It sends $\mathsf{C}$ to $\cA$.
    \item Denote the output of $\cA$ by $b'$.
\end{itemize}
We define $\mathsf{Adv}_{\Com,\cA,n}^{\mathsf{Expmt}\text{-}\mathsf{Hiding}}$ as the probability that $b' = b$ in $\mathsf{Expmt}\text{-}\mathsf{Hiding}_{\cA,n}^{\Com}$.

\begin{definition}[Hiding for Non-interactive Quantum Bit Commitments]
\label{def:hiding-com}
A non-interactive quantum bit commitment satisfies statistical (resp., computational) hiding if for every unbounded (resp., quantum polynomial-sized) adversary $\cA = \{\cA_{n}\}_{n \in \mathbb{N}}$, there exists a negligible function $\nu(\cdot)$ such that 
$\mathsf{Adv}_{\Com,\cA,n}^{\mathsf{Expmt}\text{-}\mathsf{Hiding}} \leq \frac{1}{2} + \nu(n)$.
\end{definition}

\subsubsection{Imbalanced Commitments}
We will also consider constructions of commitments parameterized by a classical security-parameter dependent advice string $\prm$, with the corresponding unitary denoted by $\mathsf{Com}_\prm$.
We will require these to satisfy completeness for all values of the advice string $\prm$.
Furthermore, for a function $s^*$, we define notions of left-$s^*$-binding, right-$s^*$-binding, left-$s^*$-hiding and right-$s^*$-hiding.

Intuitively, left (resp. right) $s^*$-binding implies that binding holds for all values of advice $s \leq s^*$ (resp. $s \geq s^*$); and similarly for hiding.

\begin{definition}[Imbalanced Binding for Quantum Bit Commitments]
\label{def:imb-binding}
Let $s^*(\cdot)$ denote a function.
A non-interactive quantum bit commitment satisfies {\bf left-}$s^*$-statistical (respectively, computational) collapse binding if for every unbounded (respectively, quantum polynomial-sized) adversary $\cA = \{\cA_{n}\}_{n \in \mathbb{N}}$, there exists a negligible function $\nu(\cdot)$ such that for every $n \in \bbN$ and every \underline{$s \leq s^*(n)$},
$\mathsf{Adv}_{\Com_s,\cA,n}^{\mathsf{Expmt}\text{-}\mathsf{Binding}} \leq \frac{1}{2} + \nu(n)$.

It satisfies {\bf right-}$s^*$-statistical (respectively, computational) collapse binding if for every unbounded (respectively, quantum polynomial-sized) adversary $\cA = \{\cA_{n}\}_{n \in \mathbb{N}}$, there exists a negligible function $\nu(\cdot)$ such that for every $n \in \bbN$ and every \underline{$s \geq s^*(n)$},
$\mathsf{Adv}_{\Com_s,\cA,n}^{\mathsf{Expmt}\text{-}\mathsf{Binding}} \leq \frac{1}{2} + \nu(n)$.
\end{definition}

\begin{definition}[Imbalanced Hiding for Quantum Bit Commitments]
\label{def:imb-hiding}
Let $s^*(\cdot)$ denote a function.
A non-interactive quantum bit commitment satisfies {\bf left-}$s^*$-statistical (respectively, computational) hiding if for every unbounded (respectively, quantum polynomial-sized) adversary $\cA = \{\cA_{n}\}_{n \in \mathbb{N}}$, there exists a negligible function $\nu(\cdot)$ such that for every $n \in \bbN$ and every \underline{$s \leq s^*(n)$},
$\mathsf{Adv}_{\Com_s,\cA,n}^{\mathsf{Expmt}\text{-}\mathsf{Hiding}} \leq \frac{1}{2} + \nu(n)$.

It satisfies {\bf right-}$s^*$-statistical (respectively, computational) hiding if for every unbounded (respectively, quantum polynomial-sized) adversary $\cA = \{\cA_{n}\}_{n \in \mathbb{N}}$, there exists a negligible function $\nu(\cdot)$ such that for every $n \in \bbN$ and every \underline{$s \geq s^*(n)$},
$\mathsf{Adv}_{\Com_s,\cA,n}^{\mathsf{Expmt}\text{-}\mathsf{Hiding}} \leq \frac{1}{2} + \nu(n)$.
\end{definition}

%% file: combiners.tex
\subsection{Obtaining a Non-Uniform Hiding Commitment}
Next, we obtain a commitment that is non-uniformly hiding (i.e., hiding for a given choice of non-uniform advice $s^*$) by combining a sequence of imbalanced commitments.
\begin{definition}[$s^*$-Non-Uniform Hiding Commitment]
\label{def:nhcom}
Let $s^*(\cdot)$ denote a function.
An $s^*$-non-uniform hiding commitment is a quantum polynomial time algorithm $\Com_{\prm}(b, 1^\n)$ that obtains a classical parameter-dependent advice string $\prm$, and on input $b \in \bin$, outputs a (potentially mixed) quantum state such that:

\begin{enumerate}
\item {\bf Non-Uniform Computational Hiding.}
For every quantum polynomial-sized circuit $\cA = \{\cA_{n}\}_{n \in \mathbb{N}}$, there exists a negligible function $\nu(\cdot)$ such that for every $n \in \bbN$,
$\mathsf{Adv}_{\Com_{s^*(n)},\cA,n}^{\mathsf{Expmt}\text{-}\mathsf{QBinding}} \leq \frac{1}{2} + \nu(n)$.

\item {\bf (Always) Computational Binding.} 
There exists a negligible function $\mu(\cdot)$ such that for 
every quantum polynomial-sized circuit $\cA_s$,
all $n, s \in \mathbb{N}$,
$$\mathsf{Adv}_{\Com_s,\cA,n}^{\mathsf{Expmt}\text{-}\mathsf{Binding}} \leq \frac{1}{2} + \mu(n)$$
where 
$\mathsf{Adv}_{\Com_s,\cA,n}^{\mathsf{Expmt}\text{-}\mathsf{Binding}}$is the probability that $b' = b$ in $\mathsf{Expmt}\text{-}\mathsf{Binding}_{\cA,n}^{\Com_s}$, defined above.
\end{enumerate}
\end{definition}

\begin{theorem}
\label{thm:nhcom}
    Assume that there exists a function $\prmlimit(\cdot)$ for which there exists
\begin{enumerate}
\item An \imbalanced quantum bit commitment satisfying right-$s^*$-statistical collapse binding according to \defref{imb-binding} and left-$s^*$-computational hiding according to \defref{imb-hiding}.
\item An \imbalanced quantum bit commitment satisfying left-$s^*$-computational collapse binding according to \defref{imb-binding} and right-$s^*$-statistical hiding according to \defref{imb-hiding}.
\end{enumerate}
    Then there exists an $s^*$-non-uniform hiding commitment according to Definition \ref{def:nhcom}.
\end{theorem}

\begin{proof} (of Theorem \ref{thm:nhcom}).
Fix $\prmlimit$ as in the statement of the Theorem.

Let $\mathsf{Com}_{1,\prm}, \mathsf{Com}_{2,\prm}$ denote commitment unitaries for the two quantum bit commitments in bullets 1 and 2 of the theorem statement.\\

\noindent {\bf Construction.}
Define unitary $\widetilde{\Com}_{\prm}$ that given advice string $\prm$ commits to its input (classical) string twice simultaneously, via $\mathsf{Com}_{1,\prm}$ and $\mathsf{Com}_{2,\prm}$, as follows:
\begin{enumerate}
\item The input register is a $p = p(n)$ qubit register $\widetilde{\mathsf{M}}$, and there is a $p(n)+2\ell(n)$ qubit auxiliary register $\widetilde{\mathsf{W}}$, divided into sub-registers $\mathsf{W}_1, \mathsf{M}_2, \mathsf{W}_2$ of $\ell(n), p(n)$ and $\ell(n)$ qubits respectively. 
\item First apply a CNOT operation from the $\widetilde{\mathsf{M}}$ register onto the $\mathsf{M}_2$ sub-register. (This serves to copy the input string onto $\mathsf{M}_2$.)
\item Next, apply the $\Com_{1,\prm}$ unitary to registers $\widetilde{\mathsf{M}}, \mathsf{W}_1$ to obtain $(\mathsf{C}_1, \mathsf{D}_1)$ and the $\Com_{2,\prm}$ unitary to registers $\mathsf{M}_2, \mathsf{W}_2$ to obtain $(\mathsf{C}_2, \mathsf{D}_2)$. 
\item The output registers are $\mathsf{C} = (\mathsf{C}_1, \mathsf{C}_2), \mathsf{D} = (\mathsf{D}_1, \mathsf{D}_2)$. 
\end{enumerate}

\begin{lemma}
\label{lem:fivethirteen}
$\widetilde{\Com}_{\prm}$ satisfies computational binding according to Definition \ref{def:nhcom}.
\end{lemma}

\begin{proof}
    Suppose, towards a contradiction, that the statement of the claim is false. 
    Then, there exists a (malicious) committer $\cA_s$ and a polynomial $q(\cdot)$ such that in {$\mathsf{Expmt}\text{-}\mathsf{Binding}_{\cA,n,s}^{\widetilde{\Com}}$},
\begin{enumerate}
\item $\mathsf{\cA}$ (arbitrarily) prepares and outputs registers $(\widetilde{\mathsf{C}}, \widetilde{\mathsf{D}})$ to a challenger.
\item {\bf Validity Check.} The challenger applies $(\widetilde{\mathsf{Com}_s})^{\dagger}$ to $(\widetilde{\mathsf{C}}, \widetilde{\mathsf{D}})$ to obtain $(\widetilde{\mathsf{M}}, \widetilde{\mathsf{W}})$, then projects $\widetilde{\mathsf{W}}$ onto $\ket{0}^{p(n)+2\ell(n)}$, i.e. measures $\{\ket{0^{p(n)+2\ell(n)}}\bra{0^{p(n)+2\ell(n)}}, \mathbb{I} - \ket{0^{p(n)+2\ell(n)}}\bra{0^{p(n)+2\ell(n)}}\}$ on $\widetilde{\mathsf{W}}$. 
If the measurement rejects, the challenger returns uniform $b'$, and ends the experiment (this corresponds to the case where decommitment fails to verify). 

\item 
Otherwise, the challenger samples $b \leftarrow \{0,1\}$ and:
\begin{enumerate}
    \item If $b = 0$, the challenger does nothing.
    \item If $b = 1$, the challenger measures $\widetilde{\mathsf{M}}$ in the standard basis. 
\end{enumerate}
\item 
Finally, the challenger applies $\widetilde{\mathsf{Com}}_s$ to $(\widetilde{\mathsf{M}}, \widetilde{\mathsf{W}})$
    to obtain $(\widetilde{\mathsf{C}}, \widetilde{\mathsf{D}})$, and returns $\widetilde{\mathsf{D}}$ to $\cA_s$. 
\item Denote the output of $\cA_s$ by $b'$, then it holds that for infinitely many $n$, there exists $s$ such that
$$\Pr[b' = b] \geq \frac{1}{2} + \frac{1}{q(n)}$$
\end{enumerate}
Fix such $\cA_s$ and $q(\cdot)$.
We will build a reduction $\cB$ that uses $\cA_s$ to obtain a contradiction as follows. \newline

\underline{Reduction $\cB_s$}:
\begin{enumerate}
    \item $\cB_s(1^n)$ obtains $s^*(n)$ as non-uniform advice.
    If $s \geq s^*(n)$, it sets $e=1$, otherwise $e=2$.
    It initializes a challenger for $\Com_e$.
    \item It obtains registers $(\widetilde{\mathsf{C}}, \widetilde{\mathsf{D}})$ from $\cA_s$.
    \item {\bf Validity Check.} $\cB_s$ applies $(\widetilde{\Com}_s)^\dagger$ to $(\widetilde{\mathsf{C}}, \widetilde{\mathsf{D}})$ to obtain $(\widetilde{\mathsf{M}}, \widetilde{\mathsf{W}})$, then projects $\widetilde{\mathsf{W}}$ onto $\ket{0}^{p(n)+2\ell(n)}$, i.e. measures $\{\ket{0^{p(n)+2\ell(n)}}\bra{0^{p(n)+2\ell(n)}}, \mathbb{I} - \ket{0^{p(n)+2\ell(n)}}\bra{0^{p(n)+2\ell(n)}}\}$ on $\widetilde{\mathsf{W}}$. 
    If the measurement rejects, $\cB_s$ returns uniform $b' \leftarrow \{0,1\}$, and ends the game. 
    
    \item Otherwise, $\cB_s$ applies a CNOT operation from the $\widetilde{\mathsf{M}}$ register onto the $\mathsf{M}_2$ sub-register\footnote{Recall that $\widetilde{\mathsf{W}}$ is divided into sub-registers $\mathsf{W}_1, \mathsf{M}_2, \mathsf{W}_2$ of $\ell(n), p(n)$ and $\ell(n)$ qubits respectively}. 
    It then applies the $\Com_{1,s}$ unitary to registers $\widetilde{\mathsf{M}}, \mathsf{W}_1$ obtaining $\mathsf{C}_1, \mathsf{D}_1$, and the $\Com_{2,s}$ unitary to registers $\mathsf{M}_2, \mathsf{W}_2$ obtaining $\mathsf{C}_2, \mathsf{D}_2$. 

    \item $\cB_s$ finally sends $\mathsf{C}_e, \mathsf{D}_e$ to its external challenger for $\Com_{e,s}$, then
    obtains $\mathsf{D}_e$ from the external challenger, and returns $\widetilde{\mathsf{D}} = (\mathsf{D}_1, \mathsf{D}_2)$ to $\cA_s$.
    \item $\cB_s$ returns the output bit $b'$ of $\cA_s$ as its own output.
\end{enumerate}
Note that $\cB_{s,n}$ attempts to break $\Com_{1,s}$ if $\prm \geq \prmlimit(n)$, and tries to break $\Com_{2,s}$ otherwise. 
Then the statement of the lemma follows by the following claim (Claim~\ref{clm:maincomb}) together with the fact that $\cB$ only incurs a polynomial overhead above the size of $\cA_s$.
\begin{claim}
\label{clm:maincomb}
At least one of the following is true.
\begin{itemize}
\item Either are infinitely many $n \in \mathbb{N}$ for which there exists $s \geq s^*(n)$ such that
$$\mathsf{Adv}_{\Com_{1,s},\cB,n}^{\mathsf{Expmt}\text{-}\mathsf{Binding}} \geq  \frac{1}{q(n)}$$
\item Or there are infinitely many $n \in \mathbb{N}$ for which there exists $s \leq s^*(n)$ such that
$$\mathsf{Adv}_{\Com_{2,s},\cB,n}^{\mathsf{Expmt}\text{-}\mathsf{Binding}} \geq  \frac{1}{q(n)}$$
\end{itemize}
\end{claim}
\begin{proof}
    In what follows, we will show that $\cA$'s state at every step in an interaction with the challenger of $\widetilde{\Com}_s$ when the $\widetilde{\Com}_s$ challenger picks bit $b$, is identical to its state at every step in an interaction with $\cB$ when the $\Com_{e,s}$ challenger picks the same bit $b$. 
    
    Because $\cB$ mirrors the output of $\cA$, this will imply that
    \[\mathsf{Adv}_{\Com_{e,s},\cB,n}^{\mathsf{Expmt}\text{-}\mathsf{Binding}} =  \mathsf{Adv}_{\widetilde{\Com}_s,\cA,n}^{\mathsf{Expmt}\text{-}\mathsf{Binding}} \geq \frac{1}{q(n)}.\]
    Fix any $n, s \in \mathbb{N}$. 
    The following experiment describes an interaction of the reduction $\cB_s$ above and external challenger for $\Com_{e,s}$, with adversary $\cA_s$ against $\widetilde{\mathsf{Com}}_s$.
    \newline

\noindent \underline{$\mathsf{Expmt}_0$}:
    \begin{enumerate}
    \item {$\cB_s$ obtains registers $(\widetilde{\mathsf{C}}, \widetilde{\mathsf{D}})$ from $\cA_s$.}
    \item {{\bf Validity Check.} $\cB_s$ applies $(\widetilde{\Com}_s)^\dagger$ to $(\widetilde{\mathsf{C}}, \widetilde{\mathsf{D}})$ to obtain $(\widetilde{\mathsf{M}}, \widetilde{\mathsf{W}})$, then applies measurement $\{\ket{0^{2\ell(n)+p(n)}}\bra{0^{2\ell(n)+p(n)}}, \mathbb{I} - \ket{0^{2\ell(n)+p(n)}}\bra{0^{2\ell(n)+p(n)}}\}$ on $\widetilde{\mathsf{W}}$. 
    If the measurement rejects, $\cB_s$ returns uniform $b' \leftarrow \{0,1\}$, and ends the game. }
    
    {Otherwise, $\cB_s$ applies $\widetilde{\Com}_s$, which entails the following. First, apply a CNOT operation from the $\widetilde{\mathsf{M}}$ register onto the $\mathsf{M}_2$ sub-register. 
    Next, apply the $\Com_{1,s}$ unitary to registers $\widetilde{\mathsf{M}}, \mathsf{W}_1$ obtaining $\mathsf{C}_1, \mathsf{D}_1$, and the $\Com_{2,s}$ unitary to registers $\mathsf{M}_2, \mathsf{W}_2$ obtaining $\mathsf{C}_2, \mathsf{D}_2$.
    $\cB_s$ then sends $\mathsf{C}_e, \mathsf{D}_e$ to its external challenger.}

    \item {\bf  Challenger's Validity Check.} The external challenger obtains $(\mathsf{C}_e, \mathsf{D}_e)$. It applies $\Com_{e,s}^\dagger$ to $(\mathsf{C}_e, \mathsf{D}_e)$ to obtain $(\mathsf{M}_e, \mathsf{W}_e)$, then applies measurement $\{\ket{0^{\ell(n)}}\bra{0^{\ell(n)}}, \mathbb{I} - \ket{0^{\ell(n)}}\bra{0^{\ell(n)}}\}$ on $\mathsf{W}_e$. 
    If the measurement rejects, it returns uniform $b' \leftarrow \{0,1\}$, and ends the game.

    \item Otherwise, the challenger samples $b \leftarrow \{0,1\}$ and:
        \begin{enumerate}
            \item If $b = 0$, the challenger does nothing.
            \item If $b = 1$, the challenger measures ${\mathsf{M}}_e$ in the standard basis. 
        \end{enumerate}
    \item The challenger applies $\Com_{e,s}$ to $(\mathsf{M}_e, \mathsf{W}_e)$ to obtain $(\mathsf{C}_e, \mathsf{D}_e)$ and returns $\mathsf{D}_e$ to $b$.
    \item $\cB_s$ obtains $\mathsf{D}_e$ from the challenger, and returns $\widetilde{\mathsf{D}} = (\mathsf{D}_1, \mathsf{D}_2)$ to $\cA_s$.
    \item $\cB_s$ returns the output bit $b'$ of $\cA_s$.
    \end{enumerate}

    Observe that the unitary $\Com_{e,s}$ applied at the end of Step 2 in the experiment above is immediately reversed by the challenger at the beginning of Step $3$. Since $\Com_{e,s}^\dagger \Com_{e,s} = \mathbb{I}$ for any unitary $\Com_{e,s}$, we can remove both operations from the experiment. 
    At this point, observe that measurement 
    $\{\ket{0^{\ell(n)}}\bra{0^{\ell(n)}}, \mathbb{I} - \ket{0^{\ell(n)}}\bra{0^{\ell(n)}}\}$ on register $\mathsf{W}_e$ is performed twice in succession, and since $M M = M$ for any measurement $M$, we can remove one application of the measurement.
    Finally, also noting that operations on disjoint subsystems commute, it follows that the above experiment is equivalent to the following simpler experiment. 
    \newline

\noindent \underline{$\mathsf{Expmt}_1$}:
    \begin{enumerate}
    \item $\cB_s$ obtains registers $(\widetilde{\mathsf{C}}, \widetilde{\mathsf{D}})$ from $\cA_s$.
    \item {\bf $\cB_s$'s Validity Check.} Apply $(\widetilde{\Com}_s)^\dagger$ to $(\widetilde{\mathsf{C}}, \widetilde{\mathsf{D}})$ to obtain $(\widetilde{\mathsf{M}}, \widetilde{\mathsf{W}})$, then apply measurement $\{\ket{0^{\ell(n)}}\bra{0^{\ell(n)}}, \mathbb{I} - \ket{0^{\ell(n)}}\bra{0^{\ell(n)}}\}$ on $\widetilde{\mathsf{W}}$. 
    If the measurement rejects,  return uniform $b' \leftarrow \{0,1\}$, and end the game. 
    
    Otherwise, apply a CNOT operation from the $\widetilde{\mathsf{M}}$ (also denoted $\mathsf{M}_1$) register onto the $\mathsf{M}_2$ sub-register.

    \item Sample $b \leftarrow \{0,1\}$ and:
        \begin{enumerate}
            \item If $b = 0$, do nothing.
            \item If $b = 1$, measure ${\mathsf{M}}_e$ in the standard basis. 
        \end{enumerate}
    \item Apply $\Com_{1,s}$ to $(\mathsf{M}_1, \mathsf{W}_1)$ to obtain $(\mathsf{C}_1, \mathsf{D}_1)$, and $\Com_{2,s}$ to $(\mathsf{M}_{2}, \mathsf{W}_{2})$ obtaining $\mathsf{C}_{2}, \mathsf{D}_{2}$.
    \item Return $\widetilde{\mathsf{D}} = (\mathsf{D}_1, \mathsf{D}_2)$ to $\cA_s$.
    \item Set $b'$ to the output of $\cA_s$.
    \end{enumerate}
    As discussed above, $\mathsf{Expmt}_1$ performs identical operations as $\mathsf{Expmt}_0$, thus we have
    \begin{equation}
    \label{eq:a1}
    \Pr_{\mathsf{Expmt}_1}[b' = b] = \Pr_{\mathsf{Expmt}_0}[b' = b]
    \end{equation}
Finally, let $\mathsf{Expmt}_2$ be identical to $\mathsf{Expmt}_1$ except that $\mathsf{M}_1$ is measured in Step 3 of $\mathsf{Expmt}_2$, instead of $\mathsf{M}_e$. That is,
\newline

\noindent \underline{$\mathsf{Expmt}_2$}:
    \begin{enumerate}
    \item $\cB_s$ obtains registers $(\widetilde{\mathsf{C}}, \widetilde{\mathsf{D}})$ from $\cA_s$.
    \item {\bf $\cB_s$'s Validity Check.} Apply $(\widetilde{\Com}_s)^\dagger$ to $(\widetilde{\mathsf{C}}, \widetilde{\mathsf{D}})$ to obtain $(\widetilde{\mathsf{M}}, \widetilde{\mathsf{W}})$, then apply measurement $\{\ket{0^{\ell(n)}}\bra{0^{\ell(n)}}, \mathbb{I} - \ket{0^{\ell(n)}}\bra{0^{\ell(n)}}\}$ on $\widetilde{\mathsf{W}}$. 
    If the measurement rejects,  return uniform $b' \leftarrow \{0,1\}$, and end the game. 
    
    Otherwise, apply a CNOT operation from the $\widetilde{\mathsf{M}}$ (also denoted $\mathsf{M}_1$) register onto the $\mathsf{M}_2$ sub-register.

    \item Sample $b \leftarrow \{0,1\}$ and:
        \begin{enumerate}
            \item If $b = 0$, do nothing.
            \item If $b = 1$, measure ${\mathsf{M}}_1$ in the standard basis. 
        \end{enumerate}
    \item Apply $\Com_{1,s}$ to $(\mathsf{M}_1, \mathsf{W}_1)$ to obtain $(\mathsf{C}_1, \mathsf{D}_1)$ and $\Com_{2,s}$ to $\mathsf{M}_{2}, \mathsf{W}_{2}$ obtaining $\mathsf{C}_{2}, \mathsf{D}_{2}$.
    \item Return $\widetilde{\mathsf{D}} = (\mathsf{D}_1, \mathsf{D}_2)$ to $\cA_s$.
    \item Set $b'$ to the output of $\cA_s$.
    \end{enumerate}
Note that the CNOT operation (Step 2) on any state \[\sum_{i \in [0,2^{p(n)}-1]} (\alpha_i \ket{i})_{\mathsf{M}_1} \otimes \ket{0^{p}}_{\mathsf{M}_2}\] results in the state 
\[\sum_{i \in [0,2^{p(n)}-1]} \alpha_i \ket{i}_{\mathsf{M}_1} \ket{i}_{\mathsf{M}_2}.\]
This implies that measuring either one of the registers $\mathsf{M}_1$ or $\mathsf{M}_2$ results in an identical (mixed) state on the system. This implies that $\mathsf{Expmt}_1$ and $\mathsf{Expmt}_2$ are identical, and thus
\begin{equation}
    \label{eq:a3}
    \Pr_{\mathsf{Expmt}_2}[b' = b] = \Pr_{\mathsf{Expmt}_1}[b' = b]
\end{equation}

Next, note that the only difference between $\mathsf{Expmt}_2$ and $\mathsf{Expmt}\text{-}\mathsf{Binding}_{\cA,n}^{\widetilde{\Com}_s}$ occurs only in the case of $b = 1$. In this case in $\mathsf{Expmt}_2$, a CNOT is {\em first} applied from $\mathsf{M}_1$ onto $\mathsf{M}_2$ and then $\mathsf{M}_1$ is measured in the standard basis. Whereas in $\mathsf{Expmt}\text{-}\mathsf{Binding}_{\cA,n}^{\widetilde{\Com}_s}$, $\mathsf{M}_1$ is first measured in the standard basis and then a CNOT is applied from $\mathsf{M}_1$ onto $\mathsf{M}_2$. 
Since these operations commute, the experiments are identical. Thus, 
\begin{equation}
    \label{eq:a4}
    \Pr_{\mathsf{Expmt}_2}[b' = b] = \mathsf{Adv}_{\widetilde{\Com_s},\cA,n}^{\mathsf{Expmt}\text{-}\mathsf{Binding}} 
\end{equation}

Combining equations (\ref{eq:a1}), (\ref{eq:a3}) and (\ref{eq:a4}) with the fact that $\mathsf{Adv}_{\widetilde{\Com}_s,\cA,n}^{\mathsf{Expmt}\text{-}\mathsf{Binding}} \geq \frac{1}{q(n)}$ gives us that for $\cB_s$ defined above, 
\[\mathsf{Adv}_{\Com_{e,s},\cB,n}^{\mathsf{Expmt}\text{-}\mathsf{Binding}} = \mathsf{Adv}_{\widetilde{\Com}_s,\cA,n}^{\mathsf{Expmt}\text{-}\mathsf{Binding}} \geq \frac{1}{q(n)}.\]
By definition of $\cB_s$, this implies that at least one of the following is true.
\begin{itemize}
\item Either are infinitely many $n \in \mathbb{N}$ for which there exists $s \geq s^*(n)$ such that
\[\mathsf{Adv}_{\Com_{1,s},\cB,n}^{\mathsf{Expmt}\text{-}\mathsf{Binding}} = \mathsf{Adv}_{\widetilde{\Com}_s,\cA,n}^{\mathsf{Expmt}\text{-}\mathsf{Binding}} \geq \frac{1}{q(n)}.\]
\item Or there are infinitely many $n \in \mathbb{N}$ for which there exists $s \leq s^*(n)$ such that\[\mathsf{Adv}_{\Com_{2,s},\cB,n}^{\mathsf{Expmt}\text{-}\mathsf{Binding}} = \mathsf{Adv}_{\widetilde{\Com}_s,\cA,n}^{\mathsf{Expmt}\text{-}\mathsf{Binding}} \geq \frac{1}{q(n)}.\]
\end{itemize}
which completes the proof of the claim.
\end{proof}
\noindent This concludes the proof of Lemma $\ref{lem:fivethirteen}$.
\end{proof}
\noindent The construction given therefore satisfies Definition $\ref{def:nhcom}$, completing the proof of Theorem $\ref{thm:nhcom}$.
\end{proof}

\begin{lemma}
For non-uniform parameter $\prmlimit(n)$,
$\widetilde{\Com}$ satisfies computational hiding according to Definition \ref{def:nhcom}.
\end{lemma}
\begin{proof}
The proof follows by a straightforward hybrid argument.
Consider an intermediate unitary $\widehat{\Com}_s$ which, in place of applying a CNOT from $\widetilde{\mathsf{M}}$ onto $\mathsf{M}_2$ (and thereby committing to the same message twice), instead initializes $\mathsf{M}_2$ to $0$, thereby potentially committing to two different messages via the unitaries $\mathsf{Com}_{1,s}$ and $\mathsf{Com}_{2,s}$. 
By computational hiding of $\Com_{1,s}$ for $\prm = \prmlimit(n)$, 
there is a negligible function $\mu(\cdot)$ such that
the state of $\cA_s$ in $\mathsf{Expmt}\text{-}\mathsf{Hiding}_{\cA,n}^{\widetilde{\Com}_s}$ when the challenger picks $b = 1$ is computationally $\mu(n)$-close to its state when $\widehat{\Com}_{s}$ is applied instead of $\widetilde{\Com}_s$.
Furthermore, by the statistical hiding of $\Com_2$ for $\prm = \prmlimit(n)$, the state of $\cA_s$ in $\mathsf{Expmt}\text{-}\mathsf{Hiding}_{\cA,n}^{\widetilde{\Com}_s}$ when the challenger picks $b = 0$ is statistically $\mu'(n)$-close to its state  when $\widehat{\Com}_{s}$ is applied instead of $\widetilde{\Com}_s$.
This implies that for every polynomial-sized quantum circuit $\cA_s$ and every $n \in \mathbb{N}$,
\[
    \left|\Prr_{y \leftarrow \widetilde{\Com}_{s^*(n)}(1^n, 0)}[\cA(y) = 1] - \Prr_{y \leftarrow \widetilde{\Com}_{s^*(n)}(1^n, 1)}[\cA(y) = 1]\right|\leq \mu(n)+\mu'(n) = \mathsf{negl}(n)
\]
which completes the proof of the lemma.
\end{proof}

\subsection{Obtaining a (Uniform) Commitment}
In the previous section, we obtained a commitment that is binding for every choice of (non-uniform advice) $\prm$, but only hiding when the advice string matches the value $\prmlimit(n)$. We show that this implies a uniform construction of commitments, as long as $|\prmlimit(n)|$ is not too large. Formally, we prove the following theorem.

\begin{theorem}
    \label{thm:nhcomToCom}
    Assume that there exists a function $\prmlimit(\cdot)$ such that 
    (1) there exists a polynomial $p(\cdot)$ such that for all $n \in \mathbb{N}$, $|\prmlimit(n)| \leq (\log p(n))$, and (2) 
    there exists an $\prmlimit$-non-uniform hiding commitment according to Definition \ref{def:nhcom}.
    Then there exists a (standard) uniform commitment satisfying computational collapse binding according to Definition \ref{def:binding-bitcom} and computational hiding according to Definition \ref{def:hiding-com}.
\end{theorem}

\begin{proof} 

Let $t = t(n) = 2^{|\prmlimit(n)|} \leq p(n)$, 
and let $\Com_1, \Com_2, \ldots, \Com_t$ denote the commitment unitaries for $m$ quantum string commitments, one corresponding to each possible value of $\prmlimit(n)$.\\
\noindent {\bf Construction.}
Define unitary $\widetilde{\Com}$ as follows:
\begin{enumerate}
\item The input is a single bit $b$, and there is a $t(n) \cdot \left(\ell(n) + 1 \right)$ qubit auxiliary register $\widetilde{\mathsf{W}}$, divided into sub-registers $\mathsf{M}_1, 
\mathsf{M}_2, \mathsf{M}_3, \ldots, \mathsf{M}_t,
\mathsf{W}_1, \mathsf{W}_2, \ldots, \mathsf{W}_t$ where the $\mathsf{M}_i$ registers are initialized to $\ket{0}$ and the $\mathsf{W}_i$ registers are initialized to $\ket{0^{\ell(n)}}$.
\item $\widetilde{\Com}$ applies the $H$ gate to each of the registers $\mathsf{M}_1, \ldots \mathsf{M}_{t-1}$. Next, it applies unitary $U$ 
that maps
\[ (m, x_1, x_2, \ldots, x_{t-1}, y) \mapsto (m, x_1, x_2, \ldots, x_{t-1}, \bigoplus_{i \in [t-1]}x_i \oplus y \oplus m)\]
to registers 
$\widetilde{\mathsf{M}}, \mathsf{M}_1, \ldots, \mathsf{M}_t$.
\item Next for $i \in [t]$, it applies the $\Com_i$ unitary to registers $\mathsf{M}_i, \mathsf{W}_i$ to obtain $(\mathsf{C}_i, \mathsf{D}_i)$. 
\item The output registers are $\mathsf{C} = (\mathsf{C}_1, \ldots, \mathsf{C}_t), \mathsf{D} = (\mathsf{D}_1, \ldots,  \mathsf{D}_t)$. 
\end{enumerate}

\begin{lemma}
$\widetilde{\Com}$ satisfies computational hiding according to Definition \ref{def:hiding-com}.
\end{lemma} 

\begin{proof}

Let $s^* = s^*(n) \in t[n]$ be such that $\Com_{s^*}$ satisfies statistical (resp., computational) hiding according to Definition \ref{def:hiding-com}.

Further, assume towards a contradiction that the lemma is not true. Then there exists an unbounded (resp., quantum polynomial-sized) adversary $\cA$ and a polynomial $q(\cdot)$ such that $\mathsf{Adv}_{\widetilde{\Com},\cA,n}^{\mathsf{Expmt}\text{-}\mathsf{Hiding}}$ (in Definition \ref{def:hiding-com}) is at least $\frac{1}{q(n)}$.
Consider the following intermediate experiment. \newline

\noindent
\underline{$\mathsf{Expmt}\text{-}\mathsf{intermediate}_{\cA,n}$}:
\begin{itemize}
    \item $\cA$ outputs a quantum state on register $\mathsf{M}$.
    \item The challenger samples $b \leftarrow \{0,1\}$, then swaps the contents of ${\mathsf{M}}$ with $\ket{0}$ if $b = 0$. It initializes $\widetilde{\mathsf{W}}$, divided into sub-registers $\mathsf{M}_1, \mathsf{M}_2, \mathsf{M}_3, \ldots, \mathsf{M}_t, \mathsf{W}_1, \mathsf{W}_2, \ldots, \mathsf{W}_t$ where for all $j \in [t]$, the $\mathsf{M}_j$ registers are initialized to $\ket{0}$ and the $\mathsf{W}_j$ registers are initialized to $\ket{0^{\ell(n)}}$.
    
    It then applies the $H$ gate to each of the registers $\mathsf{M}_1, \ldots \mathsf{M}_{t-1}$. Next, it 
    applies unitary $U$ mapping \[ (m, x_1, x_2, \ldots, x_{t-1}, y) \mapsto (m, x_1, x_2, \ldots, x_{t-1}, \bigoplus_{i \in [n-1]}x_i \oplus y \oplus m)\] to registers  $\mathsf{M}, \mathsf{M}_1, \ldots, \mathsf{M}_t$. 
    Next, it swaps out the value on register $\mathsf{M}_{s^*(n)}$ to $\ket{0}$, and finally
     for $j \in [t]$, it applies the $\Com_j$ unitary to registers $\mathsf{M}_j, \mathsf{W}_j$ to obtain $(\mathsf{C}_j, \mathsf{D}_j)$.

    It sends $\widetilde{\mathsf{C}} = (\mathsf{C}_1, \ldots, \mathsf{C}_t)$ to $\cA_s$.
    \item Denote the output of $\cA$ by $b'$.
\end{itemize}

    Note that the contents of registers $\mathsf{C}_1, \ldots, \mathsf{C}_t$ in this intermediate experiment are independent of the bit $b$ (due to swapping out $\mathsf{M}_j$ with $\ket{0}$), therefore for any $\cA$, $\Pr[b' = b] = \frac{1}{2}$ in the intermediate experiment.

    This implies that there exists a fixing of $b \in \{0,1\}$ such that $\cA$'s output in the intermediate experiment $\mathsf{Expmt}\text{-}\mathsf{Hiding}_{\cA,n}^{\widetilde{\Com}}$ with this fixing of $b$, is at least $\frac{1}{2q(n)}$-far from its output in $\mathsf{Expmt}\text{-}\mathsf{intermediate}_{\cA,n}$.
    Suppose this holds for $b = 1$ (the case of $b = 0$ follows similarly).
    We build a (non-uniform) reduction $\cB$ that interacts with an external challenger to break the purported hiding of $\mathsf{Com}_{s^*}$, as follows.
    \begin{enumerate}
    \item Obtain a quantum state on register $\mathsf{M}$ from $\cA$.
    \item  
    Prepare a $t(n) \cdot \left(\ell(n) +1 \right)$ qubit register $\widetilde{\mathsf{W}}$, divided into sub-registers $\mathsf{M}_1, 
\mathsf{M}_2, \mathsf{M}_3, \ldots, \mathsf{M}_t$, and $\mathsf{W}_1, \mathsf{W}_2, \ldots, \mathsf{W}_t$ where for all $j \in [t]$, the $\mathsf{M}_j$ registers are initialized to $\ket{0}$ and the $\mathsf{W}_j$ registers are initialized to $\ket{0^{\ell(n)}}$.
\item Apply the $H$ gate to each of the registers $\mathsf{M}_1, \ldots \mathsf{M}_{t-1}$. Next, apply unitary $U$ 
that maps
\[ (m, x_1, x_2, \ldots, x_{t-1}, y) \mapsto (x_1, x_2, \ldots, x_{t-1}, \bigoplus_{i \in [t-1]}x_i \oplus y \oplus m)\]
to registers 
$\mathsf{M}, \mathsf{M}_1, \ldots, \mathsf{M}_t$. 
    \item Send the $\mathsf{M}_{s^*}$ register to the external challenger, and obtain register $\mathsf{C}_{s^*}$.
    \item For $j \in [t]\setminus s^*$, apply the $\mathsf{Com}_j$ unitary to registers $\mathsf{M}_j, \mathsf{W}_j$ to obtain $(\mathsf{C}_j, \mathsf{D}_j)$.
    \item Send $\mathsf{C}_1, \ldots \mathsf{C}_t$ to $\cA$ and return the output $b'$ of $\cA$.
    \end{enumerate}
    By construction, the size of $\cB$ is only polynomially larger than that of $\cA$, and 
    \[ \Big|
    \Pr[\cB \text{ outputs }1|b = 0] - \Pr[\cB \text{ outputs }1|b = 1]
    \Big| \geq \frac{1}{2q(n)} \]
    This implies that
    \[ \Big|
    \Pr[\cB (\Com_{s^*}(1^n,0)) = 1] - \Pr[\cB (\Com_{s^*}(1^n,1)) = 1]
    \Big| \geq \frac{1}{2q(n)} \]
    which is a contradiction to the hiding of $\mathsf{Com}_{s^*}$, as desired, and thus the lemma must be true.
\end{proof}

\begin{lemma}
$\widetilde{\Com}$ satisfies statistical (resp., computational) collapse binding according to Definition \ref{def:binding-bitcom} as long as all of $\Com_1, \ldots, \Com_t$ satisfy statistical (resp., computational) collapse binding according to Definition \ref{def:binding-bitcom}. 
\end{lemma} 
\begin{proof}
    Suppose, towards a contradiction, that the statement of the theorem is false. 
    Then, there exists a (malicious) committer $\cA$ and a polynomial $q(\cdot)$ such that in {$\mathsf{Expmt}\text{-}\mathsf{Binding}_{\cA,n}^{\widetilde{\Com}}$},
\begin{enumerate}
\item $\mathsf{\cA}$ (arbitrarily) prepares and outputs registers $(\widetilde{\mathsf{C}}, \widetilde{\mathsf{D}})$ to a challenger.
\item {\bf Validity Check.} The challenger applies $(\widetilde{\mathsf{Com}})^{\dagger}$ to $(\widetilde{\mathsf{C}}, \widetilde{\mathsf{D}})$ to obtain $(\widetilde{\mathsf{M}}, \widetilde{\mathsf{W}})$, then projects $\widetilde{\mathsf{W}}$ onto $\ket{0}^{t(n) \cdot (1+\ell(n))}$.
If the projection rejects, the challenger returns uniform $b'$, and ends the experiment (this corresponds to the case where decommitment fails to verify). 

\item 
Otherwise, the challenger samples $b \leftarrow \{0,1\}$ and:
\begin{enumerate}
    \item If $b = 0$, the challenger does nothing.
    \item If $b = 1$, the challenger measures $\widetilde{\mathsf{M}}$ in the standard basis. 
\end{enumerate}
\item 
Finally, the challenger applies $\widetilde{\mathsf{Com}}$ to $(\widetilde{\mathsf{M}}, \widetilde{\mathsf{W}})$
to obtain $(\widetilde{\mathsf{C}}, \widetilde{\mathsf{D}})$. 

This entails applying the $H$ gate to each of the registers $\mathsf{M}_1, \ldots \mathsf{M}_{t-1}$. Next, applying unitary $U$ 
that maps
\[ (m, x_1, x_2, \ldots, x_{t-1}, y) \mapsto (m, x_1, x_2, \ldots, x_{t-1}, \bigoplus_{i \in [t-1]}x_i \oplus y \oplus m)\]
to registers 
$\widetilde{\mathsf{M}}, \mathsf{M}_1, \ldots, \mathsf{M}_t$.
Finally for $i \in [t]$, applying the $\Com_i$ unitary to registers $\mathsf{M}_i, \mathsf{W}_i$ to obtain $(\mathsf{C}_i, \mathsf{D}_i)$. 
    
\item The challenger returns $\widetilde{\mathsf{D}}$ to $\cA_s$. 
\item Denote the output of $\cA_s$ by $b'$, then it holds that (for infinitely many $n$),
$$\Pr[b' = b] \geq \frac{1}{2} + \frac{1}{q(n)}$$
\end{enumerate}

We will now consider a sequence of $t(n)$ hybrids, where $\mathsf{Hyb}_0$ is identical to the one outlined above, and for each $j \in [t(n)]$, $\mathsf{Hyb}_j$ is identical to $\mathsf{Hyb}_{j-1}$, except that the $\mathsf{M}_j$ register is measured in the standard basis before applying the $\mathsf{Com}_i$ unitary at the end of Step 4. We write the description of $\mathsf{Hyb}_j$ for completeness below, with the difference from the binding experiment underlined. 
\newline

\noindent \underline{$\mathsf{Hyb}_j$:}
\begin{enumerate}
\item $\mathsf{\cA}$ (arbitrarily) prepares and outputs registers $(\widetilde{\mathsf{C}}, \widetilde{\mathsf{D}})$ to a challenger.
\item {\bf Validity Check.} The challenger applies $(\widetilde{\mathsf{Com}})^{\dagger}$ to $(\widetilde{\mathsf{C}}, \widetilde{\mathsf{D}})$ to obtain $(\widetilde{\mathsf{M}}, \widetilde{\mathsf{W}})$, then projects $\widetilde{\mathsf{W}}$ onto $\ket{0}^{t(n) \cdot (1+\ell(n))}$.
If the projection rejects, the challenger returns uniform $b'$, and ends the experiment (this corresponds to the case where decommitment fails to verify). 

\item 
Otherwise, the challenger samples $b \leftarrow \{0,1\}$ and:
\begin{enumerate}
    \item If $b = 0$, the challenger does nothing.
    \item If $b = 1$, the challenger measures $\widetilde{\mathsf{M}}$ in the standard basis. 
\end{enumerate}
\item 
The challenger 
applies the $H$ gate to each of the registers $\mathsf{M}_1, \ldots \mathsf{M}_{t-1}$. Next, it applies unitary $U$ 
that maps
\[ (m, x_1, x_2, \ldots, x_{t-1}, y) \mapsto (m, x_1, x_2, \ldots, x_{t-1}, \bigoplus_{i \in [t-1]}x_i \oplus y \oplus m)\]
to registers 
$\widetilde{\mathsf{M}}, \mathsf{M}_1, \ldots, \mathsf{M}_t$.
\underline{Next, it measures registers $\mathsf{M}_1, \ldots \mathsf{M}_j$ in the standard basis.}

Finally for $i \in [t]$, it applies the $\Com_i$ unitary to registers $\mathsf{M}_i, \mathsf{W}_i$ to obtain $(\mathsf{C}_i, \mathsf{D}_i)$. 
    
\item The challenger returns $\widetilde{\mathsf{D}}$ to $\cA$. 
\item Denote the output of $\cA$ by $b'$.
\end{enumerate}

Note that in $\mathsf{Hyb}_t$, $\Pr[b' = b] = \frac{1}{2}$ for any $\cA$, because measuring all registers $\mathsf{M}_1, \ldots, \mathsf{M}_t$ in the standard basis implies a measurement of $\mathsf{M}$ in the standard basis.
Moreover by our assumption above, in $\mathsf{Hyb}_0$, $\Pr[b' = b] \geq \frac{1}{2} + \frac{1}{q(n)}$ for adversary $\cA$ and some polynomial $q(\cdot)$.
This implies that for every $n$, there exists $j^* = j^*(n) \in [t(n)]$ such that 
\begin{equation}
\label{eq:comb}
\Big| \Pr[b' = b|\mathsf{Hyb}_{j^*}] - \Pr[b' = b|\mathsf{Hyb}_{j^*-1}] \Big| \geq \frac{1}{q(n)t(n)}
\end{equation}

To complete the contradiction, we build a reduction $\cB$ that breaks statistical (resp., computational) binding of $\mathsf{Com}_{j^*}$, by doing the following.

\begin{enumerate}
\item Obtain registers $(\widetilde{\mathsf{C}}, \widetilde{\mathsf{D}})$ from $\cA_s$.
\item {\bf Validity Check.} $\cB$ applies $(\widetilde{\mathsf{Com}})^{\dagger}$ to $(\widetilde{\mathsf{C}}, \widetilde{\mathsf{D}})$ to obtain $(\widetilde{\mathsf{M}}, \widetilde{\mathsf{W}})$, then projects $\widetilde{\mathsf{W}}$ onto $\ket{0}^{t(n) \cdot (1+\ell(n))}$.
If the projection rejects, it returns uniform $b'$, and ends the experiment (this corresponds to the case where decommitment fails to verify). 

\item 
Otherwise, $\cB$ samples $b \leftarrow \{0,1\}$ and:
\begin{enumerate}
    \item If $b = 0$, it does nothing.
    \item If $b = 1$, it measures $\widetilde{\mathsf{M}}$ in the standard basis. 
\end{enumerate}
\item 
$\cB$ applies 
the $H$ gate to each of the registers $\mathsf{M}_1, \ldots \mathsf{M}_{t-1}$. Next, it applies unitary $U$ 
that maps
\[ (m, x_1, x_2, \ldots, x_{t-1}, y) \mapsto (m, x_1, x_2, \ldots, x_{t-1}, \bigoplus_{i \in [t-1]}x_i \oplus y \oplus m)\]
to registers 
$\widetilde{\mathsf{M}}, \mathsf{M}_1, \ldots, \mathsf{M}_t$.
Next, it measures registers $\mathsf{M}_1, \ldots \mathsf{M}_{j^*-1}$ in the standard basis.

Finally for $i \in [t]$, it applies the $\Com_i$ unitary to registers $\mathsf{M}_i, \mathsf{W}_i$ to obtain $(\mathsf{C}_i, \mathsf{D}_i)$. 

\item 
$\cB_s$ then sends $\mathsf{C}_{j^*}, \mathsf{D}_{j^*}$ to the binding challenger for $\mathsf{Com}_{j^*}$, and then obtains $\mathsf{D}_{j^*}$ from the challenger.
    
\item $\cB_s$ returns $\widetilde{\mathsf{D}} = (\mathsf{D}_1, \ldots, \mathsf{D}_t)$ to $\cA_s$. 
\item Let $b'$ denote the output of $\cA$. Output $1$ if $b' = b$, otherwise output $0$.
\end{enumerate}
Note that when the external challenger for $\mathsf{Com}_{j^*}$ samples its challenge $c=0$, the interaction of $\cB$ with $\cA$ corresponds to $\mathsf{Hyb}_{j^*-1}$, and otherwise to $\mathsf{Hyb}_{j^*}$.
Thus by equation~(\ref{eq:comb}), we have that:
\[ \Big| \Pr[\cB \text{ outputs }1|c = 1] - \Pr[\cB \text{ outputs }1|c = 0] \Big| \geq \frac{1}{q(n)t(n)}.
\]
Since the size of $\cB$ is only polynomially larger than that of $\cA$, this contradicts\footnote{Note that any $\cB$ satisfying the equation above can be converted (with polynomial overhead) to an adversary that has advantage at least $\frac{1}{2q(n)t(n)}$ in the binding game.}
the statistical (resp., computational) binding of $\mathsf{Com}_{j^*}$, as desired.
\end{proof}
This concludes the proof of Theorem $\ref{thm:nhcomToCom}$.
\end{proof}
\begin{theorem}
There exists a constant $c>0$ such that $cn$ copy secure one-way state generators with pure state outputs (\defref{owsg}) imply quantum bit commitments satisfying computational collapse binding according to Definition \ref{def:binding-bitcom} and computational hiding according to Definition \ref{def:hiding-com}.
\end{theorem}
\begin{proof}
    The theorem follows from Corollary \ref{cor:nuefi}, Theorem \ref{thm:nhcom} and Theorem \ref{thm:nhcomToCom}.
\end{proof}

Computationally hiding and computationally collapse binding quantum bit commitments are known~\cite{BCQ22} to imply EFI pairs, which in turn are known to imply secure computation for all classical and quantum functionalities~\cite{C:BCKM21b,EC:GLSV21,C:AnaQiaYue22,BCQ22}. We therefore also have the following corollary.
\begin{corollary}
\label{cor:owsgsc2}
    There exists a constant $c>0$ such that $cn$ copy secure one-way state generators with pure state outputs imply secure computation for all quantum functionalities.
\end{corollary}

\section{Acknowledgments}
We thank James Bartusek, Yanyi Liu, Amit Sahai, and Taiga Hiroka for illuminating discussions and useful comments.
Both authors were supported in part by AFOSR, NSF 2112890 and NSF CNS-2247727. 
This material is based upon work supported by the Air Force Office of Scientific Research under award
number FA9550-23-1-0543.

%% file: owpuzzle-from-crypto.tex
\section{One-Way Puzzles are Necessary for QCCC Cryptography}
\label{app:implications}

In this section, we show that a variety of protocols, including encryption, commitments and digital signatures in the hybrid or local or quantum computation classical communication (QCCC) model imply one-way puzzles.
\subsection{Public Key Encryption}
First, we define public-key encryption in the QCCC model. Here, public keys and ciphertexts are classical while secret keys may be quantum. We provide a formal definition and then prove that this type of PKE implies one-way puzzles.
Later we will also show that PKE with classical public and secret keys, but quantum ciphertexts, implies one-way puzzles.
    
\begin{definition}[Public Key Encryption]
        A public key encryption scheme consists of a set of QPT algorithms $(\KeyGen, \Enc, \Dec)$ where
        \begin{itemize}
            \item $\KeyGen(1^n)$ on input the security parameter outputs a classical public key $pk$ and a quantum secret key $\mathsf{\rho}_{sk}$.
            \item $\Enc(pk, m)$ takes a public key and a classical message as input and outputs a classical ciphertext.
            \item $\Dec(\rho, c)$ takes a secret key $\rho$ and a classical ciphertext $c$ as input and outputs a message.
        \end{itemize}
        It has the following properties.
        \begin{itemize}
            \item \textbf{Correctness. }For all $m \in \bin^n$, $$\Prr_{(pk, \rho_{sk}) \leftarrow \KeyGen(1^n)}[\Dec(\rho_{sk},\Enc(pk,m)) = m] \geq 1-\negl(n)$$
            \item \textbf{Hiding. } For all QPT adversaries $\cA$\footnote{We note that this is a weaker property than CPA-security, we use this because it will give us a stronger result.}, $$\Prr_{\substack{(pk,\rho_{sk})\leftarrow\KeyGen(1^n)\\m\sample\bin^n}}[\cA(pk,(\Enc(pk,m)) = m] \leq \negl(n)$$
        \end{itemize}
    \end{definition}
    \begin{theorem}
        Public key encryption in the QCCC model implies one-way puzzles.
    \end{theorem}
\begin{proof}
    Let $(\KeyGen, \Enc, \Dec)$ be a PKE scheme. We define  algorithms $(\Samp, \Ver)$ as follows:
    \begin{itemize}
        \item $\Samp(1^n)$: Sample $(pk,\cdot) \leftarrow \KeyGen(1^n)$, $m \leftarrow \{0,1\}^n$
        and return $(k:= m, s: =(pk,\mathsf{Enc}(pk,m))$.
        \item $\Ver(k,s)$:
        Parse $k:=m$ and $s:=(pk,ct)$.
        Let $\rho_{sk}$ denote the mixed secret key corresponding to public key $s$. Accept if and only if 
        $m = \Dec(\rho_{sk}, ct)$.
    \end{itemize}
    By the correctness of the PKE scheme, with overwhelming probability over the sampling of keys, for any message $m$, $\Dec(\rho_{sk}, \Enc(pk, m)) = m$. Correctness of the one-way puzzle follows.
    Security of the one-way puzzle also follows immediately by hiding of the PKE scheme. 
\end{proof}

We also show that public-key encryption with classical public and secret keys implies one-way puzzles, even when the ciphertexts are quantum.
    \begin{theorem}
        Public key encryption with classical keys implies one-way puzzles.
    \end{theorem}
    \begin{proof}
        Let $(\KeyGen, \Enc, \Dec)$ be a PKE scheme. We define  algorithms $(\Samp, \Ver)$ as follows:
    \begin{itemize}
        \item $\Samp(1^n)$: Sample $(pk,sk) \leftarrow \KeyGen(1^n)$ and return $(sk, pk)$.
        \item $\Ver(k,s)$: Sample $m\leftarrow \bin^n$ and accept if and only if $m = \Dec_k(\Enc_s(m))$
    \end{itemize}
    \begin{claim}
        $(\Samp, \Ver)$ satisfies \defref{owp}.
    \end{claim}
    \begin{proof}
        By the correctness of the PKE scheme, with overwhelming probability over the sampling of keys $(pk, sk)$, for any message $m$, $\Dec_{sk}(\Enc_{pk}(m)) = m$. Correctness of the one-way puzzle follows.

        Suppose there exists an adversary that breaks the security of the one-way puzzle. That is there exists a non-negligible function $p(\cdot)$ such that for all $n\in \bbN$
        \[
            \Prr_{(k,s) \leftarrow \Samp(1^n)}[\Ver(\cA(s), s) = \top] \geq p(n)
        \]
        Rewriting using the definitions of $\Samp$ and $\Ver$ we get
        \[
            \Prr_{\substack{(pk,sk)\leftarrow\KeyGen(1^n)\\m\sample\bin^n}}[sk^*\leftarrow \cA(pk) \wedge m = \Dec_{sk^*}(\Enc_{pk}(m))] 
        \]
        Let $\cB$ be an algorithm that takes $(pk, z)$ as input, runs $\cA(pk)$ to get $sk^*$ and outputs $\Dec_{sk^*}(z)$. Then
        \[
            \Prr_{\substack{(pk,sk)\leftarrow\KeyGen(1^n)\\m\sample\bin^n}}[\cB(pk,(\Enc_{pk}(m)) = m] \geq p(n)
        \]
        which contradicts the hiding of the PKE scheme.
    \end{proof}
    The theorem follows trivially from the claim.
    \end{proof}
\subsection{Digital Signatures}
We define signatures in the QCCC model. Here, verification keys and signatures are classical, while signing keys may be quantum. We provide a formal definition and then prove that such signatures imply one-way puzzles. Later, we will also show that signature schemes with classical (signing and verification) keys, but quantum signatures, also imply one-way puzzles.

\begin{definition}[Signature Scheme]
        A digital signature scheme consists of a set of QPT algorithms $(\KeyGen, \sS, \sV)$ where
        \begin{itemize}
            \item $\KeyGen(1^n)$ takes the security parameter as input and outputs a signing key $\rho_{sk}$ and a verification key $vk$.
            \item $\sS(\rho_{sk}, m)$ takes a signing key and a message as input and outputs a signature $\sigma$.
            \item $\sV(vk, \sigma, m)$ takes a verification key, a signature, and a message as input and outputs a bit $b$.
        \end{itemize}
        A signature scheme has the following properties.
        \begin{itemize}
            \item \textbf{Correctness. }For all $m \in \bin^n$, $$\Prr_{(\rho_{sk}, vk) \leftarrow \KeyGen(1^n)}[\sV(vk, \sS(\rho_{sk},m),m) = 1] \geq 1-\negl(n)$$
            \item \textbf{Unforgeability. } For all QPT adversaries $\cA$, $$\Prr_{(\rho_{sk}, vk)\leftarrow\KeyGen(1^n)}[(t, m) \leftarrow\cA(vk) \wedge \sV(vk,t,m) = 1] \leq \negl(n)$$
        \end{itemize}
    \end{definition}
    \begin{theorem}
        QCCC signature schemes imply one-way puzzles.
    \end{theorem}
    \begin{proof}
         Let $(\KeyGen, \sS, \sV)$ be a PKE scheme. We define the algorithms $(\Samp, \Ver)$ as follows:
    \begin{itemize}
        \item $\Samp(1^n)$: Sample $(\rho_{sk}, vk) \leftarrow \KeyGen(1^n)$. Use $\rho_{sk}$ to sign $m = 0^n$, obtaining $t = \sS(\rho_{sk},0^n)$. Return $(k:=t, s:=vk)$.
        \item $\Ver(k,s)$: Accept if and only if $\sV(s,k,0^n) = 1$.
    \end{itemize}
    By the correctness of the signature scheme, with overwhelming probability over the sampling of keys $(sk, vk)$, for any message $m$ (and in particular for $m = 0^n$), $\sV(vk,(\sS(sk, m), m) = 1$. Correctness of the one-way puzzle follows.
    Furthermore, any adversary that breaks security of the one-way puzzle outputs a valid signature on $0^n$ with respect to a randomly sampled verification key $vk$, breaking unforgeability as defined above.
    \end{proof}

    \begin{theorem}
        Signature schemes with classical verification and signing keys imply one-way puzzles.
    \end{theorem}
    \begin{proof}
        Let $(\KeyGen, \sS, \sV)$ be a PKE scheme. We define the algorithms $\Samp, \Ver$ as follows:
    \begin{itemize}
        \item $\Samp(1^n)$: Sample $(sk, vk) \leftarrow \KeyGen(1^n)$ and return $(vk, sk)$.
        \item $\Ver(k,s)$: Sample $m\leftarrow \bin^n$ and accept if and only if $\sV(vk,(\sS(sk, m), m) = 1$
    \end{itemize}
    \begin{claim}
        $(\Samp, \Ver)$ satisfies \defref{owp}.
    \end{claim}
    \begin{proof}
        By the correctness of the signature scheme, with overwhelming probability over the sampling of keys $(sk, vk)$, for any message $m$, $\sV(vk,(\sS(sk, m), m) = 1$. Correctness of the one-way puzzle follows.

        Suppose there exists an adversary that breaks the security of the one-way puzzle. That is there exists a non-negligible function $p(\cdot)$ such that for all $n\in \bbN$
        \[
            \Prr_{(k,s) \leftarrow \Samp(1^n)}[\Ver(\cA(s), s) = \top] \geq p(n)
        \]
        Rewriting using the definitions of $\Samp$ and $\Ver$ we get
        \[
            \Prr_{\substack{(sk,vk)\leftarrow\KeyGen(1^n)\\m\sample\bin^n}}[sk^*\leftarrow \cA(vk) \wedge \sV(vk,(\sS(sk^*, m), m) = 1] 
        \]
        Let $\cB$ be an algorithm that takes $vk$ as input, runs $\cA(vk)$ to get $sk^*$, samples $m\sample\bin^n$ and outputs $\sS(sk^*, m), m$. Then
        \[
            \Prr_{(sk, vk)\leftarrow\KeyGen(1^n)}[(\sigma, m) \leftarrow\cB(vk) \wedge \sV(vk,\sigma,m) = 1] \geq p(n)
        \]
        which contradicts the unforgeability of the signature scheme.
    \end{proof}
    The theorem follows trivially from the claim.
    \end{proof}

\subsection{Bit Commitments}
We define commitments in the QCCC model below. The definition here differs from the one in Section \ref{sec:com-def} since unlike the commitment scheme we build from pure OWSGs, general QCCC commitments may be interactive.
\begin{definition}[Bit Commitment Scheme, Syntax]
    A (QCCC) bit commitment scheme is an efficient two-party protocol $\Com = \langle \sC, \sR \rangle$ between a committer $\sC$ and a receiver $\sR$ consisting of a commit stage and an opening stage.
   \begin{itemize} 
    \item \noindent\textbf{Commit Stage.} Both parties receive the security parameter $1^n$ and the committer $\sC$ receives a private input $b \in \bin$. It interacts with the receiver $\sR$ to produce a classical transcript $z$. At the end of the stage each party outputs a (private) quantum state, denoted by $q_\sC$ and $q_\sR$ respectively.
    \item 
    \noindent\textbf{Opening Stage. } Both parties receive the transcript $z$ produced in the first stage and their respective output states. They then interact with classical communication and at the end of this stage the receiver outputs a bit or the reject symbol $\bot$.
    \end{itemize}
    A bit-commitment scheme satisfies correctness if there exists a negligible function $\mu(\cdot)$ such that for all $n\in \bbN$ and all $b\in\bin$, when $\sC$ and $\sR$ are honest, $\sR$ outputs $b$ at the end of the opening stage with probability at least $1-\mu(n)$.
\end{definition}

\begin{definition}[Computational Hiding]
    A bit-commitment scheme is computationally hiding if there exists a negligible function $\mu(\cdot)$ such that for an honest committer $\sC$ that receives a bit $b$ as input, no QPT adversarial receiver $\cR$ can distinguish interactions where $\sC$ receives $b=0$ and interactions where $\sC$ receives $b=1$ with advantage greater than $\mu(n).$
\end{definition}

Below, we define a weak notion of binding for commitments, that we label computational weak honest binding. This is implied by other standard notions such as honest binding~\cite{AC:Yan22}.

\begin{definition}[Computational Weak Honest Binding]
    A quantum bit commitment scheme is weakly honest binding if there exists a negligible function $\mu(\cdot)$ such that for $b\in \bin$, no QPT adversary $\cA$ wins the following game with probability greater than $\mu(n)$.
    \begin{itemize}
        \item Run the commit stage of the commitment with an honest committer $\sC$ that receives input $b$ and an honest receiver $\sR$. Let $z$ be the transcript and let $q_\sR$ be the receiver state at the end of the commit stage.
        \item Run the opening stage between $\cA$ and $\sR$, where $\cA$ receives $z$ (but not the state output by $\sC$) and $\sR$ receives $\q_\sR$. The adversary wins if $\sR$ outputs $1-b$.
    \end{itemize}
\end{definition}

We say that a bit-commitment scheme has non-interactive opening if the opening stage consists of a single message sent from the committer to the receiver.

\begin{theorem}
    QCCC Bit commitments that are computationally hiding, computationally weak honest binding, and have non-interactive openings imply one-way puzzles.
\end{theorem}
\begin{proof}
    Let $\Com=(\sC,\sR)$ be a computationally hiding and computationally weak honest binding bit-commitment scheme. We define the algorithms $\Samp, \Ver$ as follows:
    \begin{itemize}
        \item $\Samp(1^n)$:
        \begin{enumerate}
            \item Run the commit stage with an honest committer $\sC$ that receives input $0$ and an honest receiver $\sR$. Let $z$ be the transcript and let $q_\sC$ and $q_\sR$ be the committer and receiver outputs respectively.
            \item Run the opening stage with an honest committer $\sC$ that receives input $(z, q_\sC)$ and an honest receiver $\sR$ that receives input $(z, q_\sR)$. Let the opening message be $d$.
            \item Output $(d, z)$
        \end{enumerate}
        \item $\Ver(k,s)$:
        \begin{enumerate}
            \item Consider an honest execution of the commit stage between $\sC$ with input $0$ and $\sR$. Sample a random output state $q^*_\sR$ for $\sR$ conditioned on the transcript being $s$. This may be achieved in exponential time by rejection sampling. 
            Output $\bot$ if no such $q^*_\sR$ exists.
            \item Return $\top$ if the output of $\sR$ on input $q^*_\sR$ and message $k$ is $0$. Return $\bot$ otherwise.
        \end{enumerate}
    \end{itemize}
    \begin{claim}
        $(\Samp, \Ver)$ satisfies \defref{owp}.
    \end{claim}
    \begin{proof} 
    By construction, for a transcript $z$ generated by an honest execution of commitment to $0$, the state $q^*_\sR$ generated by $\Ver(d, z)$ for arbitrary $d$ is identically distributed to the state $q_\sR$ output by the receiver during the honest execution. $\Ver(d,z)$ accepts whenever an honest receiver accepts message $d$ as an opening to $0$ given transcript $z$ and state $\q^*_\sR$. Therefore, for any message $d$, the probability that $\Ver(d,z)$ accepts is the probability that an honest receiver accepts message $d$ as an opening to $0$ given transcript $z$ and state $\q_\sR$.
    
    \noindent\textbf{Correctness. } $\Samp$ generates the message $d$ as an honest opening for $z$, an honest commitment to $0$. Therefore by the correctness of the commitment scheme the probability that $\Ver(d,z)$ accepts is atleast $1-\negl(n)$. Therefore
    \[
        \Prr_{(k,s)\leftarrow \Samp(1^n)}[\Ver(k,s)=\top] = 1-\negl(n)
    \]
    which proves correctness of the oneway puzzle.
    
    \noindent\textbf{Security. } We prove by contradiction. Suppose there exists a QPT adversary $\cA$ that breaks one-way puzzle security. That is there exists a non-negligible function $\epsilon(\cdot)$ such that for all $n \in \bbN$:
    \[
        \Prr_{(k,s)\leftarrow \Samp(1^n)}[\Ver(\cA(s),s)=\top] \geq \epsilon(n)
    \]
    Since $s$ is sampled as the transcript of an honest commitment to zero, this means that an honest receiver accepts the output of $\cA(s)$ as an opening to zero given transcript $s$ and the receiver output from the commitment execution. This may be written as
    \[
    \Prr\left[0 \leftarrow \sR(z, q_\sR, \cA(z))\ \middle|\ z, q_\sC, q_\sR \leftarrow \langle \sC(0, 1^n), \sR(1^n)\rangle\right] \geq \epsilon(n)
    \]
     By the computational weak honest binding property of the commitment, the probability that $\cA(z)$ is accepted as an opening to zero when the value committed to is one must be negligible. That is,
    \[
        \Prr\left[0 \leftarrow \sR(z, q_\sR, \cA(z))\ \middle|\ z, q_\sC, q_\sR \leftarrow \langle \sC(1, 1^n), \sR(1^n)\rangle\right] = \negl(n)
    \]
    We now build an algorithm $\cB$ that breaks the hiding of the commitment scheme. $\cB$ interacts with an honest committer $\sC$ as follows:
    \begin{itemize}
        \item During the commit stage, simulate an honest receiver in interactions with $\sC$ to obtain a transcript $z$ and a receiver state $q_\sR$.
        \item Run $\cA$ on input $z$ to obtain opening transcript $d^*$.
        \item Return $1$ if $R(z, q_\sR, d^*)$ outputs $0$. Else output $0$.
    \end{itemize}
    By the two previous inequalities,
    the probability that $\cB$ outputs $1$ when the commitment is to $0$ is atleast $\epsilon(n)$. Likewise, the probability that $\cB$ outputs $1$ when the commitment is to $1$ is atmost $\negl(n)$. $\cB$ therefore distinguishes both interactions with advantage atleast $\epsilon(n) - \negl(n)$, which contradicts the hiding of the commitment scheme.
    \end{proof}
    The theorem follows trivially from the previous claim.
    \end{proof}

    \subsection{Symmetric Encryption with Classical Keys and Ciphertexts}
    We prove that symmetric encryption schemes for classical messages, with classical keys and ciphertexts imply one-way puzzles.
\begin{definition}[Symmetric Encryption]
        A symmetric encryption scheme consists of a set of QPT algorithms $(\KeyGen, \Enc, \Dec)$ where
        \begin{itemize}
            \item $\KeyGen(1^n)$ takes the security parameter as input and outputs a secret key $k$.
            \item $\Enc_{sk}(m)$ takes a key and a message as input and outputs a ciphertext $c$.
            \item $\Dec_{k}(c)$ takes a key and a ciphertext as input and outputs a message.
        \end{itemize}
        A symmetric encryption scheme has the following properties.
        \begin{itemize}
            \item \textbf{Correctness. }For all $m \in \bin^n$, $$\Prr_{k \leftarrow \KeyGen(1^n)}[\Dec_{k}(\Enc_{k}(m)) = m] \geq 1-\negl(n)$$
            \item \textbf{Hiding. } All QPT adversaries $\cA$ win the following game with probability less than $1/2 + \negl(n)$. 
            \begin{itemize}
                \item $m_0, m_1 \leftarrow \cA(1^n)$
                \item $b\sample \bin$
                \item $k \leftarrow \KeyGen(1^n)$
                \item $b' \leftarrow \cA(\Enc_k(m_b))$
            \end{itemize}
            $\cA$ wins the game if $b'=b$.
        \end{itemize}
        We additionally assume that the size of the keyspace is a negligible fraction of the message space, i.e. $\frac{|\Supp(\KeyGen(1^n))|}{2^n} = \negl(n)$. This is without loss of generality for symmetric encryption schemes that support multi-message encryption.
    \end{definition}
    \begin{theorem}
        Symmetric encryption in the QCCC model implies one-way puzzles.
    \end{theorem}
    \begin{proof}
        Let $(\KeyGen, \Enc, \Dec)$ be a symmetric encryption scheme. We define the algorithms $\Samp, \Ver$ as follows:
    \begin{itemize}
        \item $\Samp(1^n)$: Sample $k \leftarrow \KeyGen(1^n)$ and $m\leftarrow\bin^n$ and return $k, (m,\Enc_k(m))$.
        \item $\Ver(k,s)$: Parse $s$ as $(m, c)$ and accept if $m = \Dec_k(c)$. Reject otherwise.
    \end{itemize}
    \begin{claim}
        $(\Samp, \Ver)$ satisfies \defref{owp}.
    \end{claim}
    \begin{proof}
        By the correctness of the symmetric encryption scheme, with overwhelming probability over the sampling of key $k$, for any message $m$, $\Dec_{k}(\Enc_{k}(m)) = m$. Correctness of the one-way puzzle follows.

        Suppose there exists an adversary that breaks the security of the one-way puzzle. That is there exists a non-negligible function $\epsilon(\cdot)$ such that for all $n\in \bbN$
        \[
            \Prr_{(k,s) \leftarrow \Samp(1^n)}[\Ver(\cA(s), s) = \top] \geq \epsilon(n)
        \]
        Rewriting using the definitions of $\Samp$ and $\Ver$ we get
        \[
            \Prr_{\substack{k\leftarrow\KeyGen(1^n)\\m\sample\bin^n}}[k^*\leftarrow \cA(m, \Enc_k(m)) \wedge m = \Dec_{k^*}(\Enc_{k}(m))] \geq p(n)
        \]
        We now build an algorithm $\cB$ that breaks the hiding of the encryption scheme. $\cB$ interacts with an external CPA challenger as follows.
        \begin{enumerate}
            \item Sample $m_0, m_1 \sample \bin^n$.
            \item Send $m_0, m_1$ to the external challenger.
            \item Receive $c = \Enc_k(m_b)$ for some freshly sampled $k$ and uniform bit $b$.
            \item Compute $k^* \leftarrow \cA(m_0, c)$
            \item If $m_0 = \Dec_{k^*}(c)$ then return $0$, else return a uniform bit.
        \end{enumerate}
        We calculate the probability that $\cB$ guess the uniform bit $b$.
        \begin{align*}
            \Prr[\cB \text{ returns } b]
            &= \Prr[b = 0 \wedge m_0 = \Dec_{k^*}(c)] + \frac{1}{2}\cdot \Prr[b = 0 \wedge m_0 \neq \Dec_{k^*}(c)] \\
            &+ \frac{1}{2}\cdot\Prr[b = 1 \wedge m_0 \neq \Dec_{k^*}(c)]\\
            =& \frac{1}{2}\cdot \Prr[m_0 = \Dec_{k^*}(\Enc_k(m_0))] + \frac{1}{4}\cdot \Prr[m_0 \neq \Dec_{k^*}(\Enc_k(m_0))] \\
            &+\frac{1}{4}\cdot\Prr[m_0 \neq \Dec_{k^*}(\Enc_k(m_1))]\\
            =& \frac{1}{4} + \frac{1}{4}\cdot \Prr[m_0 = \Dec_{k^*}(\Enc_k(m_0))] +\frac{1}{4}\cdot\Prr[m_0 \neq \Dec_{k^*}(\Enc_k(m_1))]\\
            \geq& \frac{1}{4} + \frac{\epsilon(n)}{4} +\frac{1}{4}\cdot\Prr[m_0 \neq \Dec_{k^*}(\Enc_k(m_1))]
        \end{align*}
        
        For arbitrary $z$, define $\bbM_z$ as follows:
        \[
        \bbM_z := \{m\in \bin^n \colon \exists k^* \in \Supp(\KeyGen(1^n)) \text{ s.t } \Dec_{k^*}(z) = m\}
        \]
        Since $\Dec$ is deterministic, for all $z$
        \[
            |\bbM_z| \leq |\Supp(\KeyGen(1^n))|
        \]
        Then 
        \[
            \Prr_{\substack{m_0, m_1 \leftarrow \bin^n\\k\leftarrow\KeyGen(1^n)}}[m_0 \in \bbM_{\Enc_k(m_1)}] \leq \frac{|\bbM_{\Enc_k(m_1)}|}{2^n} \leq \frac{|\Supp(\KeyGen(1^n))|}{2^n} = \negl(n)
        \]
        which by definition of $\bbM$ means that
        \[
            \Prr_{\substack{m_0, m_1 \leftarrow \bin^n\\k\leftarrow\KeyGen(1^n)}}[\exists k^*\text{ s.t } m_0 = \Dec_{k^*}(\Enc_k(m_1))] = \negl(n)
        \]
        which means that in the hiding game
        \[
            \Prr[m_0 \neq \Dec_{k^*}(\Enc_k(m_1))] = 1-\negl(n)
        \]
        Substituting this back in the expression for $\cB$ guessing $b$
        \begin{align*}
            \Prr[\cB \text{ returns } b] &\geq \frac{1}{4} + \frac{\epsilon(n)}{4} +\frac{1}{4}\cdot\Prr[m_0 \neq \Dec_{k^*}(\Enc_k(m_1))]\\
             &\geq \frac{1}{4} + \frac{\epsilon(n)}{4} +\frac{1}{4}\cdot(1-\negl(n))\\
             &\geq \frac{1}{2} + \frac{\epsilon(n)}{4} -{\negl(n)}
        \end{align*}
        which contradicts the hiding of the symmetric encryption scheme.
    \end{proof}
    The theorem follows immediately from the claim.
    \end{proof}

%% file: fullproof.tex
\section{Proving Lemma \ref{lem:entropyGapCore}}
\label{app:fullproof}

In this section, we provide a complete proof of Lemma \ref{lem:entropyGapCore}.
First, we list some preliminary theorems/claims that will be useful in proving the lemma.

\subsection{Preliminary Claims}
The following well-known theorem states that the boolean inner product function is a good extractor. This can be derived as a simple consequence of the leftover hash lemma.
\begin{theorem}[Inner Product is a Good Extractor]
\label{thm:LHLforIP} Let $X$ be a random variable distributed over $\bin^n$ where $\Hs_{\min}(X) \geq k$, for some $k > 0$. Then
\[
    \SD \Big( (U_n, \langle X, U_n \rangle ), (U_n, U_1) \Big) \leq 2^{(1-k)/2}
\]    
\end{theorem}
\begin{proof} 
    The hash function that takes a uniform seed $h \leftarrow \{0,1\}^{n+1}$ and $x \in \{0,1\}^{n}$ and outputs $\langle x, y \rangle \xor z$ is known to be universal, where $y$ indicates the first $n$ bits of $h$ and $z$ is the last bit of $h$. 
    Then by Theorem \ref{thm:LHLforSmoothEntropy}
    \[
        \SD \Big( (U_n, U_1, \langle X, U_n \rangle \xor U_1 ), (U_n, U_2) \Big) \leq \epsilon
    \] 
    for $\epsilon \leq 2^{(1-k)/2}$. We now consider the distribution obtained by  replacing the last bit with the XOR of the last two bits in both distributions.
    Since performing the same operation on two distributions can only reduce the statistical distance,
    \[
        \SD \Big( (U_n, U_1, \langle X, U_n \rangle), (U_n, U_2) \Big) \leq \epsilon
    \]
    which implies that
    $
    \SD \Big( (U_n, \langle X, U_n \rangle), (U_n, U_1) \Big) \leq \epsilon$, as desired.
\end{proof}

We will also use the following claim that relates the bias in an arbitrary single-bit distribution to its Shannon entropy.
        \begin{claim}
        \label{clm:biasedCoin}
            Let $X$ be an arbitrary $d$-biased distribution on a single bit, i.e., $$\Prr[X = 1] = \frac{1+d}{2} \text{ and } \Prr[X = 0] = \frac{1-d}{2}.$$
            Then,
            \[
            \Hs(X)\leq 1-\frac{d^2}{2}
            \]
         Additionally, whenever $d\leq 1/2$, we have 
         \[
            \Hs(X)\geq 1-d^2
            \]
        \end{claim}
        \begin{proof}
            Note that
            \begin{align*}
                \Hs(X) &= \left(\frac{1+d}{2}\right)\cdot \log\left(\frac{2}{1+d}\right) + \left(\frac{1-d}{2}\right)\cdot \log\left(\frac{2}{1-d}\right)\\
                &= 1 - \left(1/2\right)\cdot\left[(1+d)\log(1+d) + (1-d)\log(1-d)\right]
            \end{align*}
            Rearranging
            \begin{align*}
                1-\Hs(X) &=\left(1/2\right)\cdot\left[\log(1-d^2) + d\cdot(\log(1+d) - \log(1-d))\right]
            \end{align*}
            Since $\log y = \ln y/\ln 2$, we can rewrite as
            \begin{align*}
                1-\Hs(X) &=\left(1/2\ln 2\right)\cdot\left[\ln(1-d^2) + d\cdot(\ln(1+d) - \ln(1-d))\right]
            \end{align*}
            By the Taylor series expansion, for $y > -1$
            \[\ln(1+y) = y - y^2/2 + y^3/3 - y^4/4 + \ldots\]
            Noting that $0<d<1$ implies that $d, -d$ and $-d^2$ are all greater than $-1$,
 we can apply this expansion to the $\ln$ terms in the previous equation.
 This gives
            \begin{align*}
                1-\Hs(X) &=\left(1/2\ln2\right)\cdot
                \Big[(-d^2 - d^4/2 - d^6/3 - d^8/4\ldots)  \Big.\\ 
                &\Big. + d \cdot \big( (d - d^2/2 + d^3/3 - d^4/4\ldots) - (-d -d^2/2 - d^3/3 - d^4/4\ldots) \big) \Big]\\
                & = \left(1/2\ln2\right)\cdot
                \Big[(-d^2 - d^4/2 - d^6/3 - d^8/4 \ldots)  + d \cdot  (2d + 2d^3/3 + 2d^5/5 \ldots) \Big]\\
                & =\left(1/2\ln2\right)\cdot
                \Big[(d^2 + d^4/6 + d^6/15 + \ldots + d^{2n}/(n \cdot (2n - 1)) + \ldots) \Big]
            \end{align*}
            We may pull out the common factor $d^2$ to obtain
            \begin{equation}
                \label{eq:series}
                1-\Hs(X)=\left(d^2/2\ln2\right)\cdot
                \Big[( 1+ d^2/6 + d^4/15 + \ldots + d^{2n-2}/(n \cdot (2n - 1)) + \ldots) \Big]
            \end{equation}
            which implies
            \begin{align*}
                1 - \Hs(X) &\geq\left(d^2/2\ln2\right)
            \end{align*}
            or equivalently,
            \begin{align*}
                \Hs(X) &\leq1-\frac{d^2}{2\ln2}\\
                    &\leq1-\frac{d^2}{2}
            \end{align*}
           Additionally, equation \eqref{eq:series} implies
           \begin{align*}
               1 - \Hs(X) 
&\leq\left(d^2/2\ln2\right)\cdot
                \Big[( 1+ d^2 + d^4+ \ldots + d^{2n} + \ldots) \Big]\\
                &= \frac{d^2}{2\ln2\cdot (1-d^2)}
            \end{align*}
            If $d\leq 1/2$, then
            \begin{align*}
                1 - \Hs(X) &\leq\frac{d^2}{2\ln2\cdot (1-\frac{1}{4})}\\
                &=\frac{2d^2}{3\ln2}
            \end{align*}
            which implies
            \begin{align*}
                \Hs(X) &\geq1-\frac{2d^2}{3\ln2} \geq1-d^2
            \end{align*}
            which concludes the proof.
        \end{proof}

\subsection{Proof of the Lemma}
In this section, we proceed to prove the main lemma.\\
     
\noindent {\bf Lemma \ref{lem:entropyGapCore} Restated.}
            Let $X = \{X_n\}_{n \in \bbN}$ be an ensemble of distributions over $\bin^n$ s.t. for sufficiently large $n \in \bbN$, there exists $x^*_n\in \bin^n$ satisfying:
            \begin{enumerate}
                \item $\Pr_{x\leftarrow X_n}[x=x^*_n] \geq 1/6n$
                \item $\forall x' \neq x^*_n$, $\Pr_{x \leftarrow X_n}[x=x']\leq 2/n^{600}$ 
            \end{enumerate}
            Let $R_n$ be uniformly distributed on $\bin^n$. Define $\alpha_{0,n}$ and $\alpha_{1,n}$ as follows.
            \begin{gather*}
                \alpha_{0,n}(x,r) := \langle x, r \rangle\\
                \alpha_{1,n}(x,r) := \left\{
                \begin{array}{cl}
                    U_1 &\text{ if $x = x^*_n$}\\
                    \langle x, r \rangle &\text{ otherwise}
                \end{array}
                \right.
            \end{gather*}
            For $b\in \bin$, define distribution $A_{b,n} := R_n, \alpha_{b,n}(X,R)$. Then for sufficiently large $n \in \bbN$,
            \[
            \Hs(A_{1,n}) - \Hs(A_{0,n}) \geq 1/100n^2
            \]
        To prove this lemma, we will first define random variables that indicate biases in the two distributions.
        Fix any sufficiently large $n$, and any $r\in \bin^n, b\in \bin$.
        Define 
        \begin{align*}
            B^b_{r,n} &:= | \Prr_{x\leftarrow X_n}[\alpha_{b,n}(x,r) =1] - \Prr_{x\leftarrow X_n}[\alpha_{b,n}(x,r) =0] |\\
            B^*_{r,n} &:= | \Prr_{x\leftarrow X_n}[\langle x,r \rangle =1 | x\neq x^*_n] - \Prr_{x\leftarrow X_n}[\langle x,r \rangle =0| x\neq x^*_n] |
        \end{align*}
        Intuitively, the variable $B^b_{r,n}$ represents the bias of $\alpha_{b,n}(X_n,r)$ away from uniform. $B^*_{r,n}$ represents the bias conditioned on $X_n \neq x^*_n$ (which is identical for both distributions).
        Then,
        \begin{align}
        B^0_{r,n} &= \left|\Prr_{X_n}[x^*_n] + \Prr_{x\leftarrow X_n}[x\neq x^*_n]\left(\Prr_{x\leftarrow X_n}[\langle x,r \rangle = \langle x^*_n,r \rangle | x \neq x^*_n] - \Prr_{x\leftarrow X_n}[\langle x,r \rangle \neq \langle x^*_n,r \rangle | x \neq x^*_n]\right)\right| \nonumber \\
         &\geq \Prr_{X_n}[x^*_n] - \Prr_{x\leftarrow X_n}[x\neq x^*_n]\left|\left(\Prr_{x\leftarrow X_n}[\langle x,r \rangle = \langle x^*_n,r \rangle | x \neq x^*_n] - \Prr_{x\leftarrow X_n}[\langle x,r \rangle \neq \langle x^*_n,r \rangle | x \neq x^*_n]\right)\right| \nonumber \\
        &\geq \Prr_{X_n}[x^*_n] - \Prr_{x\leftarrow X_n}[x\neq x^*_n]\cdot B^*_{r,n}
        \label{eq:b1}
        \end{align}
        where the second inequality follows due to a triangle inequality.
        Furthermore, 
        \begin{align}
        B^1_{r,n} &= \left|\Prr_{x\leftarrow X_n}[x\neq x^*_n]\left(\Prr_{x\leftarrow X_n}[\langle x,r \rangle = \langle x^*_n,r \rangle | x \neq x^*_n] - \Prr_{x\leftarrow X_n}[\langle x,r \rangle \neq \langle x^*_n,r \rangle | x \neq x^*_n]\right)\right| \nonumber \\
        &= \Prr_{x\leftarrow X_n}[x\neq x^*_n] \cdot B^*_{r,n}
        \label{eq:b2}
        \end{align}

The proof of the lemma follows immediately from the next two claims, which analyse the distribution $X$ in different ways depending on whether or not it is heavily biased towards $x^*_n$. 

 \begin{claim} 
 Consider any $n \in \bbN$ for which $X_n$ is heavily biased towards $x^*_n$, that is, 
 $$\Pr_{x\leftarrow X_n}[x=x^*_n] > 1- 1/n.$$
 Then, if $n$ is large enough, it holds that for all $r\in \bin^n$, $$\Hs(\alpha_1(X,r)) - \Hs(\alpha_0(X,r)) \geq 1/100n^2.$$
 \end{claim}
\begin{proof}
     Since $B^*_{r,n} \leq 1$ for every $n, r \in \bbN$, equations (\ref{eq:b1}) and (\ref{eq:b2}) imply
    $$B_{r,n}^0 \geq \Prr_{X_n}[x^*_n] - \Prr_{x\leftarrow X_n}[x\neq x^*_n]\cdot B^*_{r,n} \geq 1 - 2/n$$ 
    and 
    $$B_{r,n}^1 = \Prr_{x\leftarrow X_n}[x\neq x^*_n] \cdot B^*_{r,n} \leq 1/n.$$
    Therefore, by Claim \ref{clm:biasedCoin} applied to $\alpha_0(X,r)$ and  $\alpha_0(X,r)$, for large enough $n$
        \begin{gather*}
            \Hs(\alpha_0(X,r)) \leq 1 - \frac{(1-2/n)^2}{2} \leq 2/3 \\
            \Hs(\alpha_1(X,r)) \geq 1 - 1/n^2
        \end{gather*}
        Thus for all $r\in \bin^n$, 
        $$\Hs(\alpha_1(X,r)) - \Hs(\alpha_0(X,r)) \geq 1/3 - 1/n^2 \geq 1/100n^2$$
        as desired.
    \end{proof}
    \begin{claim}
    Consider any $n \in \bbN$ for which 
    $$\Pr_{x\leftarrow X_n}[x=x^*_n] \leq 1- 1/n.$$
    Then, if $n$ is large enough, it holds that for all $r\in \bin^n$, $$\Hs(A_{1,n}) - \Hs(A_{0,n}) \geq 1/100n^2.$$
    \end{claim}
    \begin{proof}
    The premise of the claim implies
    $$\Pr_{x\leftarrow X_n}[x\neq x^*_n] \geq 1/n.$$
        Let $X'_n$ denote the distribution $(X_n|X_n\neq x^*_n)$.
        For all $x'\in \Supp(X'_n)$, 
        \begin{align*}
        \Pr_{x\leftarrow X'_n}[x=x'] &= \frac{\Prr_{x\leftarrow X_n}[x = x']}{\Prr_{x\leftarrow X_n}[x \neq x^*_n]}\\
        &\leq n\cdot \Prr_{x\leftarrow X_n}[x = x']\\
        &\leq 2{n^{-599}}
        \end{align*}
        which implies that for all $x'\in \Supp(X'_n)$,
        \begin{align*}
            \Hs_{X'_n}(x') \geq 599\log n - \log2
            \geq 598\log n \text{~~~~(for $n \geq 2$)}
        \end{align*}
        which implies that  $\Hs_{\min}(X'_n) \geq 598\log n$.
        Then, Theorem \ref{thm:LHLforIP} applied to $X'_n$ implies that:
        \[
            \mathsf{SD} \Big((R_n, \langle X'_n, R_n\rangle), (R_n, U_1) \Big) \leq \sqrt{2{n^{-598}}} < {n^{-200}}
        \] 
         The statistical distance may be rewritten as
         \[
         \mathsf{SD} \Big((R_n, \langle X'_n, R_n\rangle), (R_n, U_1) \Big) =
         \sum_{r} \frac{1}{2}\cdot\Prr_{R_n}[r]\cdot\left|\Prr_{x\leftarrow X'_n}[\langle x,r\rangle = 1] - 1/2\right| \leq {n^{-200}}
         \]
         Since $\Pr_{x\leftarrow X'_n}[\langle x,r\rangle = 1] = 1 - \Pr_{x\leftarrow X'_n}[\langle x,r\rangle = 0]$, we can rewrite the above equation as
         \[
         \sum_{r} \frac{1}{4}\cdot\Prr_{R_n}[r]\cdot\left|\Prr_{x\leftarrow X'_n}[\langle x,r\rangle = 1] - \Prr_{x\leftarrow X'_n}[\langle x,r\rangle = 0] \right| \leq n^{-200}
         \]
         Substituting $B^*_{r,n} = \left|\Prr_{x\leftarrow X'_n}[\langle x,r\rangle = 1] - \Prr_{x\leftarrow X'_n}[\langle x,r\rangle = 0] \right|$, we obtain
         \[
         \sum_{r} \frac{1}{4}\cdot\Prr_{R_n}[r]\cdot B^*_{r,n} \leq \n^{-200}
        \]
        By a Markov argument on $r$,
        \[
        \Prr_{r\leftarrow R_n}[B^*_{r,n} \geq n^{-100}] \leq n^{-100}
        \]
        Let $\mathbb{B}_n:=\{r| B^*_{r,n} \leq n^{-100}\}$. Therefore, 
        $
        \Prr_{r\leftarrow R_n}[r \notin \mathbb{B}_n] \leq n^{-100}
        $.
        Moreover, equations (\ref{eq:b1}) and (\ref{eq:b2}) imply that for all $r \in \mathbb{B}$,
        $$B_{r,n}^0 \geq \big( \Prr_{X_n}[x^*_n] - \Prr_{x\leftarrow X_n}[x\neq x^*_n]\cdot B^*_{r,n} \big)  \geq 1/6n - n^{-100} \geq 1/7n \text{, and}
        $$
        \[
        B_{r,n}^1 = \Prr_{x\leftarrow X_n}[x\neq x^*_n] \cdot B^*_{r,n} \leq n^{-100}
        \]
        Finally, by applying the chain rule,
        \begin{align*}
            \Hs(A_{1,n}) - \Hs(A_{0,n}) &= \Hs(\alpha_{1,n}(X_n,R_n),R_n) - \Hs(\alpha_{0,n}(X_n,R_n),R_n)\\
            &= \Hs(\alpha_{1,n}(X_n,R_n)|R_n) + \Hs(R_n) - \Hs(\alpha_{0,n}(X_n,R_n)|R_n) - \Hs(R_n)\\
            &= \Hs(\alpha_{1,n}(X_n,R_n)|R_n) - \Hs(\alpha_{0,n}(X_n,R_n)|R_n)\\
            &=\sum_r \Prr_{R_n}[r]\cdot \left(\Hs(\alpha_{1,n}(X_n,r)- \Hs(\alpha_{0,n}(X_n,r)\right)\\
            &=\sum_{r\in \mathbb{B}_n} \Prr_{R_n}[r]\cdot \left(\Hs(\alpha_{1,n}(X_n,r)- \Hs(\alpha_{0,n}(X_n,r)\right) \\&+ \sum_{r\notin \mathbb{B}_n} \Prr_{R_n}[r]\cdot \left(\Hs(\alpha_{1,n}(X_n,r)- \Hs(\alpha_{0,n}(X_n,r)\right)\\
            &\geq \sum_{r\in \mathbb{B}_n} \Prr_{R_n}[r]\cdot \left(\Hs(\alpha_{1,n}(X_n,r)- \Hs(\alpha_{0,n}(X_n,r)\right) - \sum_{r\notin \mathbb{B}_n} \Prr_{R_n}[r]\\
            &\geq \sum_{r\in \mathbb{B}_n} \Prr_{R_n}[r]\cdot \left(\Hs(\alpha_{1,n}(X_n,r)- \Hs(\alpha_{0,n}(X_n,r)\right) - \Prr_{r\leftarrow R_n}[r \notin \mathbb{B}_n]
        \end{align*}
        Applying Claim \ref{clm:biasedCoin} to the distributions $\alpha_{0,n}(X_n,r)$ and $\alpha_{1,n}(X_n,r)$ with biases $B^0_{r,n}$ and $B^1_{r,n}$ respectively, and noting that $B^1_{r,n} < \frac{1}{2}$ for $n \geq 2$, 
        \begin{align*}
            \Hs(A_{1,n}) - \Hs(A_{0,n}) &\geq \sum_{r\in \mathbb{B}_n} \Prr_R[r]\cdot \left(\frac{\left(B^0_{r,n}\right)^2}{2}- \left(B^1_{r,n}\right)^2\right) - \Prr_{r\leftarrow R}[r \notin \mathbb{B}_n]\\
            &\geq \sum_{r\in \mathbb{B}_n} \Prr_R[r]\cdot \left(\frac{1}{98n^2}- \frac{1}{n^{100}}\right) - \Prr_{r\leftarrow R}[r \notin \mathbb{B}_n]\\
            &\geq \Prr_{r\leftarrow R}[r\in \mathbb{B}_n] \cdot \left(\frac{1}{99n^2}\right) - \Prr_{r\leftarrow R}[r \notin \mathbb{B}_n]\\
            &\geq (1-1/n^{100}) \cdot \left(\frac{1}{99n^2}\right) - (1/n^{100})\\
            &\geq 1/100n^2 \text{~~~~~~(for large enough $n$)}
        \end{align*}
        This completes the proof of the claim.
        \end{proof}